\documentclass[11pt]{article}

\usepackage[margin=1in]{geometry}

\usepackage{rotating}
\usepackage{makecell}

\usepackage{natbib}

\usepackage{booktabs}
\usepackage{multirow}
\usepackage{tcolorbox}

\newcolumntype{C}[1]{>{\centering\arraybackslash}p{#1}}

\usepackage[multiple]{footmisc}

\usepackage{float}
\usepackage{comment}
\usepackage{soul}
\usepackage{wrapfig}

\usepackage{amsmath,amssymb,amsthm,mathrsfs,amsfonts,dsfont, mathtools}

\usepackage{graphicx}
\usepackage{tikz}
\usetikzlibrary{calc}
\usetikzlibrary{shapes.geometric}
\usepackage{paralist}

\usepackage{wrapfig}

\usepackage{float}
\usepackage{comment}

\usepackage{thmtools}
\usepackage{thm-restate}

\usepackage{caption}
\usepackage[labelformat=simple]{subcaption}

\usepackage{aliascnt}

\usepackage{xcolor}

\usepackage[backref=page]{hyperref}
\hypersetup{colorlinks=true,citecolor=blue,linkcolor=blue}
\hypersetup{colorlinks=true}
\hypersetup{linkcolor=[rgb]{.7,0,0}}
\hypersetup{citecolor=[rgb]{0,.7,0}}
\hypersetup{urlcolor=[rgb]{.7,0,.7}}

\usepackage{cleveref}

\declaretheorem[numberwithin=section]{theorem}
\declaretheorem[sibling=theorem, style=definition]{definition}
\declaretheorem[sibling=theorem]{lemma}

\declaretheorem[sibling=theorem]{corollary}
\declaretheorem[sibling=theorem, style=definition]{remark}
\declaretheorem[sibling=theorem]{proposition}

\newcommand{\R}{\mathbb{R}}

\newcommand{\N}{\mathbb{N}}

\newcommand{\1}[1]{\mathds{1}\left[{#1}\right]}

\newcommand{\type}{\mathcal{T}}

\makeatletter
\newcommand\Ps@textstyle[2]{\mathbb{P}_{#1}\left[{#2}\right]}
\newcommand\Es@textstyle[2]{\mathbb{E}_{#1}\left[{#2}\right]}
\newcommand\Ps[2]{%
  \mathchoice %
  {\underset{{#1}}{\mathbb{P}}\left[{#2}\right]}
  {\Ps@textstyle{#1}{#2}}
  {\Ps@textstyle{#1}{#2}}
  {\Ps@textstyle{#1}{#2}}
}
\newcommand\Es[2]{%
  \mathchoice %
  {\underset{{#1}}{\mathbb{E}}\left[{#2}\right]}
  {\Es@textstyle{#1}{#2}}{\Es@textstyle{#1}{#2}}{\Es@textstyle{#1}{#2}}
}
\makeatother

\hypersetup{pdfauthor={Aadityan Ganesh, Clayton Thomas, and S. Matthew Weinberg}}

\usepackage{comment}

\usepackage{sidecap}

\usepackage{csquotes}
\usepackage{colortbl}

\usepackage{tikz}
\usetikzlibrary{automata, positioning, calc, fit, shapes.misc}
\usetikzlibrary{matrix}
\usetikzlibrary{positioning}
\usetikzlibrary{arrows, decorations.pathmorphing}
\usetikzlibrary{decorations.pathreplacing,calligraphy}

\newcommand{\T}{\mathcal{T}}

\newcommand{\FeasAlloc}{\mathcal{X}}
\newcommand{\Outcomes}{\mathcal{Y}}
\newcommand{\TypeDistr}{\mathcal{T}}

\newcommand{\rev}{\mathsf{Rev}}

\usepackage[shortlabels]{enumitem}

\usepackage{multirow}

\usepackage{pifont}%
\newcommand{\cmark}{\ding{51}}%
\newcommand{\xmark}{\ding{55}}%

\begin{document}

\title{
Revisiting the Primitives of Transaction Fee Mechanism Design\thanks{For helpful discussion and feedback, we thank Jacob Leshno, seminar and talk participants at CMU, Princeton, Microsoft Research, and EC 2024, and the three anonymous reviewers for this paper at EC 2024.}
}
\author{
  Aadityan Ganesh\thanks{Princeton University | \emph{E-mail}: \href{mailto:}{aadityanganesh@princeton.edu}.}
  \and
  Clayton Thomas\thanks{Microsoft Research | \emph{E-mail}: \href{mailto:}{thomas.clay95@gmail.com}.}
  \and 
  S. Matthew Weinberg\thanks{Princeton University | \emph{E-mail}: \href{mailto:}{smweinberg@princeton.edu}. Supported by NSF CAREER Award CCF-1942497.}
}
\date{}

\begin{titlepage}
\maketitle
\begin{abstract}
  
Transaction Fee Mechanism Design studies auctions run by untrusted miners for transaction inclusion in a blockchain. 
Under previously-considered desiderata, an auction is considered `good' if, informally-speaking, each party (i.e., the miner, the users, and coalitions of both miners and users) has no incentive to deviate from the fixed and pre-determined protocol.
In other words, prior works posit that a `good' auction should be `simple for users', `simple for miners', and `resistant to collusion'.

In this paper, we propose a novel desideratum for Transaction Fee Mechanisms.
We say that a TFM is \emph{off-chain influence proof} when the miner cannot achieve additional revenue by running a separate auction off-chain.
While the previously-highlighted EIP-1559
is the gold-standard according to prior desiderata (satisfying simplicity for users and miners along with collusion resistance in the `typical' case where supply is functionally unlimited),
we show that it \emph{does not} satisfy off-chain influence proofness.
Intuitively, this holds because a Bayesian revenue-maximizing miner can strictly increase profits by persuasively threatening to censor any bids that do not transfer a tip directly to the miner off-chain.

On the other hand, we reconsider the Cryptographic (multi-party computation assisted) Second Price Auction mechanism \citep{ShiCW23}, which is technically not `simple for miners' according to previous desiderata (since miners may wish to set a reserve by fabricating bids).
We argue that the space of TFMs can be naturally expanded to solicit an input from the miner, for example by asking them to set the reserve price of the auction. We show that in this model, the Cryptographic Second Price Auction satisfies simplicity for users and miners, and off-chain influence proofness, since it allows a Bayesian miner to maximize their revenue by posting an optimal reserve price. 

Finally, we prove a strong impossibility result: no mechanism satisfies all previously-considered properties along with off-chain influence proofness, even with unlimited supply, and even after soliciting miner input.

\end{abstract}
\thispagestyle{empty}
\end{titlepage}

\begin{titlepage}

\maketitle

\end{titlepage}

\section{Introduction}
\label{sec:intro}

Transaction fee mechanisms (henceforth, TFMs) are auctions used by blockchains to decide which transactions are included in each block, along with a price (``transaction fee'') charged to the user submitting the transaction.
TFMs are an integral part of blockchain infrastructure, and are hotly debated.
In one high-profile example, the Ethereum blockchain switched from the status-quo first-price auction (which is still in use in Bitcoin) to a mechanism known as EIP-1559.
In essence, EIP-1559 is a posted-price mechanism with a variable price exogenously set by market activity, in which all user-paid transaction fees are discarded (``burned'') and the miner\footnote{
    Following \cite{Roughgarden21} and the common terminology used in proof-of-work blockchains, we use the language of ``miner'', even though block creators are now ``stakers'' in Ethereum since the switch to proof-of-stake.
    }
who constructs the block receives only a positive constant revenue (``block reward'').

The redesign of Ethereum's TFM was meant to simplify bidding behavior for users (who may be small-scale retail users valuing ease-of-participation), while avoiding placing trust in the miner running the TFM algorithm (who may be a strategically sophisticated, untrusted, profit-maximizing agent). 
Balancing these two objectives is a challenging task.
For example, in the pre-existing First-Price Auction, users were frustrated that they had to bid strategically,\footnote{
    Anecdotally, ``autobidders'' like MetaMask were not particularly good at optimal bidding either, due to volatility in demand.
}
while under a Second-Price Auction, we would expect the pseudonymous miner to submit a shill bid equal to (just under) the highest submitted bid.
In contrast, a series of works (discussed in detail below) has shown that EIP-1559 provides formal guarantees on agents' incentives: 
participants (users, miners, or even coalitions containing both the miner and some set of users) cannot increase their utility by manipulating their inputs to the protocol.

Consider, however, the following miner strategy in EIP-1559.
The miner publicly announces \emph{off-chain} that they will censor all bids that do not directly transfer the miner an ``entry fee'' $\gamma$ via some off-chain payment platform.
If \emph{any} user transfers the miner this $\gamma$ in order to get their transaction considered, then the miner will strictly profit from this attack.
Thus, despite the fact that the miner cannot profit in EIP-1559 by deviating from the protocol while remaining entirely on-chain, the miner \emph{can} profit by making convincing threats and collecting payments off-chain.

In this paper, we take ``off-chain'' attacks of the above form seriously,
and propose a model in which the miner has the ability to conduct \emph{any off-chain mechanism whatsoever} in order to determine the inputs of the TFM algorithm. 
Using this model, we formalize and highlight the above novel attack on EIP-1559.
Alongside this, we also make the natural suggestion that the miner be \emph{allowed} to strategize on-chain, so long as this does not harm the incentives of the users.
We consider a cryptographic implementation of a second-price auction in which the miner strategizes on-chain by choosing the reserve price.
We show that (while adopting our new desiderata involves a provably-unavoidable trade-off with prior definitions of resistance to user-miner collusion) the Cryptographic Second-Price Auction achieves good incentive properties both on-chain and off-chain under our revised notions.

\paragraph{Prior frameworks on TFM design.} 
Before giving more exposition into our framework, we review past approaches to the study of transaction fee mechanism design.
Drawing on classical auction theory, 
\cite{Roughgarden21} initiated the modern study of TFM design and proposed three elegant desiderata that a transaction fee mechanism should satisfy, with a rapidly developing agenda around them~\citep{ChungS23, WuSC24, ChungRS24, GafniY24}.
At a high level, these desiderata are:
\begin{itemize}
    \item \textbf{Simple for Users.} Users who know their value $v$ for their transaction being included in a block should optimize their expected utility by bidding $v$. This is called User Incentive Compatibility (UIC) in~\cite{Roughgarden21}, and is equivalent to familiar notions of (Dominant Strategy) Incentive Compatibility.
    
    \item \textbf{Simple for Miners.} Miners should maximize their revenue by executing the intended auction, rather than deviating. This is called Myopic Miner Incentive Compatibility (MMIC) in~\cite{Roughgarden21} (which also has conceptual connections to the notion of credibility \citep{AkbarpourL20} in classical auction settings).
    
    \item \textbf{Resistant to Collusion.} No miner together with a collection of users should profit by first writing a binding side-contract and then altering their behavior in the auction.
    \cite{Roughgarden21} terms an auction Off-Chain Agreement-Proof (OCA-proof) if the miner together with all users cannot jointly profitably deviate, and \cite{ChungS23} term an auction $c$-Side Contract Proof ($c$-SCP) if the miner together with a set of any $c$ users cannot jointly profitably deviate.
\end{itemize}

These desiderata have proven very useful for describing the properties of different TFMs.
However, in our paper, we argue that as previously studied, these desiderata fail to capture important strategic aspects of two notable mechanisms---EIP-1559, and the Cryptographic Second Price Auction.

\paragraph{EIP-1559, and Off-Chain Influence Proofness.}
We now discuss the first of our two main motivating mechanisms, EIP-1559.
The ideal version of EIP-1559 is an unlimited supply posted-price mechanism with an exogenous price, where all revenue is burned. That is, any bid exceeding an exogenously set price $p$ is included and pays $p$, but those payments are burned and the miner always receives zero revenue.\footnote{
    Of course, there is not actually unlimited supply of blockspace in Ethereum. EIP-1559 behaves like this in the event that the exogenous price $p$ is sufficiently high so that supply exceeds demand. In the event that demand exceeds supply, EIP-1559 devolves into a first-price auction with reserve $p$, except that $p$ of each bidder's payment is burned and only the excess `tip' is transferred to the miner.
} 
Prior desiderata consider this ideal (i.e., unlimited supply) EIP-1559 as ``the dream TFM''---\cite{Roughgarden21} establishes that it is UIC, MMIC, and OCA-proof.\footnote{
    And moreover, with limited supply, \citet{ChungS23} show no mechanism satisfies all three of UIC, MMIC, and $1$-SCP.
}

However, observe that a miner receives zero revenue under the ideal EIP-1559,\footnote{
    This has been observed several times, and motivates follow-up work such as~\cite{WuSC24} which we discuss in \autoref{sec:related}.
} and MMIC essentially just observes that censoring bids or including fake bids cannot change that (in fact, including fake bids would even harm the miner, as now the miner has to pay the burned bids). 
But consider the miner behavior discussed above: the miner publicly announces \emph{off-chain} that they intend to censor all bids that don't pay a fee of $\gamma$.
This transfer could take place off-chain, or in the actual implementation of EIP-1559, the miner could demand this transfer \emph{as part of the included transaction tip} (a part of the EIP-1559 protocol designed to handle periods of high demand).
This strategy dominates honestly implementing EIP-1559---the worst that can happen is that all bidders ignore the miner (in which case the miner still gets zero), but if \emph{any} user is willing to pay $\gamma$ on top of the exogenous posted price, then the miner gets strictly positive utility.
This deviation is not in violation of MMIC---it is an \emph{off-chain} deviation, and cannot be implemented simply by censoring transactions or including fake bids (the main types of deviations considered by prior desiderata).\footnote{
    Nor is it a violation of OCA-proofness/$c$-SCP -- this miner strategy does not cause the sum of miner and users' on-chain utilities to increase. It just adds a direct transfer from users to miners.
}
We therefore argue that MMIC is insufficient to claim that a transaction fee mechanism is simple for miners.

Inspired by the above, we propose to add a new desiderata, \emph{off-chain influence-proofness}, which informally states that off-chain behavior of the form above is not profitable, and the miner can always optimize her revenue by engaging directly with the on-chain mechanism (see \autoref{sec:Definitions} for a formal definition).

\paragraph{Cryptographic Second-Price Auction, and On-Chain Miner-Set Reserves.}
We now discuss our second main motivating mechanism, the Cryptographic Second-Price Auction (C2PA).
Motivated by the recent push towards heavyweight cryptography in blockchains, we consider a model which abstracts from the actual cryptographic technology used, in which users can submit encrypted bids, and the only role of the miner is again to decide which bids (including fake bids of their own) are given as input to a cryptographic protocol.\footnote{
    That is, in our model, the miner must decide which users' bids to forward to the TFM algorithm without \emph{any} information on the actual contents of the bid.
    Our model and results are thus agnostic to whether Multi-Party Computation (MPC), or Verifiable Delay Functions (VDFs), or some other cryptographic primitives are used in the implementation. 
    \cite{ShiCW23} consider a slightly different model termed MPC-assisted Second-Price Auction, in which the miner cannot even censor bids.
}
From the viewpoint of simplicity for users and miners, the C2PA intuitively seems perfect, with the only drawback being the use of heavyweight cryptography. 
However, it follows from \citet{ChungS23} that C2PA is, in fact, not MMIC. 
Specifically, C2PA does not have a built-in reserve, so technically the miner will always want to add a fake bid to act as a reserve.

While correct, we argue that to some extent, this classification seems to be missing the point. 
The blockchain community initially considered MMIC because we don't trust the miner, and we intuitively conclude that when the TFM is MMIC, we can assume the miner will ``actually'' implement the auction, and when it is additionally UIC, users can bid their true values without worry.
But in C2PA, it is still the case the each user should really go ahead and bid their actual value. In fact, due to our discussion above on off-chain influence proofness, this is perhaps even more true in C2PA than in EIP-1559!
We argue therefore that MMIC in prior models may be too restrictive, because (in mechanism design terminology) it requires the miner to implement exactly a prior-free mechanism, and does not give the miner any flexibility to augment the mechanism based on their belief about the distribution of user values, {even when this flexibility could never impact the users' incentives}.

We instead propose that transaction fee mechanisms {should allow for miner augmentation}.
Specifically, we propose that a transaction fee mechanism should allow for input from the miner (this could take the form of a reserve, but is not restricted to this). 
But importantly, it should hold that these inputs from the miner do not impact the users' incentives.
To accommodate this shift in the model, and to be coherent with our framework more broadly, we rephrase UIC to \emph{on-chain user simplicity}, and MMIC to \emph{on-chain miner simplicity} (see \autoref{sec:Definitions} for the formal definitions).

\paragraph{A Bayesian Formalization.} 
EIP-1559 and C2PA are our main motivating examples for revisiting the primitives of transaction fee mechanism design. 
While doing so, we also propose a formalization that imagines a revenue-maximizing miner in a Bayesian setting,\footnote{
    Throughout the paper, we restrict attention to the i.i.d. setting, where every user has the same distribution over values.
    This seems an especially apt assumption for blockchain settings, since users can freely create new identifier if they wish to avoid the types of price discrimination-like behavior that Myersonian optimal auctions require for non-i.i.d. bidders.
} which is more in line with classical auction design \citep{Myerson81} and also some recent work that considers untrusted auctioneers \citep{AkbarpourL20, FerreiraW20, EssaidiFW22}.
Specifically, we propose to consider a miner working within the independent private values model \citep{Myerson81}, and we examine the properties of different Bayes-Nash equilibria involving users.

Using our model, we formalize the simplicity desiderata and properties satisfied by EIP-1559 and C2PA discussed above. 
We also prove a strong impossibility theorem:
\emph{no} non-trivial auction can be simultaneously on-chain simple (for both users and miners), off-chain influence proof, and strongly collusion proof (a definition which essentially rephrases $1$-SCP in our framework).
This impossibility holds \emph{even in the unlimited supply setting}.
Remember that EIP-1559 is UIC, MMIC, and $c$-SCP for all $c\ge 1$ in the unlimited supply setting. 
Since C2PA abandons strong collusion proofness, but maintains all our other desirable properties, this proves a formal tradeoff: among on-chain simple TFMs, one must choose to adopt an off-chain influence proof auction (such as C2PA) or a strongly collusion proof auction (such as EIP-1559), but cannot satisfy both.

As an aside, we also take the opportunity to make some remarks on different notions of collusion resistance for TFMs. 
We observe that prior definitions (OCA proofness, $1$-SCP, and our strong collusion proofness) mostly seem to capture collusion \emph{between trusted entities who truly profit-share}.
These definitions may make sense in some instances (if, for example, Coinbase is a miner and also paying for its own transactions rather than passing on fees to customers).
But in other instances, this strong level of collusion may be unrealistic (if, for example, a prominent miner and a prominent dApp try to jointly collude). 
In such settings, we argue that the miner would likely establish a concrete contract of transfers to redistribute any joint surplus,
and moreover, \emph{colluding users would likely act strategically} to maximize their returns while taking this contract into consideration.
Thus, following our framework more broadly, the colluding users would play in a Bayes-Nash equilibrium.
Since users will still be playing in a Bayes-Nash equilibrium of the off-chain game in the end, this means that off-chain influence proof TFMs \emph{already} guarantee some form of such ``Bayesian'' collusion-resistance (with some caveats; see \autoref{sec:AlternateCollusion}).
We thus argue that such more permissive collusion-resistance definitions are best suited to capture an untrusted miner colluding with a separate user.
Still, as discussed above, strong collusion proofness better captures situations where miners genuinely have their own transactions, so we believe both criteria are worth considering.

\paragraph{Results and Paper Structure.} After giving preliminaries on (Bayesian) mechanism design in \autoref{sec:prelims}, our paper proceeds as follows:

\begin{itemize}
    \item \autoref{sec:model} introduces our model of TFMs, allowing for both on-chain manipulations as in previous papers, and allowing for the miner to conduct arbitrary off-chain auctions to determine behavior on-chain.
    
    \item \autoref{sec:Definitions} gives our proposed simplicity and incentive-compatibility definitions for TFMs, and compares our desiderata to pre-existing definitions.
    
    \item \autoref{sec:specific-cases} analyzes several commonly-studied auctions in TFM design (namely, EIP-1559, C2PA, and first-price auction) and classifies these auctions by our desiderata. %
    
    \item \autoref{sec:impossibility} proves our impossibility result: no nontrivial mechanism simultaneously satisfies on-chain (user and miner) simplicity, off-chain influence proofness, and strong collusion proofness.
\end{itemize}

In our Appendix, we investigate additional related simplicity notions (\autoref{sec:Zoo}), provide separations between all our main simplicity conditions (\autoref{sec:OtherEg}), and briefly explore cryptographic implementations of the miner-gatekeeper model (\autoref{sec:DRA}).

\subsection{Related work} \label{sec:related}

\paragraph{Transaction Fee Mechanism Design.} \cite{LaviSZ19} are the first to formally view Bitcoin's transaction fees as an auction, and propose alternatives (such as the $k^{th}$ price auction, or the monopolistic-price auction), and similar questions are also considered in~\citet{BasuEOS19} and \citet{Yao18}.~\cite{Buterin20}
propose EIP-1559 for Ethereum, the first significant change to the transaction fee mechanism of a major blockchain, which is analyzed rigorously in~\cite{Roughgarden20}.~\citet{Roughgarden20, Roughgarden21} kicks off the more recent study of transaction fee mechanisms by proposing the desiderata mentioned above and justifying Ethereum's EIP-1559 through this lens. 
In particular,~\cite{Roughgarden21} establishes that without a supply constraint (or when EIP-1559's exogenous price is high enough so that supply exceeds demand), EIP-1559 is UIC, MMIC, and OCA-proof.~\cite{ChungS23} prove a strong impossibility result, that no auction is simultaneously UIC and $c$-SCP with limited supply (for any $c$). 

Further follow-up work of~\cite{GafniY22, GafniY24, ChungRS24} investigate the distinction between OCA-proof and $c$-SCP, while ~\cite{GafniY22} and \citet{ChenSZZ24} also consider a relaxation of UIC to Bayesian-UIC (without miner-set advice as in our paper). 
\cite{ShiCW23} are the first to introduce cryptography into transaction fee mechanism design---they design non-trivial TFMs assuming Multi-Party Computation, but also observe that the C2PA is not MMIC.~\cite{WuSC24} design mechanisms (also leveraging cryptography) that achieve non-zero revenue in a Bayesian setting, under the assumption that some users are honest.

While some of these works indeed tweak the core TFM model of~\cite{Roughgarden21}, they do so in directions that are orthogonal to our revised axioms. For example,~\cite{GafniY22} propose to consider Bayesian-IC instead of DSIC for users.
\cite{ChungS23} propose a model where it is costly for the miner to include fake bids---this approach could be combined with our revised definitions (although we do not investigate this particular avenue ourselves).
To our knowledge, our work is the first proposal to significantly revise the core initial TFM desiderata, which we believe will provide additional insight to future works alongside this agenda (and in a way that is backwards-compatible with existing tweaks to the vanilla model of~\cite{Roughgarden21}). 

\paragraph{Variants of Transaction Fee Mechanism Design.} In line with all previously-cited works, we consider a one-shot Transaction Fee Mechanism played by users who want their transactions to be included in the present block %
and miners who are mypoic (i.e.~because they are a small player in a decentralized ecosystem). While TFMs ``in the wild'' are indeed a repeated game (both because users may be content to have their transaction in a later block, and becuase miners of this block may also be miners of the next block), this one-shot setting is the core upon which later research lies, and it is therefore crucial to nail the core setting. Some recent works have explored the repeated game aspect. For example,~\cite{FerreiraMPS21,LeonardosMRSP21, LeonardosRMP23} study the dynamics of how EIP-1559's exogenous price is updated over time, but do not study further strategic aspects that appear in the one-shot game.~\cite{GafniY22} consider a particular Bayesian model of user values arriving over time and propose a simple BIC and MMIC mechanism (essentially the revelation principle applied to a first-price auction). Our work has little overlap with these, as we prioritize a deeper understanding of the core single-block setting.
\cite{ChitraFK23} study implementations of credible (non TFM) auctions, using a secure blockchain as a primitive.

Since EIP-1559, an orthogonal paradigm shift has taken over transaction fee mechanism design in practice: Proposer-Builder Separation. Specifically, while some transactions are still simple peer-to-peer payments that reward the miner with a direct payment of (say) ETH, many transactions are significantly more complicated and allow a miner to extract profit through so-called Maximal Extractable Value (MEV)~\citep{DaianGKLZBBJ19}. For example, one user's transaction might be a trade of two cryptoassets on a DeFi Exchange, and the miner might profit by inserting `fake' transactions to sandwich them -- their profit comes \emph{not at all by increased revenue in the TFM itself}, but \emph{through external opportunities not at all captured by the TFM} (indeed, this is profitable even in EIP-1559 where the miner always receives zero direct revenue).~\cite{BahraniGR23} are the first to investigate the impact of MEV on Transaction Fee Mechanism design, and largely provide impossibility results. This work is also orthogonal to ours, as we aim for a deeper understanding of the core TFM regime.

\paragraph{Broader Literatures.}
Of course, blockchains relate to a wide variety of economic phenomena. 
The rapidly-evolving economics literature on the economics of blockchain is far too vast to summarize here, but includes \cite{HubermanLM17, arnosti2022bitcoin, budish2023trust, gans2024zero}.
Our work also relies on the large literature on Bayesian mechanism design following \cite{Myerson81}.
One particularly related concept from classical auction theory is that of credible auctions \citep{AkbarpourL20}, which has parallels with MMIC and other notions of incentive-compatibility for miners.
As discussed in \citep{ChungS23}, MMIC differs from credibility (in the private-communication model of \cite{AkbarpourL20}) in that credibility stipulates that the auctioneer can only deviate in ways that individual bidders cannot distinguish from the promised protocol, where MMIC (as well as our desiderata) stipulates that the miner can deviate in any way consistent with the block-building process.
Our framework is also conceptually related to the long line of work in the cryptography literature studying mechanism design---including \cite{HalpernT04,  abraham2006distributed, kol2008cryptography, izmalkov2011perfect, garay2013rational, canetti2023zero}, among many others---though we do not work directly with precise cryptography notions relying on complexity theory, and instead consider a simple idealized model of cryptographic auctions.

\section{Preliminaries}
\label{sec:prelims}

Before we introduce our model of a transaction fee mechanism (\autoref{sec:model}),
we review the basics of (Bayesian) mechanism design for the reader not familiar with these concepts.

In mechanism design terminology, we consider single-parameter quasi-linear Bayesian environments.
In these environments, users' values for being allocated to a binary outcome (e.g., receiving an item or not, being included in a block or not) are drawn from some publicly-known distribution,
and some feasible set of users are allocated  and each charged some price.

Formally, we define an \emph{environment} $(n,\FeasAlloc, \TypeDistr_1 \times \dots \times \TypeDistr_n)$ to consist of a number $n$ of users $\{1,2,\ldots,n\}$ (variously called bidders, players, or agents in the literature), 
a set $\FeasAlloc\subseteq 2^{\{1,\ldots,n\}}$ of feasible allocations (where $X\in\FeasAlloc$ indicates that each user $i\in X$ is allocated, and each $i\notin X$ is not),
and independent type distributions $\TypeDistr_1 \times \dots \times \TypeDistr_n$, which are arbitrary continuous distributions over $\R_{\geq 0}$.
User $i$'s value from being allocated is denoted by $v_i$, and is drawn from the distribution $\TypeDistr_i$, denoted $v_i\sim\TypeDistr_i$.
We typically focus on the i.i.d. case, in which $\TypeDistr_i = \TypeDistr$ is identical for each $i=1,\ldots,n$.

The set of \emph{outcomes} $\Outcomes = \FeasAlloc \times \R^n$ in an environment consists of tuples $(X, P_1,\ldots, P_n)$, where $X\in\FeasAlloc$ is the set of allocated users, and $P_i$ is the payment charged to each user $i$.
Given an outcome $(X, P_1,\ldots,P_n)$, each user $i$ such that $i \in X$ receives utility $v_i - P_i$, and every non-allocated bidder receives a utility $- P_i$.
Users receive a quasi-linear utility, i.e, given a distribution $\mathcal{D}$ over outcomes $(X, P_1, \dots, P_n) \in \Outcomes$, user $i$ with value $v_i$ receives an expected utility $\Es{(X, P_1, \dots, P_n)\sim\mathcal{D}}{v_i \cdot \1{i \in X} - P_i}$, where $\1{\cdot}$ denotes the indicator function.
 An \emph{auction} (or a \emph{mechanism}) over some environment is any mapping $\mathcal{M} : \T^{n} \to \Delta(\Outcomes)$ from bids of each user to a distribution over outcomes.

\subsection{Dominant Strategy Incentive Compatible Auctions}
\label{sec:prelims-dsic}

Consider a mechanism $\mathcal{M}$. For a profile $v = (v_i)_{1 \leq i \leq n}$, let $X_i(v)$ denote the probability that user $i$ gets allocated and $P_i(v)$ be the expected payment charged from $i$ (where expectation is taken over the randomness in the mechanism).
Since users are assumed to be expected utility maximizers, conditioned on the bids $b_{-i}$ placed by the other users, user $i$ bids $b_i$ such that his utility $v_i \cdot X_i(b_i, b_{-i}) - P_i(b_i, b_{-i})$ is maximized.
The mechanism $\mathcal{M}$ is \emph{dominant strategy incentive compatible} (DSIC) if bidding truthfully, i.e, bidding $v_i$, is the optimal strategy for each user $i$ irrespective of the bids $b_{-i}$ placed by the other bidders. 
Formally, $\mathcal{M}$ is DSIC if for all users $i$ and for all profiles $v = (v_1, \dots, v_n)$ of values, and all alternative bids $\widetilde{v}_i$, we have 
\[
 v_i\cdot X_i(v_i, v_{-i}) - P_i(v_i, v_{-i}) 
\ge
v_i\cdot X_i(\widetilde{v}_i, v_{-i}) - P_i(\widetilde{v}_i, v_{-i}).
\]

A standard example of a DSIC mechanism is the $(k+1)\textsuperscript{th}$-price auction with a reserve $\mathsf{r}$.
To sell $k$ items, the auctioneer collects bids from users and allocates the user with the highest $k$ bids (as long as they are larger than the reserve $\mathsf{r}$) and charges the maximum among the $k\textsuperscript{th}$ highest bid and $\mathsf{r}$.
We skip the proof of incentive compatibility here, since we will be revisiting the $(k+1)\textsuperscript{th}$-price auction later in the paper (\autoref{sec:C2PA}).

\subsection{Bayesian Incentive Compatible Auctions} \label{sec:BICAuctions}

Dominant strategy incentive compatibility enforces a strong condition--- irrespective of the bids $b_{-i}$ placed by the other users, user $i$ best responds by bidding his value $v_i$.
However, if user $i$ knows the distribution $\TypeDistr_{-i} = \Pi_{j \neq i} \TypeDistr_j$  of values of other users and believes other users are going to bid truthfully, user $i$ might want to bid $v_i$ if it optimizes his utility in expectation over the values $v_{-i}$ of the other users.

In accordance to the discussion above, we define the interim allocation and payment rules of the mechanism $\mathcal{M}$.

\begin{definition}[Interim Allocation and Payment Rules] \label{def:Interim}
    For a mechanism $\mathcal{M}$ with an allocation and payment rule $X$ and $P$ respectively, and distributions $\Pi_{i = 1}^n \TypeDistr$ of user values, the interim allocation rule $x_i:\TypeDistr_i \xrightarrow{} [0, 1]$ and the interim payment rule $p_i:\TypeDistr_i \xrightarrow{} \R$ are given as follows:
    \[
    x_i(v_i) = \Es{v_{-i} \sim \TypeDistr_{-i}}{X_i(v_i, v_{-i})} \text{\hspace{1cm} and \hspace{1cm}} p_i(v_i) = \Es{v_{-i} \sim \TypeDistr_{-i}}{P_i(v_i, v_{-i})}.
    \]
\end{definition}
Intuitively, the interim allocation and payment rules describe the user's expected allocation and payment assuming all other users bid truthfully when their values are drawn from $\TypeDistr_{-i}$. Writing user $i$'s expected utility in terms of the interim rules is extremely convenient--- when he has a value $v_i$ and submits a bid $b_i$, his expected utility equals $v_i \cdot x_i(b_i) - p_i(b_i)$.\footnote{
  Note that the term ``interim'' refers to the fact that from user $i$'s perspective, his value $v_i$ is already known, but other users' valuations are not.
  For clarity, throughout this paper we use lower-case letters such as $x_i(v_i)$ to denote such interim rules, and use
  upper-case letters such as $X_i(v) = X_i(v_1,\ldots,v_n)$ to denote the ``ex-post'' outcomes conditioned on the bids of all users.
}

An auction $\mathcal{M}$ is \emph{Bayesian incentive compatible} (BIC) if bidding truthfully optimizes the expected utility for every user, taking expectation over the values of all other users (motivating the assumption that $x_i$ and $p_i$ are defined in expectation over the other user's values and not other arbitrary bids).
In other words, $\mathcal{M}$ is BIC if for all $v_i, \widetilde{v}_i \in \R_{\ge 0}$,
\[
v_i\cdot x_i(v_i) - p_i(v_i) \ge v_i\cdot x_i(\widetilde{v}_i) - p_i(\widetilde{v}_i).
\]
Observe that DSIC is a stronger condition than BIC. BIC requires incentive compatibility in expectation while DSIC requires incentive compatibility to hold pointwise over the other users' values.

The Myerson's lemma (a.k.a the payment identity; \citealp{Myerson81}) characterizes all mechanism $\mathcal{M}$ with allocation and payment rules $X$ and $P$ such that $\mathcal{M}$ is BIC.
It requires no user $i$ should expect to get allocated with a lower probability when placing a larger bid, i.e, the interim allocation rule $x_i$ is monotonically non-decreasing in the user's value $v_i$.
Further, given a set of interim allocation rules $(x_i)_{1 \leq i \leq n}$, there exist a unique (up to additive constants) set of interim payment rules $(p_i)_{1 \leq i \leq n}$ such that the resulting mechanism is BIC.\footnote{
  Note however that not every set of monotone interim allocation rules will correspond to a feasible auction which always produces feasible outcomes in $\FeasAlloc$.
}

\begin{theorem}[\citealp{Myerson81}; Payment Identity] \label{thm:payment-identity}
A mechanism $\mathcal{C}$ with interim allocation rules $(x_i)_{1 \leq i \leq n}$ and interim payment rules $(p_i)_{1 \leq i \leq n}$ is BIC if and only if for all users $i$,
\begin{enumerate}
    \item the interim allocation rule $x_i(v_i)$ is monotonically non-decreasing in $v_i$, and
    \item the interim payment rule $p_i(v_i)$ satisfies $p_i(v_i) = \int_o^{v_i} zx_i'(z) \, dz + c$ for some constant $c$.
\end{enumerate}

\end{theorem}

\subsection{Bayes-Nash Equilibria and the Direct Revelation Principle}

We also consider auctions that are neither DSIC nor BIC.
Even though utility maximizing users will not bid truthfully in such mechanisms, users will play strategies that are (simultaneously) best responses to each other's strategies (in expectation).
Such a profile of simultaneous best responses is called a \emph{Bayes-Nash equilibrium} (BNE).

More formally, consider a mechanism $\mathcal{M}$ that takes as inputs messages from an arbitrary message space $\mathsf{Bid}_{1}, \dots, \mathsf{Bid}_{n}$ with an allocation rule $X:\mathsf{Bid}_{1} \times \dots \times \mathsf{Bid}_{n} \xrightarrow{} [0, 1]^n$ and a payment rule $P:\mathsf{Bid}_{1} \times \dots \times \mathsf{Bid}_{n} \xrightarrow{} \R^n$.
A user $i$'s strategy $s_i:\R \xrightarrow{} \mathsf{Bid}_{i}$ maps the value $v_i$ of user $i$ to $i$'s bid (input in $\mathsf{Bid}_i$) to the mechanism $\mathcal{M}$. 
A profile of strategies $(s_1, \dots, s_n)$ is a Bayes-Nash equilibrium if for all users $i$, conditioned on all other users playing the strategies $s_{-i}$, user $i$ with value $v_i$ optimizes his expected utility by bidding $s_i(v_i)$. 
In other words, for all users $i$ and for all values $v_i$
\begin{align*}
 & \Es{v_{-i}\sim \mathcal{T}_{-i}}{ v_i \cdot X_i(s_i(v_i), s_{-i}(v_{-i})) - P_i(s_i(v_i), s_{-i}(v_{-i})) }
\\ & \qquad \qquad \ge \Es{v_{-i}\sim \mathcal{T}_{-i}}{ v_i \cdot X_i(b_i, s_{-i}(v_{-i})) - P_i(b_i, s_{-i}(v_{-i})) }
\end{align*}
for all $b_i \in \mathsf{Bid}_i$.

For any BNE of any mechanism $\mathcal{M}$, a corresponding interim allocation and interim payment rule can be defined as follows:
\begin{definition}[Interim Allocation and Interim Payment Rules For a BNE] \label{def:ValueInterim}
    For a BNE of user strategies $\sigma = (s_1, \dots, s_n)$ induced by a mechanism $\mathcal{M}$ with allocation and payment rules $X$ and $P$ respectively, and distributions $\Pi_{i = 1}^n \TypeDistr$ of user values, the \emph{interim allocation rule} $x_i:\TypeDistr_i \xrightarrow{} [0, 1]$ and the \emph{interim payment rule} $p_i:\TypeDistr_i \xrightarrow{} \R$ corresponding to $\sigma$ are given as follows:
    \[
    x_i(v_i) = \Es{v_{-i} \sim \TypeDistr_{-i}}{X_i(s_i(v_i), s_{-i}(v_{-i}))} \text{\hspace{1cm} and \hspace{1cm}} p_i(v_i) = \Es{v_{-i} \sim \TypeDistr_{-i}}{P_i(s_i(v_i), s_{-i}(v_{-i}))}.
    \]
\end{definition}

We now describe the classical revelation principle, which says that any BNE of any mechanism induces a BIC mechanism which implements the same interim allocation and payment rules.
Intuitively, consider an alternative mechanism $\widetilde{\mathcal{M}}$ that implements $\mathcal{M}$ supplemented with a ``strategizing service''.
Instead of having to devise the optimal strategy $s_i$, user $i$ can instead submit a bid equal to his value $v_i$.
The strategizing service bids $s_i(v_i)$ on user $i$'s behalf.
Assuming other users bid truthfully in $\widetilde{\mathcal{M}}$, user $i$ bidding $v_i$ in $\widetilde{\mathcal{M}}$ corresponds to the strategy profile $(s_i(v_i), s_{-i}(v_{-i})$ in $\mathcal{M}$.
Since $(s_1,\ldots,s_n)$ is a BNE of $\mathcal{M}$, we see that $v_i$ is user $i$'s expected-utility-maximizing strategy in $\widetilde{\mathcal{M}}$; since this is true for all users, this means $\widetilde{\mathcal{M}}$ is BIC.
The \emph{direct revelation principle} follows directly from this argument:

\begin{lemma}[Direct Revelation Principle] \label{thm:DirectRevelation}
    Fix a mechanism $\mathcal{M}$ and a Bayes-Nash equilibria $\sigma = (s_1,\ldots,s_n)$ of user strategies.
    Then, there exists a BIC mechanism $\widetilde{\mathcal{M}}$ such that the interim allocation and payment rules corresponding to $\sigma$ in $\mathcal{M}$ both equal the interim allocation and payment rules of the BIC mechanism $\widetilde{\mathcal{M}}$.
\end{lemma}

In particular, the payment identity (\autoref{thm:payment-identity}) can be applied to any BNE of $\mathcal{M}$ via the interim payment rules.
The mechanism $\widetilde{\mathcal{M}}$ is BIC, and therefore, the interim payment rule $p_i$ is uniquely determined by the interim allocation rule $x_i$.
We get the following direct consequence of \autoref{thm:payment-identity}:

\begin{theorem}[Revenue Equivalence] \label{thm:RevenueEquivalence}
    For any two BNEs $s^{(1)}, s^{(2)}$ of any two mechanisms,
    if the interim allocation rules $(x_i)_{1 \leq i \leq n}$ induced by the BNEs $s^{(1)}$ and $s^{(2)}$ are the same,
    then the interim payment rules $(p_i)_{1 \leq i \leq n}$ induced by $s^{(1)}$ and $s^{(2)}$ are also the same.
    Further, the expected total payment made by users from all such BNEs equal $\sum_{i = 1}^n \Es{v_i \sim \TypeDistr_i}{p_i(v_i)}$.
\end{theorem}

\subsection{Revenue Optimization and Virtual Values}

In classical auctions, the auctioneer collects the payments made by the users and hence, an auctioneer optimizing her revenue would optimize the total payments received from users.
Working with payment rules directly while trying to optimize revenue can be mathematically challenging, for instance, 
since one must ensure that the payment rule of the revenue optimal mechanism simultaneously satisfies the payment identity (\autoref{thm:payment-identity}).
Instead, \citet{Myerson81} provides a recipe to compute the expected payments via the expected sum of \emph{virtual values} (a.k.a \emph{virtual welfare}), described below.
Thus, optimizing revenue is equivalent to optimizing the total virtual welfare, which is mathematically similar to optimizing the total welfare (i.e, optimizing the total value).

\begin{definition}[Virtual Values]
    \label{def:virtual-value}
    For a (continuous) type distribution $\TypeDistr$ with cumulative distribution function $F_{\TypeDistr}(t) = Pr_{v\sim\TypeDistr}(v \le t)$ and probability density function $f_{\TypeDistr}(t) = \frac{d F_{\TypeDistr}(t)}{dt}$, define the \emph{virtual value} function $\phi(\cdot)$ as follows:
    \[ \phi_i(v_i) = v_i - \frac{1-F_{\TypeDistr}(v_i)}{f_{\TypeDistr}(v_i)}. \]
\end{definition}

\begin{theorem}[\citealp{Myerson81}; Revenue Equals Virtual Welfare] 
\label{thm:revenue-equals-virtual-welfare}
    Given distributions $\Pi_{i = 1}^n \TypeDistr_i$ of user values with virtual values $(\phi_i)_{1 \leq i \leq n}$, for any BIC mechanism with interim allocation rules $(x_i)_{1 \le i \le n}$ and interim payments $(p_i)_{1 \leq i \leq n}$ (or equivalently, any BNE with value interim allocation and payment rules $(x_i)_{1 \le i \le n}$ and $(p_i)_{1 \leq i \leq n}$), we have: 
    \[ \Es{v \sim \TypeDistr^n}{\sum_{i = 1}^n p_i(v_i)} 
    = \Es{v \sim \TypeDistr^n}{\sum_{i = 1}^n \phi_i(v_i) \cdot x_i(v_i)}.
    \tag{*}
    \label{eqn:revenue-equals-virtual-welfare}
    \]
\end{theorem}

For a simple example, consider the posted-price mechanism with a single user and a single good. 
In this auction, the user is posted a price $p$, and is sold the good at the price $p$ when the user has a value $v_i \geq p$.
Observing that the posted price mechanism is DSIC is straightforward.
Suppose the user's value $v$ is drawn from the uniform distribution $U[0, 1]$ (with a cumulative distribution function $F(v) = v$ and a density function $f(v) = 1$ for all $0 \leq v \leq 1$).
By posting a price $p$, the good is sold with probability $1 - F(p) = 1-p$, fetching an expected revenue equal to $p \, (1-p)$.
The user's virtual value $\phi(v) = v - \frac{1 - F(v)}{f(v)} = v - \frac{(1-v)}{1} = 2v - 1$.
The user is allocated ($x(v) = 1$) whenever $v \geq p$ and not allocated ($x(v) = 0$).
Plugging the above into the right hand side of (\ref{eqn:revenue-equals-virtual-welfare}), we get the expected virtual welfare equals $\int_p^1 (2v - 1) \, dv = p \, (1-p)$.

We now discuss revenue-optimal auctions.
\autoref{thm:revenue-equals-virtual-welfare} says that optimizing revenue corresponds to optimizing the virtual welfare.
However, for the mechanism (or equivalently, the value interim allocation rules) to be BIC, the interim allocation rules must be monotonically non-decreasing.
For general distributions $\TypeDistr_i$ of user $i$'s value, the virtual value function $\phi_i(v_i)$ need not be monotone non-decreasing, and therefore optimizing for virtual welfare does not automatically guarantee that interim allocation rules are monotone.
The revenue optimal mechanism involving such distributions can be fairly complex.
However, a standard simplifying assumption is the notion of regularity, which states that the virtual value function is monotone non-decreasing.

For regular distributions, it follows that optimizing for virtual welfare also ensures that the interim allocation rules all satisfy monotonicity.
Moreover, for a given distribution $\TypeDistr$, if there exists a value $\mathsf{r}$ such that $\phi(\mathsf{r}) = 0$, then for all values $v \leq \mathsf{r}$, the virtual value is non-positive (by monotonicity).
Thus, when the feasibility constraints are downwards-closed (i.e., when a user can always be removed from a feasible allocation and this results in another feasible allocation), we can thus simply exclude any user with a negative virtual value. This must increase the expected virtual welfare, and thus the expected revenue.
We call $\mathsf{r}$ to be the \emph{monopoly reserve} for the regular distribution $\TypeDistr$.

\begin{definition}[Regular Distributions and Monopoly Reserve] \label{def:Regular}
    A distribution $\type$ is \emph{regular} if the virtual value function $\phi(v)$ is monotonically non-decreasing in $v$. 
    A \emph{monopoly reserve} of the regular distribution $\TypeDistr$ is given by any value $\mathsf{r}$ such that $\phi(\mathsf{r}) = 0$.\footnote{
        For simplicity, we assume the existence of an $\mathsf{r}$ such that $\phi(\mathsf{r}) = 0$; 
        any such $r$ can serve as the monopoly reserve.
        However, if there does not exist such a value $\mathsf{r}$, the monopoly reserve can be defined as $\mathsf{r} = \sup \{v | \phi(v) \leq 0\}$. With slight modification, all definitions and results discussed in the paper continue to hold under this definition.
    }
\end{definition}

\begin{theorem}[\citealp{Myerson81}]
\label{thm:Regular-RevOpt}
    Let the set of feasible allocations $\FeasAlloc$ consist of any subset of users of cardinality at most $k$.
    Let users' values be drawn i.i.d. from from a regular distribution $\TypeDistr$ for any number of users $n$. 
    Then, the revenue optimal BIC mechanism is the $(k+1)$st price auction with monopoly reserve $\mathsf{r}$.
\end{theorem}

\section{Model of Transaction Fee Mechanisms}
\label{sec:model}

\newcommand{\bBuild}{\mathcal{B}}
\newcommand{\onCG}{\mathcal{C}}
\newcommand{\offCG}{\mathcal{D}}

\newcommand{\mi}{\mathsf{mi}}
\newcommand{\usr}{\mathsf{usr}}

\newcommand{\Moff}{\mathcal{M}_{\mathsf{off}}}

\subsection{Formal Model}

We now define our model of transaction fee mechanisms.
We model TFMs in three layers, starting from the primitives of the protocol, and adding increasing levels of strategic flexibility.
The innermost layer is the ``block-building process'' $\bBuild$, which is an algorithm that dictates how all new blocks must be created, and is fixed by the blockchain designer (e.g., the Ethereum community).
The second layer, the ``on-chain game'' $\onCG$, represents all strategic manipulations available on-chain to the miner and users
based only on manipulating the inputs passed to $\bBuild$.
Previously-considered notions typically work in terms of the on-chain game.

The third stage of our model, the ``off-chain game'' $\offCG$, allows the miner
to conduct a separate off-chain auction to determine the inputs to $\onCG$ (and hence to $\bBuild$).
This substantial additional ability granted to the miner is our primary departure point from previous works.

To begin, we define the set of possible outcomes of a TFM.
Given a set of bids indexed by $[N]$, a (deterministic) \emph{outcome} of a TFM (given bids indexed by $[N]$) is a triple $\big( X, (P_i)_{i\in [N]}, \mathsf{Rev} \big)$, where $X \subseteq [N]$ specifies a binary outcome for each bid (i.e., a $X \subseteq [N]$ where bids indexed by $i\in X$ are included in the block), a payment $P_i$ for each bid, and a revenue $\mathsf{Rev}$ transferred to the miner.\footnote{
    Note that we allow miner revenue to be less or greater than the sum of the transaction fees (money creation, e.g. from a constant block reward, or money burning).
}
Denote the set of all outcomes by $\mathcal{X}$.

We now define our main model. A running example of the definitions below is provided in \autoref{sec:ModelExample}.

\begin{definition}[Transaction Fee Mechanisms] 
\label{def:model}

A transaction fee mechanism (TFM) is specified by a block-building process $\bBuild$. Strategic play can proceed according to on-chain game $\onCG$ and off-chain game $\offCG$.
These are defined as follows:

\begin{enumerate}
    \item \label{item:model-bBuild}
    The \emph{Block-Building Process} is a function $\bBuild(a_{\mi}, b_1, \dots, b_{N})$ from a miner advice $a_{\mi}$ and bids $(b_1, \dots, b_{N})$ to distributions over outcomes $(X, (P_i)_{i\in [N]}, r) \in \mathcal{X}$. In more detail, this is any algorithm that takes any number of bids $b_i$ from a pre-defined bid space $\mathsf{Bid}_{\usr}$,\footnote{
        We typically consider the bid space $\mathsf{Bid}_{\usr}$ to be $\R_{\geq 0}$ (like in all sealed bid auctions). Other natural bid spaces include $\{\text{yes, no}\}$ in a posted-price mechanism (signifying the answer to the question ``are you willing to buy the good at the given price?''), or $\{\text{yes, no}\} \times \R_{\ge 0}$ in EIP-1559 (with the real number signifying the tip).
    \label{fn:BidSpace}} and some advice $a_{\mi}$ from an advice space $\mathsf{Adv}_{\mi}$ provided by the miner.
    We restrict attention to simultaneous, single-round mechanisms.
    The block building process returns a distribution over outcomes $\mathcal{X}$. 

    The block-building process also specifies a cryptographic model, which determines the level of cryptography in the communication between users and the miner in $\bBuild$.
    We primarily consider two abstract cryptographic models: the ``plaintext" model, where the miner will be able to see the bids before submitting to $B$, and the ``miner-gatekeeper" model, where the bids are encrypted by users before sending them to the miner.\footnote{
        Technically, the input to the miner-gatekeeper model are encrypted commitments to bids. However, for simplicity, we assume that the input to $\bBuild$ is just the set of bids, and model the fact that the miner is oblivious to the exact contents of the bids by restricting the miner's strategy space.
    }
    The effects of the cryptographic model is captured in the miner's strategy space defined in the on-chain game, as we describe below.

    \item \label{item:model-onCG}
    The block-building process $\bBuild$ induces an \emph{On-Chain Game} $\onCG$ involving the miner and the users.  
    In $\onCG$, users' strategy space $S^\onCG_\usr$ is submitting zero or more bids.
    Formally, each user $i$'s strategy $s_{\usr, i}^\onCG$ is a mapping from the user's type (i.e., willingness-to-pay) to a set of bids to be submitted to $\bBuild$.\footnote{
        Note that, as in \autoref{sec:prelims}, we use the term ``strategies'' for mappings from users' types to their actions / messages in games.
    }

    The miner's strategy space $S^\onCG_\mi$ is different under the two levels of cryptography. 
    In the plaintext model, the miner can arbitrarily condition their behavior on the submitted bids. 
    Formally, in the plaintext model, the miner's strategy $s^\onCG_\mi \in S^\onCG_\mi$ is a mapping from the set $H$ of all submitted bids to an advice $a_{\mi} \in \mathsf{Adv}_{\mi}$, 
    a subset $I\subseteq H$ of submitted bids, and another set of ``fabricated'' bids $J$.
    The outcome of $\onCG$ is then the block built according to $\bBuild(a_{\mi};I, J)$.

    In the miner-gatekeeper model, the miner can only condition their behavior on \emph{who} submits a bid, not the contents of the bid. 
    Formally, in the miner-gatekeeper model, the miner's strategy $s^\onCG_\mi \in S^\onCG_\mi$ is a mapping from the set $H$ of all \emph{bidders} $i$ who submit a bid, to an advice $a_\mi \in \mathsf{Adv}_\mi$, a subset of submitted bidders $I\subseteq H$, and another set of fabricated bids $J$.
    The outcome of $\onCG$, analogous to the plaintext model, is then the block built according to $\bBuild(a_\mi; (b_i)_{i\in I}, J)$.\footnote{
      Note in particular that in both models, the bidder's index $i$ serves as an ``identifier'' for the particular bid.
      However, we do not restrict users to a single identifier,  so a bidder $i$ can submit a fake bid with identifier $i'$ for some $i'$ different than any other bidders' (though we typically leave this implicit in the notation for simplicity).
    }\textsuperscript{,}\footnote{We are agnostic to the cryptographic machinery used by the TFM. While \citet{ChungSW23} recommend the use of multi-party computation, fully homomorphic encryption, digital signature schemes, or verifiable random functions are also potential cryptography tools that can be used by the TFM.}

    Given a miner strategy $s_\mi^\onCG$ and a profile of user strategies $s^\onCG_\usr = \big(s^\onCG_{\usr, i} \big)_i$ and values $v = (v_i)_i$,
    we let $\onCG( s_\mi^\onCG; s^\onCG_\usr(v) )$ denote the result of the on-chain game when the miner plays according to $s_\mi^\onCG$ and users play according to $s^\onCG_\usr(v) = \big(s^\onCG_{\usr, i}(v_i) \big)_i$.
    Let $X_i^\onCG(s_\mi^\onCG; s_\usr^\onCG(v))$ denote the probability that user $i$ is allocated and let $P_i^\onCG(s_\mi^\onCG; s_\usr^\onCG)$ denote $i$'s expected transaction fee;\footnote{
        If user $i$ submits more than one bid, we let $X_i(\cdot)$ denote the probability that at least one of these bids is allocated in $\mathcal{B}$, and we let $P_i(\cdot)$ denote the expected sum of the payments charged to these bids. This reflects our assumption that each user has just one bid they want to submit, or equivalently, that the users' total value is additive over the different transactions they want to submit.
    } bidder $i$'s total expected utility is then 
    $v_i \cdot X_i^{\onCG}(s_\mi^\onCG; s_\usr^\onCG(v)) - P_i^{\onCG}(s_\mi^\onCG; s_\usr^\onCG(v))$.
    Let $\rev^\onCG(s_\mi^\onCG; s_\usr^\onCG(v))$ denote the miner's expected reward returned by $\bBuild$, minus the expected payments made by the fabricated bids injected into $\bBuild$ by the miner (where the expectation is taken over the randomness in the mechanism); this is the miner's total utility from $\onCG$.

    \item \label{item:model-offCG}
    The \emph{Off-Chain Game} $\offCG$ induced by $\bBuild$ is a game defined as follows. 
    To begin, the miner announces and commits to an ``off-chain mechanism'' $\Moff$ with message space $\mathcal{U}_{\mathsf{off}}$.
    This off-chain mechanism allows the miner to run any mechanism of her choosing to determine how she and the users will behave in $\onCG$.
    Formally, $\Moff$ is a function from any number $N'$  of reports in $\mathcal{U}_{\mathsf{off}}$, to the miner's and users' behavior $s_\mi^\onCG$ and $s_\usr^\onCG = (s_{\usr, i}^\onCG)_{i\in\{1,\ldots,N'\}}$ in $\onCG$, as well as off-chain  payments $p_{\mathsf{off}} = ( p_{\mathsf{off}, i} )_{i \in \{1,\ldots,N'\}}$.

    The only assumption we make about $\Moff$ is that each user always has the option of not participating in the off-chain mechanism and instead submits bids directly to $\onCG$.
    Formally, we assume that in $\Moff$, for all users $i$ and every strategy $s^\onCG_{\usr, i}$ in $\onCG$, there exists at least one message $\bot_\mathsf{off}(s^\onCG_{\usr, i}) \in \mathcal{U}_{\mathsf{off}}$ (called an abstaining bid) such that 
    \begin{enumerate}
        \item whenever user $i$ submits $\bot_\mathsf{off}(s^\onCG_{\usr, i})$, the result of $\Moff$ specifies that $i$ will play $s^\onCG_{\usr, i}$ in $\onCG$ and pay $p_{\mathsf{off}, i} = 0$, and,
        \item the result of $\Moff$ for all users $j$ that submit non-abstaining bids to the off-chain mechanism is invariant to the set of abstaining bids.   
    \end{enumerate}
    Property~(b) signifies that whenever a user $i$ refuses to participate in the off-chain mechanism, the miner and other users lack any information about $i$, including its existence. 
    
    After the miner commits to an off-chain mechanism $\Moff$,
    users' strategies consist of submitting some message to $\Moff$. The on-chain game proceeds according to the result of $\Moff$, and utilities are given by the sum of utilities from the on-chain game and the off-chain payments from $\Moff$.
    Formally, a strategy $s^{\offCG, \Moff}_{\usr, i} \in S^{\offCG, \Moff}_{\usr}$ of user $i$ is any mapping from user $i$'s type to some set of messages to $\Moff$.\footnote{
        As in the case of $\onCG$, we allow users to submit multiple messages $\Moff$; in this case, the result of $\Moff$ will specify separate strategies and charge separate payments in $\onCG$ for the different messages.
        As for $\onCG$, we typically leave this implicit in the notation for simplicity.
    }
    We let $\offCG(\Moff; s_\usr^{\offCG, \Moff}(v) )$ denote the result of running the on-chain game $\onCG$ according to $\Moff(s_\usr^{\offCG, \Moff}(v))$, along with off-chain payments.
    In particular, denote by $X_i^{\offCG}(\Moff; s_\usr^{\offCG, \Moff}(v)) = X_i^{\onCG}(s_\mi^\onCG; s_\usr^{\onCG})$, where $s_\mi^\onCG, s_\usr^{\onCG}$ are set by $\Moff(s_\usr^{\offCG, \Moff}(v))$.
    Denote by $P_i^{\offCG}(\Moff; s_\usr^{\offCG, \Moff}(v)) = P_i^{\onCG}(s_\mi^\onCG; s_\usr^{\onCG}) + p_{\mathsf{off}, i}$, where $p_{\mathsf{off,i}}$ is additionally set by $\Moff(s_\usr^{\offCG, \Moff}(v))$.
    User $i$'s total expected utility is then
    \[  v_i \cdot X_i^{\offCG}(\Moff; s_\usr^{\offCG, \Moff}(v)) - P_i^{\offCG}(\Moff; s_\usr^{\offCG, \Moff}(v)).\]
    Finally, let $\rev^\offCG(\Moff; s_\usr^{\offCG, \Moff}(v)) = \rev^\onCG(s_\mi^\onCG; s_\usr^\onCG(v)) + \sum_i p_{\mathsf{off}, i}$ denote the miner's expected reward from $\onCG$ plus the sum of the off-chain payments.

    We say that the miner runs a \emph{trivial off-chain mechanism} if $\Moff$ has a message space $\mathcal{U}_{\mathsf{off}}$ which consists of only abstaining bids
    (and hence the miner's on-chain strategy specified by the result of $\Moff$ must be a constant, though we note that this strategy may not be compliant). 
    This corresponds to the miner not running an off-chain mechanism and instead playing a fixed strictly on-chain strategy.
    We will sometimes abuse notation to describe an on-chain profile of strategies $(s_\mi^\onCG, s^\onCG_\usr)$ as a strategy profile in the off-chain game, which corresponds to the trivial mechanism $\Moff$ which always specifies that the miner plays the fixed strategy $s_\mi^\onCG$, and users play strategies $\bot_{\mathsf{off}}(s^\onCG_\usr)$.
\end{enumerate}
    
\end{definition}

See \autoref{fig:model} for a high-level illustration of our model.
As we will see in \autoref{sec:Definitions}, previously-considered notions such as UIC and MMIC are naturally expressed in in the language of the game $\onCG$.
In contrast, modeling the miner attack on EIP-1559 motivating our paper (as discussed in \autoref{sec:intro}) requires the miner manipulating the off-chain game $\offCG$. 

\begin{figure}
    \centering
    \includegraphics[width=\textwidth]{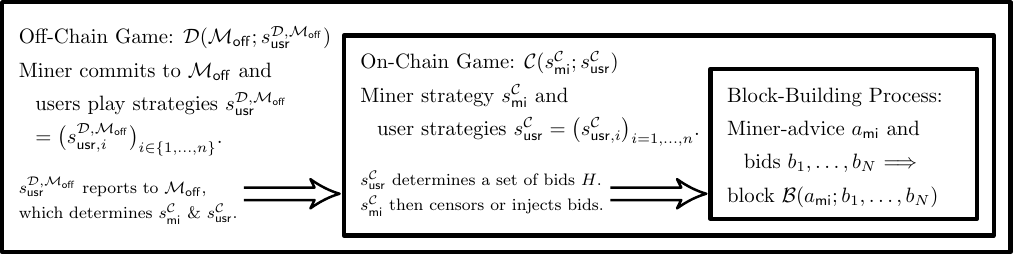}
    \caption{Model of TFMs.}
    \label{fig:model}
\end{figure}

\subsection{Example: The Squared Revenue Second Price Auction (SR2PA)} \label{sec:ModelExample}

We run over the above notations with an example of a novel auction format under which the miner can benefit by using a completely separate off-chain auction to determine her behavior on-chain.
We define the Squared Revenue Second Price Auction (SR2PA) to auction off inclusion in a block of size $1$ (a single item).

Informally, the SR2PA implements the second-price auction in the plaintext model, while allowing the miner to set a reserve, and burning some function of the on-chain payments.
In particular, we will consider cases where users have values less than $1$ and will always bid $b_i \le 1$, and the miner will receive only the \emph{square} of the payment $p_i$ made by the highest bidder (so that $p_i - p_i^2 \ge 0$ is burned).
Formally, the block-building process $\bBuild_{\textsf{SR2PA}}$ of the SR2PA is the following:
\begin{enumerate}
    \item Users submit a bid from $\mathsf{Bid}_{\usr} = \R_{\geq 0}$, as in any sealed-bid auction. Bids are submitted in the plaintext model.
    \item $\mathsf{Adv}_{\mi} = \R_{\geq 0}$, allowing the miner to set a reserve price.
    \item The allocation rule $X_i$ to user $i$ is given by
    \begin{equation*}
        \notag
        X_i(a_{\mi} ;b_i, b_{-i}) = \begin{cases}
            1 & \text{if } b_i \geq b_j \text{ for all } j \neq i \text{ and } b_i \geq a_{\mi}  \\
            0 & \text{otherwise}
        \end{cases}
    \end{equation*}
    with ties broken lexicographically.
    \item The payment rule $P_i$ is as follows.
    \begin{equation*}
        \notag
        P_i(a_{\mi}; b_i, b_{-i})
        =
        \begin{cases}
            0 & \text{if } X_i(b_i, b_{-i}) = 0 \\
            \max_{j \neq i} \{b_j, a_{\mi}\} & \text{otherwise}
        \end{cases}
    \end{equation*}
    \item The miner receives the square of the collected payments if the total payments is less than $1$, and receives nothing otherwise.
    $$\mathsf{Rev}(a_{\mi}; b) = \Big(\sum_{i = 1}^N P_i(a_{\mi}; b_i, b_{-i}) 
    \cdot \1{P_i(a_{\mi}; b_i, b_{-i}) \le 1}\Big)^2$$
\end{enumerate}

We now illustrate the on-chain game $\onCG$ arising from $\bBuild_{SR2PA}$ through a specific profile of strategies.
Consider the miner strategy which, given a set of bids $H$, sets reserve $a_{\mi} = \max_{i} \{b_i\} - \epsilon$ (for some small $\epsilon > 0$) to be just below the largest submitted bid.
This way, the miner gets a revenue of $(\max_{i} \{b_i\} - \epsilon)^2$, instead of the square of the second largest bid.
As a consequence, the user $i$ with the highest bid $b_i$ always gets allocated and pays $b_i - \epsilon$.
The auction is reduced to a first-price auction, since the winner (essentially) pays its bid.
The miner gets a revenue of $(\max_{i} \{b_i\} - \epsilon)^2$ whenever the largest bid is less than $1$.
In the first-price auction, the equilibrium bidding strategy $s_{\usr, i}^\onCG$ for user $i$ involve shading their bids according to the distribution of values $\TypeDistr$, i.e, placing a bid $b_i(v_i)$ less than their value $v_i$.
These strategies constitute one equilibrium in the on-chain game $\onCG$.
We discuss an equilibrium with a larger revenue from the off-chain game below. 

In the off-chain game $\offCG$ arising from $\bBuild_{SR2PA}$, the miner can use the off-chain mechanism to avoid having a portion of the revenue burnt, as we illustrate next. 
In $\Moff$, the miner simulates a second-price auction with an optimal reserve chosen according to the distribution $\TypeDistr$, and collects off-chain payments as dictated by the second-price auction.
The miner also commits to setting reserve $a_{\mi} = 0$ in the on-chain game $\onCG$. Further, the miner censors all bids except the highest bid in $\Moff$, and the highest bidder is instructed to bid zero on-chain.
Since the user with the highest bid is the only uncensored bid on-chain, the user gets allocated and does not have to make any payments on-chain.
All told, this off-chain mechanism implements a revenue-optimal second-price auction (when the distribution $\TypeDistr$ is regular), while avoiding money burning on-chain and thus getting a strictly higher miner revenue than the equilibrium in $\onCG$ discussed above.

Formally, the off-chain mechanism and user strategies are:
\begin{enumerate}
    \item $\Moff$: The message space $\mathcal{U}_{\mathsf{off}}$ equals $\R_{\geq 0}$.
    $\Moff(u_1, \dots, u_{N'})$ returns:
    \begin{itemize}
        \item $s^\onCG_\mi = (0, \{b_i: u_i \geq u_j \text{ for all } j \ \neq i \text{ and } u_i \geq r_0 \text{ for some reserve } r_0\}, \emptyset)$, i.e, the miner sets a zero reserve on-chain, censors all but the largest off-chain bid (if it is larger than the off-chain reserve $r_0$) and injects no fake bids.
        \item $s^\onCG_{\usr, i} = 0$. The mechanism recommends bidding zero to each user.
        \item The payments in the off-chain mechanism are determined as follows:
        \begin{equation*}
            \notag
         p_{\mathsf{off}, i}(u_1, \dots, u_{N'}) = \begin{cases}
             0 & \text{if } u_i < r_0 \text{ or } u_i \neq \max_{j} \{u_j\} \\
             \max_{j \neq i} \{u_j\} & \text{otherwise}
         \end{cases}
         \end{equation*}
    \end{itemize}
    \item $s_\usr^{\offCG, \Moff}(v)$ and $\offCG(\Moff; s_\usr^{\offCG, \Moff}(v) )$:
    \begin{itemize}
        \item User $i$ plays the strategy $s^{\offCG, \Moff}_{\usr, i}$ that bids truthfully in $\Moff$, i.e, bidding $u_i = v_i$. This is the user's dominant strategy since the miner has committed to running a second-price auction off-chain.
        \item $X_i^{\offCG}(\Moff; s_\usr^{\offCG, \Moff}(v))$: the highest bidder is included in the block if the highest bid is larger than the reserve. The block is empty otherwise.
        \item The users make zero payment on-chain. Thus, $P_i^{\offCG}(\Moff; s_\usr^{\offCG, \Moff}(v)) = p_{\mathsf{off}, i}(u_1, \dots, u_{N'})$.
    \end{itemize}
\end{enumerate}

In summary, $\Moff$ defines a completely separate off-chain auction which allows the miner to collect payments and maximize their revenue, while avoiding any of the users' transactions being burnt on chain.
Besides using SR2PA above to illustrate the off-chain game, we revisit the various equilibria of this mechanism in more detail in \autoref{sec:SR2PA} in order to discuss the relationship between our main definitions.

\section{Simplicity Definitions} \label{sec:Definitions}

In this section, we give the main definitions of our paper.
Throughout, we assume that users values are distributed i.i.d from a distribution $\TypeDistr$ (though these definitions can be extended to environments where values are drawn from independent, but not necessarily identical distributions).

To begin, we define two simple pieces of terminology for convenience.
First, we consider a strategy profile $(s_\mi^\onCG, s_{\usr, 1}^\onCG, \dots, s_{\usr, n}^\onCG)$ in the on-chain game $\onCG$---or $(\Moff, s_{\usr, 1}^{\offCG, \Moff}, \dots, s_{\usr, n}^{\offCG, \Moff})$ in the off-chain game $\offCG$---such that conditioned on the miner's strategy, the users' strategy profile is a Bayes-Nash equilibrium. 
We call such a strategy profile to be a \emph{user Bayes-Nash equilibrium} (user BNE), and all of our definitions restrict attention to user Bayes-Nash equilibria.

Second, we call a miner strategy \emph{compliant} if there exists a miner advice $a_\mi \in \mathsf{Adv}_\mi$ such that, for any submitted set of user bids $H$, the action played by the miner is $(a_\mi, H, \emptyset)$. 
In other words, a compliant miner strategy chooses a constant miner advice (constant independent of the set of bids), does not censor any bids, and does not fabricate any bids;
informally, these are exactly the miner strategies we think of as ``following the protocol''.
For any $a_\mi$, we denote the corresponding compliant miner strategy by $s_{\mi}^{\mathsf{comp}}(a_\mi)$.

\subsection{On-Chain Simplicity}
\label{sec:Definitions-On-Chain-Simple}

We now define on-chain simplicity.
On-chain simplicity is a property of a strategy profile $\sigma^\onCG = (s_\mi^\onCG, s_{\usr,1}^\onCG,\allowbreak\ldots,\allowbreak s_{\usr,n}^\onCG)$ in the on-chain game $\onCG$, 
and specifies that (i) all agents ``follow the protocol'' (for users, this means bidding their value, and for the miner, this means playing a compliant strategy) in $\sigma^\onCG$, and (ii) $\sigma^\onCG$ forms an equilibrium.
Thus, on-chain simplicity is very close in spirit to the combination of UIC and MMIC (\citealp{Roughgarden20, Roughgarden21}; see also \autoref{rem:on-chain-vs-previous} below).

\begin{definition}[On-Chain Simplicity] \label{def:OnChainSimple}
    Consider the on-chain game $\onCG$ induced by some block-building process, and a distribution $\TypeDistr$ of user values.
    The user BNE $(s_\mi^\onCG, s_{\usr, 1}^\onCG,\allowbreak \dots, \allowbreak s_{\usr, n}^\onCG)$ is \emph{on-chain simple} for
    $\onCG$ and $\TypeDistr$ if it satisfies the following two properties: 
    \begin{itemize}
        \item \emph{(On-Chain User Simple)} Conditioned on the miner's strategy $s_\mi^\onCG$, the resulting auction is DSIC and individually rational for users. In other words:
        \begin{enumerate}
            \item $s_{\usr, i}^\onCG(v_i) = v_i$, i.e, users report their values truthfully in $\onCG$.\footnote{We abuse notation to describe the action suggested by a strategy to be a real number instead of a set of bids (as is necessitated by our model). Technically, the above mentioned strategy should be $\{v_i\}$, signifying that the user submits only one bid, equal to his value. \label{fn:StrategyNotationAbuse}}\textsuperscript{,}\footnote{We are assuming that bids are non-negative real numbers in the above definition. However, the assumption is without loss of generality and can be easily extended to other bid spaces. Instead of requiring ``bidding their value'' being the dominant strategy, we require that there exists a bid that optimizes the user's utility irrespective of the bids placed by the other users. For instance, if the bid space is \{yes, no\} in a posted-price mechanism (see \autoref{fn:BidSpace} for details and more examples of alternate bid spaces), a user with value larger (lower) than the price would always want to bid yes (no) irrespective of the bids of other users.}
            \item For any set of bids $b_{-i}$ submitted by other users, user $i$'s best response is to bid truthfully. 
            Formally, for any value $v_i$ of user $i$ and strategy $\widetilde{s}_{\usr, i}^\onCG(\cdot) \in S_{\usr}^\onCG$ for user $i$,
            \[ v_i \cdot X_i^{\onCG}(s_\mi^\onCG; v_i, b_{-i}) - P_i^{\onCG}(s_\mi^\onCG; v_i, b_{-i})
            \ge
            v_i \cdot X_i^{\onCG}(s_\mi^\onCG; \widetilde{s}_{\usr, i}^\onCG(v_i), b_{-i}) - P_i^{\onCG}(s_\mi^\onCG; \widetilde{s}_{\usr, i}^\onCG(v_i), b_{-i}). \]
            \item $v_i \cdot X_i^{\onCG}(s_\mi^\onCG; v_i, b_{-i}) - P_i^{\onCG}(s_\mi^\onCG; v_i, b_{-i}) \ge 0$.
            The user is guaranteed a non-negative utility by bidding his value.
            Otherwise, (as a consequence of DSIC), the user nets a negative utility and is better off from not participating in the auction.
            
        \end{enumerate}
        \item \emph{(On-Chain Miner Simple)}
        The miner's strategy $s^\onCG_\mi$ is compliant, and is a best response to the users' strategies. Formally, the latter condition says that conditioned on users bidding according to the strategy profile $(s_{\usr, 1}^\onCG, \dots, s_{\usr, n}^\onCG)$, the miner optimizes her expected revenue by playing $s_\mi^\onCG$, 
        i.e.,  for any strategy $\widetilde{s}_\mi^\onCG \in S^\onCG_\mi$ of the miner,
        \[
        \Es{v \sim \TypeDistr}{\rev^\onCG(s_\mi^\onCG; s_\usr^\onCG(v))}
        \ge
        \Es{v \sim \TypeDistr}{\rev^\onCG(\widetilde{s}_\mi^\onCG; s_\usr^\onCG(v))}.
        \]
    \end{itemize}
\end{definition}

The following corollary follows directly from the definition.
\begin{corollary}
    If the strategy profile $(s_\mi^\onCG, s_{\usr, 1}^\onCG, \dots, s_{\usr, n}^\onCG)$ is on-chain miner simple for the on-chain game $\onCG$ induced by some block-building process $\bBuild$ and distribution $\TypeDistr$ of user values, then $(s_\mi^\onCG, s_{\usr, 1}^\onCG, \dots, s_{\usr, n}^\onCG)$ is a Bayes-Nash equilibrium in $\onCG$ (considering both users and the miner as strategic agents).
\end{corollary}

\begin{remark}[On-Chain Simplicity vs. Previous Desiderata]
\label{rem:on-chain-vs-previous}
We remark that the ideas behind on-chain user and on-chain miner simplicity are essentially the same as the ideas behind UIC and MMIC \cite{Roughgarden20}, respectively.
However, there are two important differences between our definition of on-chain simplicity and these previously-considered definitions in TFM design.

First, we use a more-expansive model of TFM protocols: our model of TFM protocols allows for the miner to submit an advice. 
Thus, our notion of a compliant miner strategy must be more permissive than previous ones, which only consider a single unique miner strategy as ``following the protocol'' (namely, divulging all users' bids without fabricating any of her own).

Second, following \cite{AkbarpourL20}, we adopt different semantics in that our definition posits that a \emph{particular} strategy profile $(s_\mi^\onCG, s_\usr^\onCG)$ is on-chain simple, where prior definitions posit that the mechanism is `simple' if some such equilibrium exists.
We posit that this approach is more clear when certain properties (such as being DSIC) might hold in some equilibria of the mechanism and might not hold in others, or when the equilibrium strategies are not obvious from the mechanism.
Relatedly, our definitions specify that the strategy profile is simple for a particular distribution of user values.
\end{remark}

\subsection{Off-Chain Influence Proofness}
\label{sec:Definitions-OCIP}

We now proceed to define the main novel simplicity condition of our paper, off-chain influence proofness.
Off-chain influence proofness is a property of a strategy profile $(\Moff, s_{\usr, 1}^{\offCG, \Moff}, \dots, s_{\usr, n}^{\offCG, \Moff})$ in the off-chain game $\offCG$,
and specifies that (i) the miner and users stay entirely on-chain,
and (ii) the miner cannot profit even by conducting any alternative off-chain mechanism.
Formally:

\begin{definition}[Off-Chain Influence Proofness] \label{def:OffChainInfluenceProof}
    For some off-chain game $\offCG$ and distribution $\TypeDistr$,
    a strategy profile $(\Moff, s_{\usr, 1}^{\offCG, \Moff}, \dots, s_{\usr, n}^{\offCG, \Moff})$ 
    is \emph{off-chain influence proof} if:
    \begin{enumerate}
        \item \label{item:OCIP-on-chain}
        $\Moff$ is a trivial off-chain mechanism (i.e., $\Moff$ has no off-chain component; see \autoref{def:model}, \autoref{item:model-offCG}).
        Further, $(\Moff, s_{\usr, 1}^{\offCG, \Moff}, \dots, s_{\usr, n}^{\offCG, \Moff})$ is a user BNE in the off-chain game $\offCG$.
        
        \item $\Moff$ is revenue optimal over all off-chain mechanisms that the miner can commit to, and any corresponding user Bayes-Nash equilibrium, conditioned on using $\bBuild$ to finalize the block. 
        Formally, for any off-chain mechanism $\widetilde{\mathcal{M}}_{\mathsf{off}}$ and any user Bayes-Nash equilibrium strategy profile $(\widetilde{\mathcal{M}}_{\mathsf{off}}, \widetilde{s}_{\usr, 1}^{\offCG, \widetilde{\mathcal{M}}_{\mathsf{off}}}, \dots, \widetilde{s}_{\usr, n}^{\offCG, \widetilde{\mathcal{M}}_{\mathsf{off}}})$ for users in $\widetilde{\mathcal{M}}_{\mathsf{off}}$, we have
        \[
            \Es{v \sim \TypeDistr}{\rev^\offCG(\Moff; s_\usr^{\offCG, \Moff}(v))}
            \ge
            \Es{v \sim \TypeDistr}{\rev^\offCG(\widetilde{\mathcal{M}}_{\mathsf{off}}; \widetilde{s}_\usr^{\offCG, \widetilde{\mathcal{M}}_{\mathsf{off}}}(v))}.
        \]
    \end{enumerate}
\end{definition}

Since \autoref{item:OCIP-on-chain} of this definition states that the miner and users stay entirely on-chain, we can also view on-chain influence proofness as a property of a user BNE $(s_\mi^\onCG, s^\onCG_{\usr,1}, \ldots, s^\onCG_{\usr,n})$ in the on-chain game $\onCG$.
Throughout this paper, we refer to such a strategy profile in $\onCG$ as off-chain influence proof if the corresponding off-chain profile $(\Moff^{s_\mi^\onCG}, \bot(s^\onCG_{\usr,1}), \ldots, \bot(s^\onCG_{\usr,n}))$ satisfies off-chain influence proofness.

From the perspective of a revenue-maximizing miner,
\autoref{def:OffChainInfluenceProof} gives a very strong property of the user BNE $\sigma^\onCG = (s_\mi^\onCG, s^\onCG{\usr,1}, \ldots, s^\onCG{\usr,n})$.
Namely, even though the miner stays completely on-chain in $\sigma^\onCG$,
in \emph{any} other off-chain mechanism, and \emph{any} user Bayes-Nash equilibrium induced by the other off-chain mechanisms,
the revenue cannot increase beyond that of
$(s_\mi^\onCG, s^\onCG_{\usr,1}, \ldots, s^\onCG_{\usr,n})$.
Thus, the miner cannot profit by deviating from $\sigma^\onCG$, and this is true in a much stronger sense than typically considered under most formal notions of an equilibrium.
In \autoref{sec:weaker-influence-proof}, we examine several different notions of influence-proofness, and observe that off-chain influence-proofness gives the strongest guarantee among all that we consider.\footnote{
    For example, we consider definitions which allow the miner to influence the selected user BNE without using an off-chain mechanism,
    and we consider whether the miner needs commitment power or not.

    Of all previous definitions of which we are aware, off-chain influence proofness is most connected to the notion of a subgame perfect equilibrium.
    Specifically, an off-chain influence proof strategy profile is a subgame perfect equilibrium of the off-chain game, in which the miner moves first in order to commit to the mechanism $\Moff$, and all users play simultaneous after learning $\Moff$. 
}

\subsection{Strong Collusion Proofness} \label{sec:ColProof}

Lastly, we introduce the notion of strong collusion proofness, which stipulates that the miner and users cannot jointly increase their utility by colluding.
This definition is close in spirit to side-contract proofness (SCP) as defined in \citep{ChungS23}.
However, similar to \autoref{rem:on-chain-vs-previous}, we make two modifications of this definition to fit our framework.
First, we consider arbitrary user Bayes-Nash equilibria (not just equilibria where users truthfully bid their values).
Second, we allow the miner to increase the cartel's utility by manipulating her advice to the on-chain game $\onCG$.

\begin{definition}[1-1-Strong Collusion Proofness]
\label{def:strong-collusion-proof}
    Let $\onCG$ be the on-chain game induced by some block building process, and consider some distribution $\TypeDistr$ of user values.
    The strategy profile $(s_\mi^\onCG, s_{\usr, 1}^\onCG, \dots, s_{\usr, n}^\onCG)$ in user Bayes-Nash equilibrium is \emph{1-1-strong collusion proof}\footnote{
        For $a \in \N$ and $b \in \{0, 1\}$, we call a mechanism to be $a$-$b$-strong collusion proof if it is resistant to collusion when the colluding cartel has $a$ users and $b$ miners. Given this notation, it is not hard to extend the above definition of 1-1-strong collusion proofness to cartels of larger sizes.
        \label{fn:abColProof}} 
    (or \emph{strong collusion proof} for short) if for each user $i$ and value $v_i$ for user $i$, the miner and user $i$ should not be able to increase their joint utility by deviating from playing $s_\mi^\onCG$ and $s_{\usr, i}^\onCG(v_i)$ respectively.
    Formally, over all values $v_i$ and bids $b_i$ for user $i$ and all strategies $\widetilde{s}_\mi^\onCG$ for the miner,
        \begin{equation*}
            \notag
            \begin{split}
                &\Es{v_{-i} \sim \TypeDistr^{n-1}}{\rev^\onCG(s_\mi^\onCG; s_\usr^\onCG(v))
                \vphantom{\Big |}
                + v_i \cdot X_i^{\onCG}(s_\mi^\onCG; s_\usr^\onCG(v)) - P_i^{\onCG}(s_\mi^\onCG; s_\usr^\onCG(v))}  \\ 
                & \qquad \ge
                \Es{v_{-i} \sim \TypeDistr^{n-1}}{\rev^\onCG(\widetilde{s}_\mi^\onCG; b_i, s_{\usr, -i}^\onCG(v_{-i}))
                \vphantom{\Big |}
                + v_i \cdot X_i^{\onCG}(\widetilde{s}_\mi^\onCG; b_i, s_{\usr, -i}^\onCG(v_{-i})) - P_i^{\onCG}(\widetilde{s}_\mi^\onCG; b_i, s_{\usr, -i}^\onCG(v_{-i}))}       
            \end{split}
        \end{equation*}
\end{definition}

At a high level, being robust to collusion is a very desirable notion, especially in decentralized settings (where agents may be able to collude via smart contracts, and where centralization is a strong concern).
However, it turns out that achieving strong collusion proofness is impossible when combined with on-chain simplicity and off-chain influence proofness, as we demonstrate in \autoref{sec:impossibility}.
While a full investigation of different notions of collusion is beyond the scope of our paper,
we include here a brief discussion of why notions of collusion proofness similar to \autoref{def:strong-collusion-proof} may be too strong, and why blockchains might be better served by adopting off-chain influence proof TFMs instead.

Most significantly, we observe that \autoref{def:strong-collusion-proof} assumes that all agents in the cartel are not strategic with each other.
This assumption is certainly appropriate if the miner herself is submitting a bid to the TFM.
However, such an assumption is necessarily not appropriate when agents do not fully incorporate into a single entity.
We posit that even while colluding, when the miner and users are distinct parties, 
they would be colluding according to a formally specified contract or profit-sharing function, and would be playing strategies which are in equilibrium with each other.\footnote{
    See \autoref{sec:AlternateCollusion} for an example of this idea in a posted price mechanism.
}
Thus, even though \autoref{def:strong-collusion-proof} captures relevant concerns of collusion for some settings, it may be much stronger than needed in many cases.

In \autoref{sec:AlternateCollusion}, we investigate weaker definitions of collusion resistance which capture the above idea of colluding agents (a.k.a the ``cartel''), being strategic with each other.
In \autoref{sec:TrustlessColProof}, we observe that in fact, off-chain influence proofness directly implies \emph{trustless collusion proofness}, a weak form of collusion resistance where all users play in equilibrium according to a proposed contract of collusion.
This models collusion where all agents (including the ones not in the cartel) learn about the collusion.
In \autoref{sec:WeakColProof}, we define weak collusion proofness, a stronger (but still weaker than \autoref{def:strong-collusion-proof}) model of collusion in which the cartel best-responds to a proposed contract, but agents not in the cartel are oblivious to the contract, and thus cannot react to the cartel deviating strategically.
The cryptographic $(k+1)\textsuperscript{th}$-price auction (\autoref{def:C2PA}), for instance, satisfies weak collusion proofness (\autoref{thm:C2PAColPos}), but not strong collusion proofness (\autoref{thm:C2PACol}).
We believe that these observations strengthen our proposal to consider off-chain influence proof %
TFMs such as the cryptographic $(k+1)\textsuperscript{th}$-price auction, even though it is not strongly collusion proof.
we leave a more thorough investigation of these and other notions of collusion as a promising direction for future work.

\section{Evaluating Canonical Mechanisms}
\label{sec:specific-cases}

In this section, we consider various mechanisms, and ask whether they satisfy on-chain simplicity, off-chain influence proofness and strong collusion proofness.
As in \autoref{sec:Definitions}, we assume user values are drawn i.i.d from a regular distribution $\TypeDistr$.

\subsection{EIP-1559} 
\label{sec:EIP1559}

To begin, we review the definition of EIP-1559 \citep{Buterin20, Roughgarden21}---the transaction fee mechanism employed by the Ethereum blockchain---for blocks with an unlimited capacity.
In this idealized setting, EIP-1559 essentially becomes a posted price mechanism with an exogenously-determined price, where all payments are burnt.
We adopt this simplified version of EIP-1559 for our analysis.

\begin{definition}[EIP-1559 with Unlimited Supply] \label{def:EIP1559}
    For an exogenously given price $p$, the block-building algorithm $\bBuild$ of EIP-1559 is as follows:
    \begin{itemize}
        \item The mechanism takes no advice from the miner. Formally, this means that the miner advice is always some fixed $a_\mi = \emptyset$, i.e. $\mathsf{Adv}_{\mi} = \{ \emptyset \}$.
        
        \item $\mathsf{Bid}_{\usr} = \R_{\ge 0}$. Bids consist of a non-negative real number each. Bids are submitted in the plaintext model.
        
        \item All bids greater than the price $p$ are included in the block. For a set of bids $b$,
        \begin{equation*}
            \notag
            \begin{split}
                X_i(\emptyset ;b_i, b_{-i}) = \begin{cases}
                    1 & \text{if } b_i \geq p \\
                    0 & \text{otherwise}
                \end{cases}
            \end{split}
        \end{equation*}
        \item All included bids pay the price $p$ posted by the mechanism. In other words, for a set of bid $b$ submitted to the mechanism,
        \begin{equation*}
            \notag
            \begin{split}
                P_i(\emptyset ;b_i, b_{-i}) = \begin{cases}
                    p & \text{if } b_i \geq p \\
                    0 & \text{otherwise}
                \end{cases}
            \end{split}
        \end{equation*}
        \item Finally, all payments transferred by the users are burnt. The miner receives zero revenue (or equivalently, any constant block reward).
        \[
        \mathsf{Rev}(\emptyset; b) = 0
        \]
    \end{itemize}
\end{definition}

\begin{remark}
    The above definition is a simplified version of EIP-1559, specific to the unlimited supply setting, which captures all necessary components required for its analysis.
    The original proposal and analysis of EIP-1559 \citep{Buterin20, Roughgarden20} differs from \autoref{def:EIP1559} as follows.
    The bid space $\mathsf{Bid}_{\usr} = \R_{\ge 0}\times \R_{\ge 0}$
    consists of two real numbers $(\delta_i, c_i)$.
    The first coordinate $\delta_i$ of the bid $b_i$ is the \emph{tip} to be transferred to the miner, while the second coordinate $c_i$ is the \emph{fee cap}.
    Bid $(\delta_i, c_i)$ indicates that user $i$ will pay $p + \delta_i$, where $p$ is burnt and $\delta_i$ is the tip to the miner, up to a maximum of a total payment of $c_i$.
    For a given price $p$, any bid $i$ with a fee cap $c_i \ge p$ can be included (as long as the block capacity is not exceeded, in case of a finite capacity).
    User $i$ pays $\min \{p + \delta_i, c_i\}$ if included, and $0$ otherwise.
    The miner receives a revenue of $\sum_i \min \{\delta_i, c_i - p\}$, where the sum is over all allocated users $i$ (plus a constant block reward).
    
    While the fee cap and tips are important in case of blocks with a finite supply, the case of an unlimited supply (or a supply much larger than the demand) essentially reduces to \autoref{def:EIP1559}. 
    In practice, we indeed often observe users including only very small tips and essentially just choosing if they are willing to pay the current posted price, particularly in periods of low demand.
    We refer the reader to \citet{Roughgarden20, Roughgarden21} for a thorough description of tips and fee caps.    
    Note that the fact that we work in the simplified infinite supply case makes our negative result for EIP-1559 only stronger.\footnote{This is in contrast to the results of \citet{ChungS23}, who argue that EIP-1559 is the gold standard when supplemented with a pricing oracle which ensures that the demand is always smaller than the supply.}
\end{remark}

For the remainder of this subsection, we analyze the simplicity notions satisfied by the strategy profile $\sigma^\onCG = (s_{\mi}^{\mathsf{comp}}(\emptyset), v_1, \dots, v_n)$ in the on-chain game, where $s_{\mi}^{\mathsf{comp}}(\emptyset)$ denotes the compliant miner strategy that specifies playing $(\emptyset, H, \emptyset)$ for a set $H$ of submitted bids (i.e, the miner does not submit an advice, does not censor any bids and does not fabricate bids),
and $v_i$ denotes user $i$ bidding truthfully (\autoref{fn:StrategyNotationAbuse}).
We argue that $\sigma^\onCG$ satisfies on-chain simplicity and strong collusion proofness, but not off-chain influence proofness.

\begin{theorem}[Adaptation of \cite{Roughgarden20}] \label{thm:EIP1559Rou}
    The strategy profile $\sigma^\onCG$ in the on-chain game $\onCG$ induced by EIP-1559 for any distribution $\TypeDistr$ of user values is on-chain simple and strong collusion proof.
\end{theorem}

Our arguments are similar to the proof of the analogous claims in \citet{Roughgarden20}. For completeness, we prove \autoref{thm:EIP1559Rou} formally in \autoref{sec:ExcEIP1559}. We sketch the proof below.

To begin, note that on-chain user simplicity of $\sigma^\onCG$ (and the fact that $\sigma^\onCG$ is a user BNE)  follows immediately since a posted-price mechanism is DSIC for users.
On-chain miner simplicity follows from the fact that the miner can only decrease her revenue by getting some of her fabricated bids included and burning her own payments in the process.
Thus, she maximizes her revenue by playing the compliant strategy.
Finally, we verify strong collusion proofness.
Since the miner gets constant revenue (unless she submits fabricated bids, which as argued above only leads to burning her own money), optimizing the joint utility of a coalition of the miner and a user reduces to optimizing the user's utility.
As argued in on-chain user simplicity, bidding truthfully optimizes the user's utility, so strong collusion proofness follows.

In contrast, EIP-1559 does not satisfy off-chain influence proofness.

\begin{theorem} \label{thm:EIPOffChain}
    There exists regular distributions $\TypeDistr$ of values such that the strategy profile $\sigma^\onCG = (s_{\mi}^{\mathsf{comp}}(\emptyset), v_1, \allowbreak \dots, \allowbreak v_n)$ is not off-chain influence proof in the off-chain game $\offCG$.
\end{theorem}
\begin{proof}
    Recall that the miners' on-chain revenue is zero (which is strategically equivalent to receiving a constant block reward); hence, it would suffice to construct any off-chain mechanism in which the miner must collect off-chain payments and receives positive expected revenue.
    In fact, we sketch the miner's optimal off-chain mechanism, and observe that it indeed (typically) gives the miner a positive expected revenue.

    Consider a general off-chain mechanism $\Moff$ and fix a user BNE under $\Moff$.
    For any profile of user values $v$ and any user $i$, let $p_i(v)$ be the expected (on-chain and off-chain) payment made by user $i$ under this user BNE of $\Moff$ in $\offCG$.
    By including a user $i$ with value $v_i$, the miner earns a virtual welfare $\phi(v_i)$ but expends $p$ towards burns.
    Thus, the miner's expected revenue equals
    \begin{align*}
    & \Es{v \sim \TypeDistr^n}{\sum_{i = 1}^n p_i(v) - p \cdot \1{\text{user $i$ is included when values are $v$}}}
    \\ 
    & \qquad  = 
    \Es{v \sim \TypeDistr^n}{\sum_{i = 1}^n \big(\phi(v_i) - p \big) \cdot \1{\text{user $i$ is included when values are $v$}}},
    \end{align*}
    where the equality follows from \autoref{thm:revenue-equals-virtual-welfare}.
    Now, due to the unlimited supply model, user $i$ can be included independently of the other users.
    Thus, the miner optimizes expected revenue by allocating to users with a virtual value larger than the price $p$.
    By regularity, this corresponds to allocating users with a value $v_i \geq \phi^{-1}(p)$.

    Thus, the miner's optimal off-chain mechanism can be described as follows.
    The miner posts a price $\phi^{-1}(p)$ off-chain.
    Users willing to get included pay $\phi^{-1}(p) - p$ to the miner and burn $p$ on-chain.
    All other users are censored.
    This off-chain mechanism gets the miner revenue strictly higher than any on-chain strategy whenever the value distribution $\TypeDistr$ has positive mass above $p$; hence, EIP-1559 is not off-chain influence proof.
\end{proof}

\subsection{The Cryptographic $(k+1)\textsuperscript{th}$-price Auction} 
\label{sec:C2PA}

Following EIP-1559, we discuss the \emph{Cryptographic} $(k+1)\textsuperscript{th}$ \emph{Price Auction} (abbreviated C$(k+2)$PA).
The C$(k+1)$PA implements the $(k+1)\textsuperscript{th}$-price auction in the miner-gatekeeper model, while allowing the miner to set a reserve price.

\begin{definition}[Cryptographic $(k+1)\textsuperscript{th}$-price Auction with Reserve] \label{def:C2PA}
    The block building process $\bBuild$ of the cryptographic $(k+1)\textsuperscript{th}$ auction is as follows:
    \begin{itemize}
        \item $\mathsf{Adv}_\mi = \R_{\ge 0}$. The mechanism is advised by the miner on the reserve price $a_\mi$.
        \item $\mathsf{Bid}_{\usr} = \R_{\ge 0}$, similar to typical sealed-bid auctions. Bids are submitted in the miner-gatekeeper model.
        \item All bids $b_i$ larger than the reserve price, up to $k$ bids are included in the block. In other words, for a set of bids $b$ and a reserve $a_\mi$,
        \begin{equation}
            \notag
            \begin{split}
                X_i(a_\mi; b_i, b_{-i}) =
                \begin{cases}
                    1 & \text{if } b_i \geq a_\mi \text{ and } |\{j \neq i | b_j > b_i\}| \leq k-1 \\
                    0 & \text{otherwise}
                \end{cases}
            \end{split}
        \end{equation}
        \item All included users pay the minimum bid needed to get included, i.e, the maximum over the reserve price and the $(k+1)\textsuperscript{th}$ largest bid. For a set of bids $b$ such that $b^{(k+1)}$ is the $(k+1)\textsuperscript{th}$ largest bid (set $b^{(k+1)} = 0$ if there are less than $k+1$ bids),
        \begin{equation}
            \notag
            \begin{split}
                P_i(a_\mi; b_i, b_{-i}) =
                \begin{cases}
                    \max \{b^{(k+1)}, a_\mi\} & \text{if } b_i \geq a_\mi \text{ and } |\{j \neq i | b_j > b_i\}| \leq k-1 \\
                    0 & \text{otherwise}
                \end{cases}                
            \end{split}
        \end{equation}
        \item All collected payments are transferred to the miner.
        \[
        \mathsf{Rev}(a_\mi; b_i, b_{-i}) = \sum_i P_i(a_\mi; b_i, b_{-i})
        \]
    \end{itemize}
\end{definition}

Consider the strategy profile $\sigma^\onCG = (s_{\mi}^{\mathsf{comp}}(a_\mi), v_1, \allowbreak \dots, \allowbreak v_n)$ in the on-chain game $\onCG$, where $s_{\mi}^{\mathsf{comp}}(a_\mi)$ specifies playing $(a_\mi, H, \emptyset)$ for each set of bids $H$, and where $a_\mi$ is the monopoly reserve for a regular distribution $\TypeDistr$ (as defined in \autoref{def:Regular}).
We will argue that $\sigma^\onCG$ is on-chain simple and off-chain influence proof, but typically does not satisfy strong collusion proofness.

\begin{theorem} \label{thm:C2PAPosSummary}
     For all regular distributions $\TypeDistr$ with a monopoly reserve $a_\mi$, the strategy profile $\sigma^\onCG = (s_{\mi}^{\mathsf{comp}}(a_\mi), v_1, \allowbreak \dots, \allowbreak v_n)$ in the on-chain game $\onCG$ induced by C$(k+1)$PA is on-chain simple and off-chain influence proof.
\end{theorem}

We sketch the proof of \autoref{thm:C2PAPosSummary}.
The $(k+1)\textsuperscript{th}$-price auction is DSIC for users (\citealp{Vickrey61, Myerson81}; we include a proof in \autoref{sec:C2PADSIC} for completeness).
Thus,  $\sigma^\onCG$ is a user BNE and on-chain user simple for all distributions $\TypeDistr$.

Next, we show that $\sigma^\onCG$ in C$(k+1)$PA is off-chain influence proof.
From \autoref{thm:Regular-RevOpt}, we know $(s_{\mi}^{\mathsf{comp}}(a_\mi), v_1, \allowbreak \dots, \allowbreak v_n)$ is the revenue optimal mechanism over all possible BNE for a regular distribution $\TypeDistr$.
Any off-chain mechanism announced by the miner is going to induce a user BNE in the off-chain game and cannot have a larger revenue than $\sigma^\onCG$.

We now show that $\sigma^\onCG$ in C$(k+1)$PA is on-chain miner simple.
Informally, this follows because the miner cannot see the value of submitted bids, and hence can only optimize revenue by implementing the ex-ante revenue optimal auction (as in $\sigma^\onCG$).
One might hope to formalize this by showing that for any possible miner strategy $\widetilde{s}^\onCG_\mi$, the users' strategies in $\sigma^\onCG$ still constitute a truthful user BNE, and hence \autoref{thm:Regular-RevOpt} shows that the miner's expected revenue when playing $\widetilde{s}^\onCG_\mi$ cannot be higher than it originally was in $\sigma^\onCG$.
However, it is not technically true that users always best-respond by truth-telling for all miner strategies.\footnote{
    For instance, if the miner plays the following on-chain strategy--- if an odd number of bids are submitted, run the C$(k+1)$PA with the monopoly reserve $a_\mi$; otherwise, when an even number of bids are submitted, censor all bids.
    Any user is incentivized to fabricate a fake bid equal to zero when an even number of bids are submitted. \label{fn:OneBid}
}

To work around the above technical issue,
consider a modification of our mechanism design environment in which each user must submit \emph{exactly} one bid; note that user strategies in $\sigma^\onCG$ can be naturally viewed as strategies in this modified environment.
Now, for every miner strategy $\widetilde{s}^\onCG_\mi$ and every user strategy in the modified environment, we observe that truthful bidding is a dominant strategy for the users within the modified environment.
Recall that the miner's action can only depend on the set of bidders who submit a bid (in the miner-gatekeeper model), and in the modified environment, we stipulate that all users submit exactly one bid.
Hence,  the miner's action is a constant.
Fixing the set of fabricated bids, and the set of bidders whose bids are forwarded to the block-building process, it is straightforward to see that truthful bidding is a dominant strategy for each user.
Thus, for every deviation $\widetilde{s}^\onCG_\mi$ the miner might consider, $\sigma^\onCG$ is still a user BNE when the miner plays $\widetilde{s}^\onCG_\mi$; 
hence, \autoref{thm:Regular-RevOpt} applies in the modified environment to show that the miner's expected utility cannot possibly increase under $\widetilde{s}^\onCG_\mi$.
In particular, no miner strategy $\widetilde{s}^\onCG_\mi$ gets them higher utility than in $\sigma^\onCG$ when users all play according to $\sigma^\onCG$.

However, $\sigma^\onCG$ in C$(k+1)$PA is not strongly collusion proof.

\begin{proposition} \label{thm:C2PACol}
    There exist regular distributions $\TypeDistr$ of user values such that the strategy profile $\sigma^\onCG = (s_{\mi}^{\mathsf{comp}}(a_\mi), v_1, \allowbreak \dots, \allowbreak v_n)$, for the monopoly reserve $a_\mi$, is not strongly collusion proof in the on-chain game $\onCG$ induced by C$(k+1)$PA.
\end{proposition}
\begin{proof}
    We describe the optimal collusion strategy between the miner and a user $i$ with value $v_i$, and observe that this optimal collusion strategy typically requires the miner and user $i$ to deviate away from $\sigma^\onCG$.
    In fact, our argument parallels a result for revenue optimality provided in \citet{BulowR89} for settings where the auctioneer has a non-zero value for the good.

    To state this argument, we treat the cartel containing the miner and a user $i$ to be a single entity, who we refer to as the \emph{grand auctioneer}.
    Fix the value $v_i$ of user $i$.
    Imagine that the grand auctioneer can always allocate user $i$, provided that at lease one of the $k$ slots is not filled by other bidders.
    Then the grand auctioneer incurs a reward $v_i$ equal to the value of user $i$ if all $k$ slots in the block are not filled. 
    
    As argued above for on-chain miner simplicity, any strategy played by the grand auctioneer will result in a DSIC on-chain game for all users apart from $i$ (formally, when constraining all users to submitting exactly one bid, similar to the discussion surrounding \autoref{fn:OneBid}).
    Thus, if the grand auctioneer induces an allocation rule $(\widetilde{X}_j)_{j\ne i}$ for each bidder other than $i$, the expected revenue is given by (applying \autoref{thm:payment-identity})
    \[
    \Es{v_{-i} \sim \TypeDistr^{n-1}}{
       v_i \cdot \1{ {\textstyle \sum_{j\ne i}} \widetilde{X}_j(v_{-i}) \le k} 
       + \sum_{j \neq i} \phi(v_j) \cdot  \1{\widetilde{X}_j(v_{-i})} }
    \]

    The grand auctioneer can optimize the above pointwise (i.e., for all $(v_j)_{j\ne i}$) by
    ensuring that user $i$ is included 
    whenever the $k$\textsuperscript{th}-largest bid (among users other than $i$) has a virtual value less than $v_i$ (or, if there are less than $k$ bids).
    This can be achieved by the colluding user $i$ placing a bid equal to $\phi^{-1}(v_i)$ (and the miner continuing to play $s_{\mi}^{\mathsf{comp}}(a_\mi)$).
    Under the above mentioned strategy profile, whenever the $k$\textsuperscript{th} largest bid has a value less than $\phi^{-1}(v_i)$, it will have a virtual value smaller than $v_i$ (by monotonicity) and will be outbid by user $i$; thereby optimizing the grand auctioneer's (and the cartel's) objective. 
    Indeed, for typical regular distributions, there will be many values $v_i$ such that $v_i \ne \phi^{-1}(v_i) > a_\mi $, and hence the utility of the grand auctioneer strictly increases via collusion.
\end{proof}

\paragraph{Revisiting the notions of collusion resistance.}

The Cryptographic $(k+1)\textsuperscript{th}$-Price Auction satisfies all the desiderata highlighted in \autoref{sec:Definitions}, except for strong collusion proofness.
We now revisit our arguments on the appropriateness of strong collusion proofness from \autoref{sec:ColProof}, and argue that
C$(k+1)$PA may be intuitively considered fairly resistant to collusion.
We argue that this holds for two main reasons.

First, we consider a notion of \emph{weak} collusion proofness, which says that when the miner colludes with some user $i$, user $i$ will participate \emph{strategically} in the specific contract of collusion / profit sharing proposed by the miner.
Intuitively, such a situation is equivalent to the miner running a modified mechanism which treats user $i$ differently, and (since the colluding user best-responds, and the mechanism is still DSIC for other users) users play a BNE of this modified mechanism.
Hence, the miner cannot increase her expected revenue in this situation.
We defer a full discussion to \autoref{sec:AlternateCollusion}, where
we formally define weak collusion proofness (\autoref{def:WeakColProof}). 
We also consider another variant of this definition (which we call trustless collusion proofness) in which users other than $i$ can strategically respond to the collusion between the miner and $i$ (\autoref{def:TrustlessColProof}).
We show C-$(k+1)$PA satisfies both properties in \autoref{thm:C2PAColPos} and \autoref{thm:OffCIP-TCP}.

Second, we consider situations where a user and a miner engage in ``strong collusion'' as in \autoref{def:strong-collusion-proof} (for instance, when the miner herself also submits transactions), and observe that (when the miner and user collude optimally) such collusion does not impact other users' best response.
In particular, we showed in the proof of \autoref{thm:C2PACol} that when the miner has value $v_\mi^{(i)}$ for being included, she would want to place a bid equal to $\phi^{-1}(v_\mi^{(i)})$, but otherwise implement a normal C$(k+1)$PA.
We observe that users not in the coalition best respond by bidding truthfully.
Thus, even when the miner engages in non-truthful behaviour that increases her revenue, other users do not want to deviate from bidding truthfully in the TFM.

\begin{remark}
    Ethereum burns money in EIP-1559 not just to ensure simple participation for users, but also to maintain a deflationary pressure on the cryptocurrency.
    The C$(k+1)$PA can be supplemented with an (exogenous) burn too, while still allowing the user to set a reserve on top of the burn.
    All of our analysis of the C$(k+1)$PA continues to hold with very minor modifications.
    In particular, there exists an optimal reserve that the miner can set such that truthful bidding continues to remain on-chain simple and off-chain collusion proof.
    Getting the exogenously set reserve price so that the demand does not exceed the capacity of the block is necessary for EIP-1559 to satisfy on-chain user simplicity.
    On the other hand, C$(k+1)$PA satisfies on-chain simplicity and off-chain influence proofness irrespective of the burn fee.
    The burn fee has to now satisfy one less constraint---the capacity meeting the demand.
\end{remark}

\subsection{The Winner-Pays-Bid Mechanism} \label{sec:FPA}

We now consider the winner-pays-bid mechanism, the generalization of the first-price auction to a block capacity $k$, which was the status quo in Ethereum before EIP-1559 and still the TFM deployed in Bitcoin.

The block-building algorithm $\bBuild$ of the winner-pays-bid mechanism is identical to the C$(k+1)$PA, except the payment rule---instead of paying the maximum amongst the reserve or the $(k+1)\textsuperscript{th}$ largest bid, users included in the block pay their bid. Further, the winner-pays-bid mechanism can be implemented in both, the plaintext and the miner-gatekeeper cryptographic models.

\begin{definition}[Winner-Pays-Bid Mechanism with Reserve] \label{def:FPA}
    The block building process $\bBuild$ of the winner-pays-bid mechanism is as follows:
    \begin{itemize}
        \item The miner sets a reserve price $a_\mi$ as part of its advice, i.e, $\mathsf{Adv}_\mi = \R_{\ge 0}$.
        \item $\mathsf{Bid}_{\usr} = \R_{\ge 0}$. Bids can be submitted in the plaintext model, or the miner-gatekeeper model.
        \item All bids $b_i$ larger than the reserve price, up to $k$ bids are included in the block. In other words, for a set of bids $b$ and a reserve $a_\mi$,
        \begin{equation}
            \notag
            \begin{split}
                X_i(a_\mi; b_i, b_{-i}) =
                \begin{cases}
                    1 & \text{if } b_i \geq a_\mi \text{ and } |\{j \neq i | b_j > b_i\}| \leq k-1 \\
                    0 & \text{otherwise}
                \end{cases}
            \end{split}
        \end{equation}
        \item All included users pay their bid. For a set of bids $b$,
        \begin{equation}
            \notag
            \begin{split}
                P_i(a_\mi; b_i, b_{-i}) =
                \begin{cases}
                    b_i & \text{if } b_i \geq a_\mi \text{ and } |\{j \neq i | b_j > b_i\}| \leq k-1 \\
                    0 & \text{otherwise}
                \end{cases}                
            \end{split}
        \end{equation}
        \item All collected payments are transferred to the miner.
        \[
        \mathsf{Rev}(a_\mi; b_i, b_{-i}) = \sum_i P_i(a_\mi; b_i, b_{-i})
        \]
    \end{itemize}
\end{definition}

We observe two similar looking, yet different user BNEs satisfying different set of desiderata.\footnote{
    The existence of the two user Bayes-Nash equilibria ${\sigma}^\onCG_{\mathsf{val}}$ and ${\sigma}^\onCG_{0}$ in the winner-pays-bid mechanism that satisfy two different set of simplicity assumptions motivate describing simplicity as a property of the equilibrium as opposed to being a property of the mechanism (see \autoref{rem:on-chain-vs-previous}).
    While the two equilibria discussed here for the Winner-Pays-Bid mechanism are similar, see \autoref{sec:BOMB} for an example of an auction which induces two very different looking equilibria satisfying two very different set of simplicity definitions. 
}
In both the user BNEs, we will not be considering truthful bidding by users since it is a strictly dominated strategy (indeed, bidding truthfully always results in a zero utility for users irrespective of whether they are included or excluded).
We consider equilibria where users shade their bids.

Given a reserve $a_\mi$, \citet{chawlaH13} show the existence of a strategy profile $s_\usr^{\onCG, \TypeDistr, a_\mi}$ which is a user BNE.
Further, they prove that $s_{\usr,i}^{\onCG, \TypeDistr, a_\mi}(v_i)$ is unique for $v_i \geq a_\mi$.
We will discuss two variants of the above strategy that differ on the bidding behaviour of users with value less than the reserve (who will get allocated probability zero under best response).
Users with value $v_i \geq a_\mi$ bid according to the strategy $s_\usr^{\onCG, \TypeDistr, a_\mi}$ in both the equilibria.
In the first equilibrium, which we denote $\sigma^\onCG_{\mathsf{val}}$, users with value less than $a_\mi$ bid their value.
In the second equilibrium, which we denote $\sigma^\onCG_{0}$, users with value less than $a_\mi$ place a bid equal to zero.
Let the above mentioned user strategies in $\sigma^\onCG_{\mathsf{val}}$ and $\sigma^\onCG_{0}$ be denoted $s_{\usr, \mathsf{val}}^{\onCG, \TypeDistr, a_\mi}$ and $s_{\usr, 0}^{\onCG, \TypeDistr, a_\mi}$ respectively. 
In both equilibria, the miner plays the compliant strategy with advice $a_\mi$.

\begin{proposition}\label{thm:FPAValueEquilibriumSummary}
    Fix a regular distribution $\TypeDistr$ of user values, and let $a_\mi$ be the monopoly reserve of $\TypeDistr$.
    Then, the user BNE $\sigma^\onCG_{\mathsf{val}} = (s_{\mi}^{\mathsf{comp}}(a_\mi), \allowbreak s_{\usr, \mathsf{val}}^{\onCG, \TypeDistr, a_\mi}, \allowbreak \allowbreak \dots, \allowbreak s_{\usr, \mathsf{val}}^{\onCG, \TypeDistr, a_\mi})$ is not on-chain user simple and on-chain miner simple and also not strong collusion proof. However, $\sigma^\onCG_{\mathsf{val}}$ satisfies off-chain influence proofness.
\end{proposition}
\begin{proof}
    Users are not bidding truthfully, and thus, $\sigma^\onCG_{\mathsf{val}}$ is not on-chain user simple.

    We argue off-chain influence proofness next.
    Similar to off-chain influence in C2PA, we will argue that the revenue optimal equilibrium is already induced on-chain.
    For the monopoly reserve $a_\mi$, observe that the highest $k$ bids, conditional on being larger than $a_\mi$ get allocated.
    Since all users play the same strategy $s_{\usr, \mathsf{val}}^{\onCG, \TypeDistr, a_\mi}$, the users with the $k$ largest values (conditional on being larger than the reserve $a_\mi$) are allocated.
    Thus, $\sigma^\onCG_{\mathsf{val}}$ has the same interim allocation rule as the truth-telling BNE in the $(k+1)$\textsuperscript{th}-price auction with reserve $a_\mi$ and by revenue equivalence (\autoref{thm:RevenueEquivalence}), has an identical expected revenue. 
    But, from \autoref{thm:Regular-RevOpt}, the truth-telling user BNE with reserve $a_\mi$ is revenue optimal over all possible user Bayes-Nash equilibria.
    Thus, $\sigma^\onCG_{\mathsf{val}}$ is a revenue optimal user BNE, and the miner will not want to induce a different equilibria via an off-chain mechanism.

    Next, we show the lack of strong collusion proofness and the lack of on-chain miner simplicity.
    In fact, these two statements will follow from the same arguments.
    Consider the miner deviating and setting the reserve to zero (but continues to not fabricate or censor bids).
    At a high-level, this can potentially include a user with value $v_i < a_\mi$.
    In detail, with a non-zero probability, all users are going to have a value less than $a_\mi$.
    In such scenarios, the miner gets a non-zero revenue when setting the reserve to zero, but a zero revenue when setting the reserve to $a_\mi$.
    However, since winners pay their bids, all users with values larger than the reserve continue to pay their bids.
    Thus, the miner gets a strictly larger revenue by setting the reserve to zero, showing that $\sigma^\onCG_{\mathsf{val}}$ is not on-chain miner simple.
    Moreover, a colluding user $i$ with a value $v_i < a_\mi$ always gets zero utility (either by paying $v_i$ and getting included, or by being excluded).
    Thus, by increasing her revenue, the miner also increases the joint utility of the cartel.
    Thus, $\sigma^\onCG_{\mathsf{val}}$ is also not strong collusion proof.
\end{proof}

\begin{proposition}\label{thm:FPAZeroEquilibriumSummary}
    Fix a regular distribution $\TypeDistr$ of user values, and let $a_\mi$ be the monopoly reserve of $\TypeDistr$.
    Then, the user BNE $\sigma^\onCG_{0} = (s_{\mi}^{\mathsf{comp}}(a_\mi), s_{\usr, 0}^{\onCG, \TypeDistr, a_\mi}, \allowbreak \dots, \allowbreak s_{\usr, 0}^{\onCG, \TypeDistr, a_\mi})$ is not on-chain user simple and not strong collusion proof. However, $\sigma^\onCG_{0}$ satisfies on-chain miner simplicity and off-chain influence proofness.
\end{proposition}
\begin{proof}
    The arguments regarding lack of on-chain user simplicity and off-chain influence proofness are identical to $\sigma^\onCG_{\mathsf{val}}$ (\autoref{thm:FPAValueEquilibriumSummary}).

    Proof of lack of strong collusion non-proofness is also similar, except the miner sets a reserve zero and the colluding user $i$ with a value $v_i < a_\mi$ bids his value.
    Following the above deviation allocates user $i$, for instance, when all other users have a value less than the reserve (which happens with a non-zero probability).
    Instead of getting a joint utility zero, the cartel gets a joint utility $v_i$.
    In other scenarios, the above mentioned deviation does not decrease the joint utility since included users pay their bids.
    User $i$ may get included, increasing the joint utility, or get excluded, which does not decrease the cartel's utility.

    Finally, we observe that $\sigma^\onCG_{0}$ is on-chain miner simple.
    By decreasing the reserve, the miner can potentially include users who would not have been included at a reserve $a_\mi$.
    However, all such users would be bidding $0$, and thus, does not increase the miner's revenue.
\end{proof}

\begin{remark}
    Interestingly, the notions of collusion resistance proposed in \citet{Roughgarden20} (off-chain agreement proof, OCA-proof) and \citet{ChungS23} (side contract proofness, SCP) both classify the winner-pays-bid mechanism as collusion resistant.
    However, we show above that the winner-pays-bid mechanism is not strongly collusion proof (i.e., our adaptation of OCA proofness and SCP).
    This difference comes from the fact that OCA-proofness and SCP both assume that users play truth-telling strategies.
    We view this as a further reason to consider simplicity properties of \emph{particular equilibria}, rather than the mechanism itself, as discussed in \autoref{rem:on-chain-vs-previous}.
    (As an aside, we note that by our \autoref{thm:OffCIP-TCP} in \autoref{sec:AlternateCollusion}, the winner-pays-bid mechanism satisfies a weaker form of collusion resistance which we call trustless collusion proofness (\autoref{def:TrustlessColProof}), albeit for a different reason than suggested by prior frameworks).
\end{remark}

\begin{remark}
    Note that the optimal shading strategy for users $s_\usr^{\onCG, \TypeDistr, a_\mi}(v_i)$ is dependent on the number of users participating in the mechanism.
    For instance, if $k=1$, $\TypeDistr$ equals $U[0, 1]$ and the reserve $a_\mi$ is set to zero, $s_\usr^{\onCG, \TypeDistr, a_\mi}(v_i) = \frac{(n-1)}{n} \cdot v_i$ \citep{Myerson81, chawlaH13} where $n$ is the number of participants.
    However, playing the optimal strategy might be unrealistic since users will not precisely know the number of participants in a decentralized environment.\footnote{
        On the other hand, users might have a belief on the number of participants in the auction. Informally speaking, users can thus shade their bids in best response to the distribution $\TypeDistr$ of other users' values and $\mathcal{N}$, the believed distribution on the number of participants.}
\end{remark}

In \autoref{sec:OtherEg}, we discuss various other standard mechanisms like the plaintext $(k+1)\textsuperscript{th}$-price auction and the posted-price mechanism.
We also study a variety of auctions like the bonus-on-matching bids auctions and the squared revenue second price auction so that the equilibria in our examples satisfy all possible combinations of our core definitions---on-chain user simplicity, on-chain miner simplicity and off-chain collusion proofness.

In all the mechanisms considered in \autoref{sec:specific-cases} and \autoref{sec:OtherEg}, we either assume the plaintext model with no cryptography or black-box the cryptographic tools and the associated implementation details in the miner-gatekeeper model.
In \autoref{sec:DRA}, we consider implementing the miner-gatekeeper model via the deferred-revelation auction \citep{FerreiraW20, ChitraFK23} which only makes use of commitment schemes and digital signatures.

\section{Impossibility Results} \label{sec:impossibility}

In this section, we present our main impossibility result: except for trivial cases, no TFM has an equilibrium which is simultaneously on-chain simple, off-chain influence proof, and collusion proof.
We note that, unlike the impossibility result proved by \citet{ChungS23}, who prove that no TFM can satisfy UIC and $1$-SCP for a finite block-capacity, 
our impossibility result holds even in the case of an infinite-supply.

At a high level, we first show that no equilibrium that satisfies on-chain simplicity and off-chain influence proofness can award the miner a non-constant revenue.
Second, we show that a miner awarded constant revenue on-chain can achieve a larger revenue by running a separate auction off-chain.

\subsection{On-Chain Simplicity and Strong Collusion Proofness Implies Constant Expected Revenue}

We begin by showing an impossibility result similar to \citet{ChungS23}, who show that any mechanism satisfying user incentive compatibility and side contract proofness will have to award the miner a zero revenue.
Analogously, we show that the expected revenue of a miner (expectation taken over the randomness of the mechanism and values of users) must be independent of the number of users participating in the mechanism if the mechanism is on-chain simple and collusion proof. 

\begin{lemma} \label{thm:expected-constant-rev}
    For a distribution $\TypeDistr$ of user values, consider a block-building process $\bBuild$ and a sequence of user Bayes-Nash equilibria $\Big((s_\mi^\onCG, s_{\usr, 1}^{\onCG, n}, \dots, s_{\usr, n}^{\onCG, n})\Big)_{n \in \N \cup \{0\}}$ in the on-chain game $\onCG$, where the strategy profile $(s_\mi^\onCG, s_{\usr, 1}^{\onCG, n}, \dots, s_{\usr, n}^{\onCG, n})$ corresponds to the equilibrium induced on-chain when $n$ users participate.
    Suppose for each $n \in \N \cup \{0\}$, $(s_\mi^\onCG, s_{\usr, 1}^{\onCG, n}, \dots, s_{\usr, n}^{\onCG, n})$ is on-chain (user and miner) simple, and strongly collusion-proof. Then:
    \begin{enumerate}
        \item In expectation over users' values and the randomness of the TFM, the miner's expected revenue is independent of the number of users participating in the TFM.
        \item The expected miner revenue is invariant to the number of users even after conditioning on the bid of some single user.
    \end{enumerate}
\end{lemma}
\begin{proof}
    Fix the distribution $\TypeDistr$ of user values.
    Consider the strategy profile $\sigma^{\onCG, n} = (s_\mi^\onCG, s_{\usr, 1}^{\onCG, n}, \dots, s_{\usr, n}^{\onCG, n})$.
    Let $x_i^{(n)}$ and $p_i^{(n)}$ be the interim allocation and payment rules (\autoref{def:Interim}) induced by $\sigma^{\onCG, n}$ for user $i$.

    First, for each user $i$, we prove that the miner revenue, in expectation over the values of all other users, is independent of user $i$'s value.
    By on-chain user simplicity, we have that the induced mechanism is DSIC (and therefore BIC) and thus,
    \begin{equation} \label{eqn:user-bic-payment-id}
    p_i^{(n)}(v_i) = \int_0^{v_i} z x_i^{(n)\prime}(z) \,dz + c
    \end{equation}
    for some constant $c$ (by the payment identity, \autoref{thm:payment-identity}).

    Let $\mathsf{rev}^{(n)}(v_i) = \Es{v_{-i} \sim \TypeDistr^{n-1}}{\rev^{\onCG}(s_\mi^\onCG; v_i, v_{-i})}$ be the miner's ``interim'' revenue when user $i$ bids $v_i$. 
    We now apply collusion-proofness on the cartel containing the miner and the user $i$.
    Fixing the miner's strategy $s_\mi^\onCG$ but varying user $i$'s bid $b_i$, strong collusion-proofness means that the cartel optimizes its joint utility when user $i$ bids his value $v_i$:
    \begin{equation} \label{eqn:CartelUtil}
        v_i \cdot x_i^{(n)}(v_i) - p_i^{(n)}(v_i) + \mathsf{rev}^{(n)}(v_i)
    \ge
    v_i \cdot x_i^{(n)}(\widetilde{v}_i) - p_i^{(n)}(\widetilde{v}_i) + \mathsf{rev}^{(n)}(\widetilde{v}_i)
    \end{equation}
    for all possible bids $\widetilde{v}_i$ that user $i$ can submit.
    Define the expected net \emph{expenditure} of the cartel, $\mathsf{exp}^{(n)}(v_i) =  p_i^{(n)}(v_i) - \mathsf{rev}^{(n)}(v_i)$, to be the difference in money held by the cartel before and after the mechanism.
    Then, \autoref{eqn:CartelUtil} becomes
    \begin{equation} \label{eqn:Cartel-BIC-like-utility}
    v_i \cdot x_i^{(n)}(v_i) - \mathsf{exp}^{(n)}(v_i)
    \ge
    v_i \cdot x_i^{(n)}(\widetilde{v}_i) - \mathsf{exp}^{(n)}(\widetilde{v}_i).
    \end{equation}

    Now, the main observation towards our proof is the following. 
    If we pretend that user $i$ pays all of the expenditure $\mathsf{exp}^{(n)}$, then \emph{\autoref{eqn:Cartel-BIC-like-utility}  is precisely the definition of BIC}.
    Thus, the payment identity (\autoref{thm:payment-identity}) can be applied to the expenditure $\mathsf{exp}^{(n)}$ to get
    \[
    \mathsf{exp}^{(n)}(v_i) = \int_0^{v_i} z x_i^{(n)\prime}(z) \,dz + \widetilde{c}.
    \]
    for a constant $\widetilde{c}$, maybe different from $c$.
    Thus, applying \autoref{eqn:user-bic-payment-id} and the definition of 
     $\mathsf{exp}^{(n)}$, we get
    \begin{equation}
        \notag
        \begin{split}
            p_i^{(n)}(v_i) - C = \mathsf{exp}^{(n)}(v_i)
            = p_i^{(n)}(v_i) - \mathsf{rev}^{(n)}(v_i),
        \end{split}
    \end{equation}
    and thus, $\mathsf{rev}^{(n)}(v_i) = C$ for some constant $C$ and for all $v_i \in \R_{\ge 0}$.

    Next, we prove that the $\mathsf{rev}^{(n)}(0) = \mathsf{rev}^{(n-1)}(0)$.
    Since $\mathsf{rev}^{(n)}(v_i) = \mathsf{rev}^{(n)}(0)$ and $\mathsf{rev}^{(n-1)}(v_i) = \mathsf{rev}^{(n-1)}(0)$, this will prove the second claim in \autoref{thm:expected-constant-rev}.
    Intuitively, this follows due to on-chain miner simplicity (and strong-collusion proofness), since the miner (and the cartel, respectively) has no incentive to add a zero bid (or drop a zero bid, respectively).
    We now go through this argument fully.
    
    To begin, we argue $\mathsf{rev}^{(n)}(0) \ge \mathsf{rev}^{(n-1)}(0)$.
    Suppose that the miner is colluding with a user $i$ with value $0$.
    The cartel's expected joint utility when user $i$ bids truthfully equals $\mathsf{rev}^{(n)}(0) + 0 \times  x_i^{(n)}(0) - p_i^{(n)}(0) = \mathsf{rev}^{(n)}(0) - p_i^{(n)}(0)$.
    By individual rationality that follows due to on-chain user simplicity, user $i$ will make no payments by bidding zero. Thus, the joint utility of the cartel equals $\mathsf{rev}^{(n)}(0)$.
    Consider the deviation from $\sigma^{\onCG, n}$ where the miner censors user $i$.
    User $i$ continues to contribute nothing to the joint utility, while the miner's revenue becomes $\Es{v_{-i} \sim \TypeDistr^{n-1}}{\rev(s_\mi^\onCG; v_{-i})}$.
    Thus, by collusion proofness, we have
    \[
    \mathsf{rev}^{(n)}(0) \ge \Es{v_{-i} \sim \TypeDistr^{n-1}}{\rev(s_\mi^\onCG; v_{-i})}.
    \]
    On the other hand, we have proved $\mathsf{rev}^{(n-1)}(v_i)$ is a constant, and \[\Es{v_i \sim \TypeDistr}{\mathsf{rev}^{(n-1)}(v_i)} = \Es{v_{-i} \sim \TypeDistr^{n-1}}{\rev(s_\mi^\onCG; v_{-i})},\]
    by definition. Thus, $\mathsf{rev}^{(n-1)}(0) = \Es{v_{-i} \sim \TypeDistr^{n-1}}{\rev(s_\mi^\onCG; v_{-i})}$.
    Combining the above arguments, we have $\mathsf{rev}^{(n)}(0) \ge \mathsf{rev}^{(n-1)}(0)$.

    It remains to prove $\mathsf{rev}^{(n)}(0) \le \mathsf{rev}^{(n-1)}(0)$.
    As argued above $\mathsf{rev}^{(n-1)}(0) = \Es{v_{-i} \sim \TypeDistr^{n-1}}{\rev(s_\mi^\onCG; v_{-i})}$.
    Consider a miner considering fabricating a bid $\tilde{b} = 0$ with $n-1$ participating users.
    Prior to fabricating the bid, the miner received an expected revenue equal to $\Es{v_{-i} \sim \TypeDistr^{n-1}}{\rev(s_\mi^\onCG; v_{-i})}$.
    By fabricating $\tilde{b}$, the miner makes no payments (by on-chain user simplicity) but receives an expected revenue equal to $\mathsf{rev}^{(n)}(0)$.
    By on-chain miner simplicity,
    \[
    \mathsf{rev}^{(n)}(0) \le \Es{v_{-i} \sim \TypeDistr^{n-1}}{\rev(s_\mi^\onCG; v_{-i})} = \mathsf{rev}^{(n-1)}(0).
    \]

    The above concludes the proof that 
    $\mathsf{rev}^{(n)}(0)  = \mathsf{rev}^{(n-1)}(0)$, and hence the proof
    of the second claim in \autoref{thm:expected-constant-rev}. From here, the first claim is straightforward.
    \[
    \mathsf{rev}^{(n)}(v_i) = \mathsf{rev}^{(n-1)}(v_i)
    \]
    for all $v_i$, and taking expectation over $v_i$ on both sides, the proof follows.
\end{proof}

As a quick sanity check, remember that the strategy profile $(s_{\mi}^{\mathsf{comp}}(\emptyset), v_1, \dots, v_n)$ in EIP-1559 with unlimited supply satisfies on-chain simplicity and strong collusion proofness (\autoref{thm:EIP1559Rou}). As \autoref{thm:expected-constant-rev} claims, the miner gets a constant block reward, independent of the number of users.

\begin{remark}
    To be precise, \autoref{thm:expected-constant-rev} claims that, in expectation over users' bids, the miner's revenue is independent over the number of users.
    We believe %
    that the miner revenue must be independent of the bids submitted to the on-chain game too.
    More formally, we conjecture that if a sequence of user Bayes-Nash equilibria $\Big( \sigma^{\onCG, n}\Big)_{n \in \N \cup \{0\}}$ satisfy on-chain simplicity and collusion proofness, then the miner revenue must be a constant with probability 1 over the bids placed by users. %
    We leave this conjecture as an open problem for future work.
\end{remark}

\subsection{On-Chain Simplicity, Strong Collusion Proofness, and Off-Chain Influence Proofness Implies No Non-Trivial Mechanism}

Next, building on \autoref{thm:expected-constant-rev}, we show that if a TFM satisfies previously-considered desiderata as well as off-chain influence proofness, then the mechanism must be trivial (namely, the mechanism cannot award allocate any user with positive probability).
Intuitively, this follows because off-chain influence proofness claims that the miner is already receiving the optimal revenue over all user Bash-Nash equilibria that can be induced in the on-chain game $\onCG$ over the block-building process $\bBuild$, and this clashes with the claim of \autoref{thm:expected-constant-rev} that claims the miner revenue is independent of the number of users participating in the TFM.

Our proof proceeds in two steps.
First, we show that in a user BNE satisfying on-chain simplicity, off-chain influence proofness, and strong collusion proofness, all users must get zero utility.
Second, we prove that any mechanism where users are always awarded a zero utility to users must be trivial in in a formal sense.

\begin{lemma}
    For a given distribution $\TypeDistr$ of user values, and a sequence of on-chain simple and collusion proof user Bayes-Nash equilibria $\Big((s_\mi^\onCG, s_{\usr, 1}^{\onCG, n}, \dots, s_{\usr, n}^{\onCG, n})\Big)_{n \in \N \cup \{0\}}$ in the on-chain game $\onCG$ 
    Then, for all $n \in \N$ and all users $i$, the expected utility of user $i$ equals zero. Formally,
    \[
    \Es{v \sim \TypeDistr^n}{v_i \cdot x_i^{(n)}(v_i) - p_i^{(n)}(v_i)} = 0.
    \]
\end{lemma}
\begin{proof}
    Let $\mathsf{util}_i^{(n)}(v_i) = v_i \cdot x_i^{(n)}(v_i) - p_i^{(n)}(v_i)$ be user $i$'s interim utility function.
    Suppose there exists a user $i$ with an expected utility $\Es{v_i \sim \TypeDistr}{\mathsf{util}_i^{(n)}(v_i)} > 0$.
    The miner runs the following off-chain mechanism $\Moff$.
    The miner asks user $i$ (and only user $i$) to pay $(1/2)\cdot{\Es{v_i \sim \TypeDistr}{\mathsf{util}_i^{(n)}(v_i)}}$ off-chain, barring which user $i$ would be censored by the miner in the on-chain game.
    Irrespective of user $i$'s strategy, the miner plays $s_\mi^\onCG$ on-chain.

    From the perspective of all user apart from $i$, conditioned on the miner playing the strategy $s_\mi^\onCG$ on-chain, the resulting game is on-chain user simple.
    Thus, bidding $v_{-i}$ continues to be the best response in the off-chain game induced by the miner's off-chain mechanism $\Moff$.

    From user $i$'s perspective, $i$ receives an expected utility $\Es{v_i \sim \TypeDistr}{\mathsf{util}_i^{(n)}(v_i)}$, and thus, with positive probability, must be willing to pay $(1/2)\cdot{\Es{v_i \sim \TypeDistr}{\mathsf{util}_i^{(n)}(v_i)}}$ to be included in the block.
    Thus, $\Es{v \sim \TypeDistr^n}{p_{\mathsf{off}, i}(\Moff; v_i, v_{-i})} > 0$.

    From the miner's perspective, the miner receives the same on-chain revenue after announcing $\Moff$ as playing the trivial off-chain mechanism corresponding to $s_\mi^\onCG$.
    This is because, (since the user Bayes-Nash equilibria satisfy on-chain simplicity and collusion proofness) from the second part of \autoref{thm:expected-constant-rev},
    \[
    \Es{v_{-i} \sim \TypeDistr^{n-1}}{\rev(s_\mi^\onCG; v_i, v_{-i})} =     \Es{v \sim \TypeDistr^{n}}{\rev(s_\mi^\onCG; v_i, v_{-i})}
    \]
    and thus, in expectation over the values $v_{-i}$ of other users, the miner's revenue is independent of user $i$'s on-chain strategy.
    However, the miner also receives strictly positive revenue off-chain, and thus, prefers the equilibrium induced by $\Moff$.

    Thus, $(s_\mi^\onCG, s_{\usr, 1}^{\onCG, n}, \dots, s_{\usr, n}^{\onCG, n})$
    is not off-chain influence proof, as required.
\end{proof}

Finally, we argue that in general auction environments, if every user gets an expects a zero utility in a mechanism, then the mechanism allocates with a positive probability only if the user draws a value equal to the supremum of the distribution $\TypeDistr$.

\begin{lemma} \label{thm:PartConverse}
    Consider some DSIC and individually rational mechanism given by an allocation and payment rule $(X, P)$. Let users' values be distributed according to $\TypeDistr$ with a supremum $v_{\sup} \in \R \cup \{\infty\}$.\footnote{As is standard, the supremum of a distribution $\mathcal{T}$ equals $\inf \{s: Pr_{x \sim \mathcal{T}}(x \leq s) = 1 \}$}
    \begin{enumerate}
        \item If user $i$ receives a zero expected utility in $(X, P)$ (in expectation over the randomness in the mechanism and values of all users, including user $i$), i.e,
        \[
        \Es{v \sim \TypeDistr^n}{v_i \cdot x_i(v_i, v_{-i}) - p_i(v_i, v_{-i})} = 0,
        \]
        then $\Es{v_{-i} \sim \TypeDistr^{n-1}}{x_i(v_i, v_{-i})} = 0$ for all $v_i \neq v_{\sup}$.
        \item Suppose $\Es{v_{-i} \sim \TypeDistr^{n-1}}{x_i(v_i, v_{-i})} = 0$ for all $v_i \neq v_{\sup}$. Then, there exists a user that can receive a non-zero expected utility (in expectation over the outcomes of the mechanism and the values $v$ of all users including $i$) only if the following conditions are met:
        \begin{itemize}
            \item $\Ps{v_i \sim \TypeDistr}{v_i = v_{\sup}} > 0$ and $\Es{v_{-i} \sim \TypeDistr^{n-1}}{x_i(v_{\sup}, v_{-i})} > 0$
            \item The supremum of the distribution $\widetilde{\TypeDistr}$ given by $\TypeDistr$ conditioned on the draw being less than $v_{\sup}$. Then, the supremum $\tilde{v}_{\sup}$ of $\widetilde{\TypeDistr}$ must be strictly less than $v_{\sup}$.
        \end{itemize}
    \end{enumerate}
\end{lemma}
We will unpack the second claim before we prove the lemma.
The second claim is a partial converse to the first, and helps characterize all mechanisms where each user gets an expected utility equal to zero.
Consider the distribution $\TypeDistr$ given by
\[
\Ps{v_i \sim \TypeDistr}{v_i = 0} = 1/2
\qquad
\qquad
\Ps{v_i \sim \TypeDistr}{v_i = 1} = 1/2
\]
Conditioned on the draw from $\TypeDistr$ not being equal to $1$, its supremum, the resulting distribution is just the constant distribution at $0$, whose supremum is indeed less than $1$.
If a user with value distributed according to $\TypeDistr$ is posted a price $\frac{1}{2}$, he receives an expected positive utility.
The second part in the lemma claims a user can potentially receive a non-zero utility only if the supremum of the distribution $\TypeDistr$ has a non-zero probability of getting realized, and is not a limit point of the support of the distribution $\widetilde{\TypeDistr}$. 

\begin{proof}
    \begin{enumerate}
        \item Suppose a user $i$ with value $v_i$ is included in the block with a non-zero probability $x_i(v_i)$.
        By monotonicity of the allocation rule, whenever the user has a value $\tilde{v}_i > v_i$, he has to be included with at least the probability $x_i(v_i)$.
        Further, if the user pretends to have a value $v_i$, he receives a utility
        \[\tilde{v}_i \cdot x_i(v_i) - p_i(v_i) > v_i \cdot x_i(v_i) - p_i(v_i) \ge 0\]
        where $p_i$ is the interim payment rule for user $i$.
        Thus, user $i$ gets strictly positive utility whenever he has a value larger than $v_i$.
        If he receives an expected utility of zero, then, the probability of realizing a value larger than $v_i$ must be zero.
        \item Suppose that $\Es{v_{-i} \sim \TypeDistr^{n-1}}{x_i(v_{\sup}, v_{-i})} = 0$ for all $v_i \neq v_{\sup}$ and user $i$ receives a strictly positive expected utility (in expectation over the values $v$ of all users).
        Whenever $i$ has a value strictly smaller than $v_{\sup}$, he gets included with probability zero.
        By individual rationality, $i$ nets a zero utility.
        Thus, user $i$ receives a positive utility when his value $v_i = v_{\sup}$.
        For his expected utility to be positive, he needs to make an expected payment $p$ per unit probability of getting included, strictly smaller than $v_{\sup}$.

        Assuming $Pr_{v_i \sim \TypeDistr}(p < v_i < v_{\sup}) > 0$ for contradiction, then user $i$ can bid $v_{\sup}$ to make an expected payment $p$ per unit probability of getting included and receive a positive utility.
        This is a contradiction, since the mechanism is DSIC.
    \end{enumerate}
\end{proof}

From here, proving our main impossibility result becomes straightforward.

\begin{theorem} \label{thm:PriorImpossibility}
    For a given distribution $\TypeDistr$ of user values, and a sequence of on-chain simple, off-chain influence proof, and strongly collusion proof user Bayes-Nash equilibria $\Big((s_\mi^\onCG, s_{\usr, 1}^{\onCG, n}, \dots, \allowbreak s_{\usr, n}^{\onCG, n})\Big)_{n \in \N \cup \{0\}}$.
    Then, for all $n \in \N$ and all users $i$, user $i$ is allocated with probability zero whenever his value $v_i$ is smaller than the supremum $v_{\sup}$ of $\TypeDistr$.
\end{theorem}

The partial converse in \autoref{thm:PartConverse} also rules out any user being allocated with a positive probability when the supremum of the distribution does not have a point mass.
Of course, most commonly occurring continuous distributions like the uniform distribution, the normal distribution, the exponential distribution, etc., satisfy the above property.

\begin{theorem}
\label{thm:main-impossibility}
    There can exist no prior-independent block-building process $\bBuild$ such that for all distribution $\TypeDistr$, there exist strategy profiles that satisfy on-chain simplicity, off-chain influence proofness and collusion proofness and also includes users with a positive probability.
\end{theorem}

\begin{remark}
    We believe an impossibility result similar to \autoref{thm:main-impossibility} can be achieved by assuming only on-chain miner simplicity, strong collusion-proofness, and off-chain influence-proofness.
    While users need not bid truthfully, the corresponding interim allocation and payment rules (\autoref{def:ValueInterim}) are BIC.
    We believe that the BIC guarantee over the value allocation rule, coupled with on-chain miner simplicity, collusion proofness and off-chain influence proofness should lead to an impossibility result similar to \autoref{thm:main-impossibility}.
\end{remark}

\section{Conclusion}

We propose a set of desiderata for TFMs, and under these desiderata, users can be advised to simply bid their values, without worrying about either the bids placed by other users, or the strategy played by the miner, including any off-chain chatter.
If a user BNE in a TFM is off-chain influence proof, even a miner with unchecked commitment power would not want to induce a different equilibrium since she is receiving the optimal revenue over all user BNEs.
Moreover, analogous to UIC and MMIC, on-chain user simplicity guarantees a DSIC mechanism to users conditioned on the miner's \emph{expected} behaviour, and on-chain miner simplicity says the miner will not deviate on-chain from the expected on-chain behaviour.

We suggest that TFMs consider the cryptographic $(k+1)\textsuperscript{th}$ auction, which satisfies on-chain simplicity and off-chain influence proofness, but not strong collusion proofness. This is in contrast to the truth-telling user BNE in EIP-1559, which satisfies on-chain simplicity and strong collusion proofness, but fails off-chain influence proofness.
We prove that this tradeoff between strong collusion proofness and off-chain influence proofness is unavoidable, and
argue that off-chain influence proofness has innate collusion resistance guarantees which may suffice by themselves.
Still, even after our paper, the optimal TFM design should remain in consideration; for one, because the best explicit cryptographic or dynamic implementations of a second-price auction is not clear.
Moreover, in the era of ``MEV'', novel new design considerations for TFMs may emerge beyond the ones we consider here (see e.g. \cite{BahraniGR23}), and we hope that our framework provides relevant insights into future investigations.

\bibliographystyle{ACM-Reference-Format}

\bibliography{Bib}{}

\appendix

\section{A Zoo of Simplicity Definitions} \label{sec:Zoo}

\subsection{Weaker Notions of Influence Proofness}
\label{sec:weaker-influence-proof}

In this appendix, we explore weaker versions of off-chain influence proofness.
Our primary purpose is not to provide thorough investigations into any of these notions, but merely to emphasize that off-chain influence proofness gives a very strong garantee on the incentives of the miner: she cannot raise additional revenue, even if she is able to influence the way users behave in a wide variety of ways, so long as users play in a user BNE (consistent with the block-building process).

Recall that off-chain influence proofness stipulates that (i) the users and miner stay on-chain, while (ii) the miner gets her optimal revenue even when considering alternative miner-strategies that use nontrivial off-chain mechanisms.
To capture the different alternative miner strategies (i.e., ways the miner may try to ``game the system''), we focus on part (ii).
Thus, we phrase this subsection as giving various \emph{optimality} properties.
See \autoref{fig:deviating-sets-of-equilibria} for a list of different sets of strategy profiles, 
$\Sigma^\onCG_{\text{Simple-and-Equilibrium}} 
,\Sigma^\onCG_{\text{Equilibrium}} 
,\Sigma^\onCG_{\text{On-Chain}} 
,\Sigma^\offCG_{\text{Off-Chain}}
$.
The set of equilibria respectively correspond to: 
truthful user BNEs in which the miner's strategy is on-chain, compliant, and a best-response to the users' strategies; 
arbitrary user BNEs in which the miner's strategy is on-chain and a best-response to the users' strategies;
arbitrary user BNEs in which the miner's strategy is on-chain, but \emph{need not} be a best response;
and finally arbitrary user BNEs of arbitrary miner strategies in the off-chain game.
(Note that $\Sigma^\onCG_{\text{Simple-and-Equilibrium}}$ may be empty if all equilibria of $\onCG$, but all other sets can be shown to be nonempty by general results.)

We consider optimality conditions for each such set, as follows:

\begin{definition}
    For each $E \in \{ \text{Simple-and-Equilibrium}
,  \text{Equilibrium}
,  \text{On-Chain} 
,  \text{Off-Chain}
\}$,
    We say that a user-BNE $\sigma$ (in either on on-chain or off-chain game) is \emph{$E$-miner-optimal} if $\sigma\in \Sigma_E$, and moreover, the miner's revenue under $\sigma$ is at least as high as the miner's revenue in any other strategy profile in $\Sigma_E$.
\end{definition}

\begin{figure}
    \centering
    \begin{align*}
       \Sigma^\onCG_{\text{Simple-and-Equilibrium}} 
       & = \big\{ \sigma = (s_\mi^\onCG, s_{\usr,1}^\onCG, \ldots, s_{\usr,n}^\onCG)
       \\ & \qquad \big|\text{ $\sigma$ is a user-BNE, $s_\mi^\onCG$ is compliant and 
            is the miner's best-response  }
       \\ & \qquad \ \ \text{to $(s_{\usr,1},\ldots,s_{\usr,n})$, and $s_{\usr,i}$ is truth-telling and DSIC for each $i=1,\ldots,n$} \big\}
       \\
       & \rotatebox{-90}{ \scalebox{2}{$\subseteq$} }
       \\
       \Sigma^\onCG_{\text{Equilibrium}} 
       & = \big\{ \sigma = (s_\mi^\onCG, s_{\usr,1}^\onCG, \ldots, s_{\usr,n}^\onCG) 
       \\ & \qquad \big|\text{ $\sigma$ is a user-BNE and $s_\mi^\onCG$ is 
            the miner's best-response to $(s_{\usr,1},\ldots,s_{\usr,n})$ } \big\}
       \\
       & \rotatebox{-90}{ \scalebox{2}{$\subseteq$} }
       \\
       \Sigma^\onCG_{\text{On-Chain}} 
       & = \big\{ \sigma = (s_\mi^\onCG, s_{\usr,1}^\onCG, \ldots, s_{\usr,n}^\onCG) 
       \\ & \qquad \big|\text{ $\sigma$ is a user-BNE } \big\}
       \\
       & \rotatebox{-90}{ \scalebox{2}{$\subseteq$} }
       \\
       \Sigma^\offCG_{\text{Off-Chain}} 
       & = \big\{ \sigma = (\Moff, s_{\usr,1}^\onCG, \ldots, s_{\usr,n}^\onCG) 
       \\ & \qquad \big|\text{ $\sigma$ is a user-BNE and off-chain influence proof } \big\}
    \end{align*}
    \caption{Sets of strategies a miner might be able to influence the equilibrium to.}
    \label{fig:deviating-sets-of-equilibria}
\end{figure}

Observe that a strategy profile $\sigma$ is off-chain influence proof if and only if it is Off-Chain-miner-optimal, and also, the miner's strategy in $\sigma$ is completely on-chain.
As mentioned in \autoref{sec:Definitions}, 
we thus observe that off-chain influence proofness gives a stronger grantee than all of the other optimality notions we consider.
Moreover, we observe that when $\sigma$ is on-chain (miner and user) simple, and also off-chain influence proof, then the optimal miner revenue from $\Sigma^\onCG_{\text{Simple-and-Equilibrium}} \ni \sigma$ equals the optimal miner revenue from 
$\Sigma^\offCG_{\text{Off-Chain}}$.
We believe this gives compelling reasons that everyone will ``just follow $\sigma$''; perhaps more compelling than any previous set of desiderata.

\subsection{Alternate Notions of Collusion Resistance} \label{sec:AlternateCollusion}

In this appendix, we explore various (\emph{weaker}) notions of collusion resistance than the one given by \autoref{def:strong-collusion-proof}.
All the notions discussed below capture collusion where agents in the cartel continue to remain strategic with each other even while colluding.

To begin with, we sketch an example of a mechanism that is not collusion proof (via \autoref{def:strong-collusion-proof}) but how agents in the cartel being strategic with each other could make colluding difficult.
We consider a miner and a user colluding in a simple posted-price mechanism for a good with an unlimited supply, in which all payments are transferred to the miner (which we define formally in \autoref{sec:PostedPrice}).
We argue that truthful bidding in the posted-price mechanism is not collusion-proof.
For a price $p$ posted by the mechanism, the miner can collude with a user $i$ whose value $v_i$ is less than $p$, via the following ``contract'' of collusion.
The miner recommends the user to bid $p$, thereby procuring the good.
The miner keeps (say) $\frac{v_i}{2}$ from the on-chain payment made by the user and rebates the rest (i.e., transfers $\frac{v_i}{2}$).
The user gets included and gets a positive utility even though his value is less than the price, and the miner is able to achieve a non-zero revenue by transacting from user $i$.
Both the agents in the cartel are better off from colluding, which shows that this mechanism fails strong collusion proofness.

However, when colluding under the above (informally specified) contract of collusion, truthfully reporting their value is no longer an equilibrium for the users in the cartel. Colluding users would be incentivized to tell the miner that their value is very close to zero, even if their values are well above the price. This way, the miner would only dock a payment close to zero, which is far lesser than the price $p$ the user would have to pay if he were truthful.

\subsubsection{Trustless Collusion Proofness} \label{sec:TrustlessColProof}

For the following definition and the other alternate notion of collusion resistance in \autoref{sec:WeakColProof}, we find defining collusion resistance through the off-chain game much more natural than the on-chain game.\footnote{Similar to off-chain influence proof, both weaker notions of collusion resistance can be viewed partly as a property of an on-chain user BNE since we claim no cartel containing the miner and a user would want to communicate off-chain. We will, therefore, abuse notation to say a user BNE in the on-chain game $\onCG$ satisfies trustless collusion proofness and weak collusion proofness (defined in \autoref{sec:WeakColProof}).}

\begin{definition}[1-1-Trustless Collusion Proofness] \label{def:TrustlessColProof}
    For some off-chain game $\offCG$ and distribution $\TypeDistr$,
    a user BNE $(\Moff, s_{\usr, 1}^{\offCG, \Moff}, \dots, s_{\usr, n}^{\offCG, \Moff})$ where $\Moff$ is a trivial off-chain mechanism is 1-1-trustless collusion proof\footnote{Similar to collusion proofness (\autoref{fn:abColProof}), the above definition can be extended to $a$-$b$-trustless collusion proofness.} if, for any off-chain mechanism $\widetilde{\mathcal{M}}_{\mathsf{off}}$ (called the \emph{collusion contract}) and any user Bayes-Nash equilibrium $(\widetilde{\mathcal{M}}_{\mathsf{off}},\widetilde{s}_{\usr, 1}^{\offCG, \widetilde{\mathcal{M}}_{\mathsf{off}}}, \dots, \widetilde{s}_{\usr, n}^{\offCG, \widetilde{\mathcal{M}}_{\mathsf{off}}})$ in $\widetilde{\mathcal{M}}_{\mathsf{off}}$, at least one of the following conditions holds:
    \begin{enumerate}
        \item The miner prefers the the trivial mechanism $\Moff$ over the  collusion contract $\widetilde{\mathcal{M}}$.
        \[
            \Es{v \sim \TypeDistr}{\rev^\offCG(\Moff; s_\usr^{\offCG, \Moff}(v))}
            \ge
            \Es{v \sim \TypeDistr}{\rev^\offCG(\widetilde{\mathcal{M}}_{\mathsf{off}}; \widetilde{s}_\usr^{\offCG, \widetilde{\mathcal{M}}_{\mathsf{off}}}(v))}.
        \]
        \item No user $i$, irrespective of his value $v_i$, prefers the collusion contract $\widetilde{\mathcal{M}}$ over the trivial mechanism $\Moff$. Thus, for all users $i$,
        \begin{equation*}
            \notag
            \begin{split}
                &\Es{v_{-i} \sim \TypeDistr^{n-1}}{v_i \cdot X_i^{\offCG}(\Moff; s_\usr^{\offCG, \Moff}(v)) - P_i^{\offCG}(\Moff; s_\usr^{\offCG, \Moff}(v))} \\
                & \qquad \ge
                \Es{v_{-i} \sim \TypeDistr^{n-1}}{v_i \cdot X_i^{\offCG}(\widetilde{\mathcal{M}}_{\mathsf{off}};\widetilde{s}_{\usr}^{\offCG, \widetilde{\mathcal{M}}_{\mathsf{off}}}(v)) - P_i^{\offCG}(\widetilde{\mathcal{M}};\widetilde{s}_{\usr}^{\offCG, \widetilde{\mathcal{M}}_{\mathsf{off}}}(v))}.
            \end{split}
        \end{equation*}
    \end{enumerate}
\end{definition}
Summarizing the above definition, a user Bayes-Nash equilibrium is trustless collusion proof if signing a collusion contract leads to an equilibrium that is adverse either to the miner or to the user colluding with the miner.

Suppose that a strategy profile $(\Moff, s_{\usr, 1}^{\offCG, \Moff}, \dots, s_{\usr, n}^{\offCG, \Moff})$ in the off-chain game $\offCG$ for some distribution $\TypeDistr$ satisfies off-chain influence proof.
Then, the miner prefers $(\Moff, s_{\usr, 1}^{\offCG, \Moff}, \dots, s_{\usr, n}^{\offCG, \Moff})$ over any other user BNE induced by any alternate off-chain mechanism.
The first condition in the definition of trustless collusion proofness (\autoref{def:TrustlessColProof}) is automatically satisfied, leading to the following conclusion.

\begin{corollary} \label{thm:OffCIP-TCP}
    If $\sigma = (\Moff, s_{\usr, 1}^{\offCG, \Moff}, \dots, s_{\usr, n}^{\offCG, \Moff})$ is off-chain influence proof for some off-chain game $\offCG$ and distribution $\TypeDistr$ of values, then, $\sigma$ is also trustless collusion proof.
\end{corollary}

\subsubsection{Weak Collusion Proofness} \label{sec:WeakColProof}

In this section, we address the issue of secrecy with trustless collusion proofness and the natural collusion resistance provided by off-chain influence proofness.
Trustless collusion proofness models situations where users outside a cartel learn about the collusion and adjust their strategies to optimize their utilities in the subsequent off-chain game following the miner's announcement of the collusion contract.
However, if the cartel's formation is kept secret, assuming all users adjust their strategies based on the collusion contract might not be accurate.
To address this, we assume that agents outside the cartel cannot modify their actions based on the miner's announcement.
Note that we define weak collusion proofness only for user BNE induced by trivial off-chain mechanisms.
In this scenario, all users, particularly those outside the cartel, would play an abstaining bid.
An abstaining bid is a valid strategy for users outside the cartel in any off-chain mechanism declared by the miner.
Therefore, assuming users cannot modify their strategies based on the miner's announcement does not contradict the syntax of the framework.

\begin{definition}[1-1-Weak Collusion Proofness] \label{def:WeakColProof}
    Let $\offCG$ be the off-chain game induced by some block-building process for some distribution $\TypeDistr$ of user values. A strategy profile $(\Moff, s_{\usr, 1}^{\offCG, \Moff}, \allowbreak \dots, \allowbreak s_{\usr, n}^{\offCG, \Moff})$ in user BNE for a trivial off-chain mechanism $\Moff$ is 1-1-weak collusion proof (or weak collusion proof) if for all off-chain mechanisms $\widetilde{\mathcal{M}}_\mathsf{off}$ and  all users $i$, the following property holds. Consider any strategy $\widetilde{s}_{\usr, 1}^{\offCG, \widetilde{\mathcal{M}}_{\mathsf{off}}}$ for $i$ in $\widetilde{\mathcal{M}}_{\mathsf{off}}$ %
    which is a best response for user $i$ to $\widetilde{\mathcal{M}}_{\mathsf{off}}$ conditioned on other users being oblivious to the non-trivial off-chain component and playing the strategy $s_{\usr, -i}^{\offCG, \Moff}$.
    In other words, suppose that for all strategies $\hat{s}_{\usr, i}^{\offCG, \widetilde{\mathcal{M}}_{\mathsf{off}}}$ for user $i$, we have
    \begin{equation}
        \notag
        \begin{split}
            &\Es{v_{-i} \sim \TypeDistr^{n-1}}{v_i \cdot X_i^{\offCG}(\widetilde{\Moff}; \widetilde{s}_{\usr, i}^{\offCG, \widetilde{\mathcal{M}}_{\mathsf{off}}}(v_i), s_\usr^{\offCG, \Moff}(v_{-i})) - P_i^{\offCG}(\Moff; \widetilde{s}_{\usr, i}^{\offCG, \widetilde{\mathcal{M}}_{\mathsf{off}}}(v_i), s_\usr^{\offCG, \Moff}(v_{-i}))} \\
            & \qquad\ge \Es{v_{-i} \sim \TypeDistr^{n-1}}{v_i \cdot X_i^{\offCG}(\widetilde{\Moff}; \hat{s}_{\usr, i}^{\offCG, \widetilde{\mathcal{M}}_{\mathsf{off}}}(v_i), s_\usr^{\offCG, \Moff}(v_{-i})) - P_i^{\offCG}(\Moff; \hat{s}_{\usr, i}^{\offCG, \widetilde{\mathcal{M}}_{\mathsf{off}}}, s_\usr^{\offCG, \Moff}(v_{-i}))}.
        \end{split}
    \end{equation}
    Then, user $i$ and miner cannot both receive higher utility under the colluding strategies; formally, one of the following two properties must hold: 
    \begin{enumerate}
        \item the miner has a higher expected revenue in $\Moff$ than in $\widetilde{\mathcal{M}}_{\mathsf{off}}$, i.e,
        \[
            \Es{v \sim \TypeDistr}{\rev^\offCG(\Moff; s_\usr^{\offCG, \Moff}(v))}
            \ge
            \Es{v \sim \TypeDistr}{\rev^\offCG(\widetilde{\mathcal{M}}_{\mathsf{off}}; \widetilde{s}_{\usr, i}^{\offCG, \widetilde{\mathcal{M}}_{\mathsf{off}}}(v_i), s_{\usr, -i}^{\offCG, \Moff}(v_{-i}))},
        \]
        \item or, user $i$ has a higher expected utility in $\Moff$ than in $\widetilde{\mathcal{M}}_{\mathsf{off}}$, i.e,
        \begin{equation*}
            \notag
            \begin{split}
                &\Es{v_{-i} \sim \TypeDistr^{n-1}}{v_i \cdot X_i^{\offCG}(\Moff; s_\usr^{\offCG, \Moff}(v)) - P_i^{\offCG}(\Moff; s_\usr^{\offCG, \Moff}(v))} \\
                & \qquad \ge
                \Es{v_{-i} \sim \TypeDistr^{n-1}}{v_i \cdot X_i^{\offCG}(\widetilde{\mathcal{M}}_{\mathsf{off}};\widetilde{s}_{\usr, i}^{\offCG, \widetilde{\mathcal{M}}_{\mathsf{off}}}(v_i), s_{\usr, -i}^{\offCG, \Moff}(v_{-i})) - P_i^{\offCG}(\widetilde{\mathcal{M}};\widetilde{s}_{\usr, i}^{\offCG, \widetilde{\mathcal{M}}_{\mathsf{off}}}(v_i), s_{\usr, -i}^{\offCG, \Moff}(v_{-i}))}.
            \end{split}
        \end{equation*}
    \end{enumerate}
\end{definition}

Weak collusion proofness is similar to the notion of a semi-strong Nash equilibrium \citep{pennPT2005}.
A Nash equilibrium is semi-strong if fixing the strategies of all agents not in the coalition and for any strategy played by agents in the collusion, at least one of the following two conditions are satisfied:
\begin{enumerate}
    \item At least one of the agents in the coalition achieve a larger utility prior to colluding that after.
    \item If all agents in the coalition bag a larger utility, at least one colluding agent has an incentive to betray the cartel and deviate further, increasing their individual utility.
\end{enumerate}

To our knowledge,
weak collusion proofness is not a straightforward consequence of off-chain influence proofness.
While off-chain influence proofness compares the miner's revenue in two user BNEs, weak collusion proofness compares the miner's (and the colluding user's) revenue between a user BNE, and a deviation from the user BNE just for the miner and the colluding user.
Other users' strategy profiles need not be an equilibrium given the new strategies played by the miner and the colluding user.
On the other hand, strong collusion proofness implies weak collusion proofness.

Truthful bidding by users and the miner setting the monopoly reserve ($a_\mi$ such that the virtual value $\phi(a_\mi) = 0$) in the cryptographic $(k+1)\textsuperscript{th}$-price auction for a regular distribution $\TypeDistr$ is an example of a user Bayes-Nash equilibrium that is not strong collusion proof, but satisfies weak collusion proofness (\autoref{thm:C2PAColPos}).

\subsubsection{Example: Trustless and Weak Collusion Proofness in C(k+1)PA}

In this section, we explore the weaker collusion resistance guarantees in the truth-telling equilibrium in the C$(k+1)$PA.
Remember that $\sigma^\onCG$ is the user BNE corresponding to truthful bidding by users and the miner playing the compliant strategy corresponding to the monopoly reserve of the regular value distribution $\TypeDistr$.

We will argue that $\sigma^\onCG$ satisfies both trustless and weak collusion proofness.
Trustless collusion proofness is straightforward since $\sigma^\onCG$ satisfies off-chain influence proofness (\autoref{thm:OffCIP-TCP}).
Next, we argue that $\sigma^\onCG$ satisfies weak collusion proofness.
Consider a coalition between the miner and a user $i$. An analogous argument to on-chain miner simplicity (\autoref{thm:C2PAPosSummary}) would prove that truthful bidding for users $j \neq i$ is still $j$'s optimal response in a modified environment where all users apart from $i$ have to submit exactly one bid, regardless of the actions played by the miner and user $i$.
Since user $i$ knows about the collusion contract, user $i$ would play a best response to the miner's strategy and truthful bidding by users apart from $i$.
Thus, the strategy profile played by users is a user Bayes-Nash equilibrium in some mechanism (where the colluding user $i$ is treated differently from others).
The miner achieves the optimal revenue over all user Bayes-Nash equilibria by playing $s_{\mi}^{\mathsf{comp}}(a_\mi)$, where $a_\mi$ is the monopoly reserve of the regular value distribution $\TypeDistr$.
Thus, the miner cannot increase her expected revenue through collusion.

\begin{proposition} \label{thm:C2PAColPos}
$(s_{\mi}^{\mathsf{comp}}(a_\mi), v_1, \allowbreak \dots, \allowbreak v_n)$ is both trustless collusion proof and weak collusion proof for a regular distribution $\TypeDistr$ of values, where $a_\mi$ is the monopoly reserve.
\end{proposition}

\section{Separating the Simplicity Definitions through Examples} \label{sec:OtherEg}

In this section, we run over various standard mechanisms like the posted-price mechanism and the (plaintext) $(k+1)\textsuperscript{th}$-price auction, as well as other more artificial examples.
The primary focus of the examples discussed below
is to complete a separation of the three core definitions in \autoref{sec:Definitions}---on-chain user simplicity, on-chain miner simplicity, and off-chain influence proofness---i.e., show that the properties are logically independent.

See \autoref{tab:all-combinations-of-properties} for a summary of the separation.

\begin{table}[htb]
    \begingroup
    \renewcommand{\arraystretch}{1.5} 
    \begin{center}
    \begin{tabular}{cccc}
    \toprule
    & \ \ \multirow{2}{*}{\makecell{Off-Chain\\}} & \multicolumn{2}{c}{On-Chain} 
    \\
    \cmidrule{3-4}
    Mechanism &  \ \ \makecell{Influence \\Proof} & \makecell{User\\Simple} \ \ & \makecell{Miner\\Simple} \\
    \midrule
    C2PA  & \cmark & \cmark & \cmark \\
    EIP-1559  & \xmark & \cmark & \cmark \\
    P2PA  & \cmark & \cmark & \xmark \\
    BoMB, posted-price equilibrium & \cmark & \xmark & \cmark \\
    Winner-pays-bid mechanism, $\sigma^{\onCG}_{0}$ & \cmark & \xmark & \xmark \\
    SR2PA, second-price equilibrium & \xmark & \cmark & \xmark \\
    BoMB, winner-pays-bid equilibrium & \xmark & \xmark & \cmark \\
    SR2PA, first-price equilibrium & \xmark & \xmark & \xmark \\
    \bottomrule
    \end{tabular}
    \end{center}
    \endgroup
    
    \caption{Different auctions and equilibria satisfying all combinations of our core desiderata.}
    \label{tab:all-combinations-of-properties}

\begin{minipage}{\textwidth}
    \footnotesize
    \textbf{Notes:}
    The block-building processes of C2PA (\autoref{sec:C2PA}), P2PA (\autoref{sec:P2PA}), the winner-pays-bid mechanism (\autoref{sec:FPA}) and SR2PA (\autoref{sec:SR2PA}) allow the miner to set a reserve on-chain, while EIP-1559 (\autoref{sec:EIP1559}) and BoMB (\autoref{sec:BOMB}) admit no advice from the miner.
    In C2PA, P2PA, the winner-pays-bid mechanism and the second-price equilibrium on SR2PA (\autoref{sec:SR2PA2PA}) we consider the miner's compliant strategy of setting the optimal reserve. In both equilibria of BoMB and EIP-1559, we consider the miner playing the compliant strategy by giving no advice, while in the first-price equilibrium of SR2PA (\autoref{sec:SR2PA1PA}), we consider the miner fabricating a fake bid just below the highest bid.
    In C2PA, EIP-1559, P2PA and the second-price equilibrium of SR2PA we consider the truth-telling strategy of users. In the winner-pays-bid mechanism, the winner-pays-bid equilibrium in BoMB (\autoref{sec:BOMBWPB}) and SR2PA, we consider user shading their bid optimally according to the equilibrium in the winner-pays-bid mechanism from classical settings \citep{chawlaH13}. In the posted-price equilibrium in BoMB (\autoref{sec:BOMBPP}), users willing to get included bid the reserve price while the rest bid zero.
\end{minipage}

\end{table}

\subsection{The Plaintext $(k+1)\textsuperscript{th}$-price Auction} \label{sec:P2PA}

The \emph{plaintext} $(k+1)\textsuperscript{th}$ \emph{price auction} (P$(k+1)$PA) is identical to C$(k+1)$PA except that the bids are submitted in the plaintext model.

\begin{proposition} \label{thm:P2PASummary}
     For all regular distributions $\TypeDistr$ with a monopoly reserve $a_\mi$, the strategy profile $\sigma^\onCG = (s_{\mi}^{\mathsf{comp}}(a_\mi), v_1, \allowbreak \dots, \allowbreak v_n)$ in the on-chain game $\onCG$ induced by P$(k+1)$PA is on-chain user simple and off-chain influence proof.
     However, $\sigma^\onCG$ is not necessarily on-chain miner simple and collusion proof for all regular distributions $\TypeDistr$.
\end{proposition}
The proofs for on-chain user simplicity, off-chain influence proofness and lack of collusion proofness are identical to the analogous claims for C$(k+1)$PA.
Note that off-chain influence proofness does not depend on the cryptographic model.
Since the miner commits to her behaviour in the off-chain game and users best respond to the miner's commitment, a user BNE is always induced in the off-chain game.
Since $\sigma^\onCG$ is the revenue optimal equilibrium, the miner would not want to induce a different equilibrium in the off-chain game.

On-chain miner simplicity on the other hand depends on the cryptographic model.
By being able to look at the bids before playing her action, the miner can fabricate bids to increase her revenue \citep{AkbarpourL20}.
Given a set of bids $b$, the miner sets the reserve $\tilde{a}_\mi = \arg \max_{1 \le i \le k} i \cdot b^{(i)}$, where $b^{(i)}$ is the $i$\textsuperscript{th}-highest bid.
The miner gets a revenue equal to $\max_{1 \le i \le k} i \cdot b^{(i)} \ge k \cdot b^{(k)} \ge k \cdot b^{(k+1)}$, the miner's revenue if all $k$ slots get filled.
Even if only $j$ bids are larger than the monopoly reserve $a_\mi$ (i.e, $b^{(j)} \ge a_\mi \ge b^{(j+1)}$), $\max_{1 \le i \le k} i \cdot b^{(i)} \ge k \cdot b^{(k)} \ge j \cdot b^{(j)} \ge j \cdot a_\mi$.

Before we conclude the section, we showcase a different user Bayes-Nash equilibrium in the on-chain game when $k = 1$.
Suppose that the miner sets the reserve just below the largest bid and users best respond to the above miner's strategy.
The highest bid wins and since the reserve is (just below) the highest bid, the winner pays his bid.
Thus, the mechanism reduces to the first-price auction.
Thus, the user BNE equilibrium discussed in the winner-pays-bid mechanism (\autoref{sec:FPA}) is induced on-chain.

\subsection{The Posted-Price Mechanism} \label{sec:PostedPrice}

We discuss the posted-price mechanism for a block with unlimited supply.
The posted-price mechanism differs from EIP-1559 in two ways---the miner sets the price $a_\mi$ for inclusion in the block as part of her advice, and all payments collected from users are transferred to the miner, as opposed to getting burnt.
We discuss the posted-price mechanism in both the plaintext and the miner-gatekeeper model.

\begin{definition}[Posted-Price Mechanism with Unlimited Supply] \label{def:PostedPrice}
The block-building algorithm $\bBuild$ of the posted-price mechanism is as follows:
    \begin{itemize}
        \item The miner informs the mechanism of the price $a_\mi$ she wants to post. $\mathsf{Adv}_{\mi} = \R_{\ge 0}$.
        
        \item $\mathsf{Bid}_{\usr} = \R_{\ge 0}$. Bids consist of a non-negative real number each. The \emph{plaintext posted-price mechanism} elicits bids in the plaintext model, while the \emph{cryptographic posted-price mechanism} collects bids in the miner-gatekeeper model.
        
        \item All bids greater than the price $a_\mi$ are included in the block. For a set of bids $b$,
        \begin{equation*}
            \notag
            \begin{split}
                X_i(a_\mi ;b_i, b_{-i}) = \begin{cases}
                    1 & \text{if } b_i \geq a_\mi \\
                    0 & \text{otherwise}
                \end{cases}
            \end{split}
        \end{equation*}
        \item All included bids pay the price $a_\mi$ posted by the mechanism. In other words, for a set of bid $b$ submitted to the mechanism,
        \begin{equation*}
            \notag
            \begin{split}
                P_i(a_\mi ;b_i, b_{-i}) = \begin{cases}
                    p & \text{if } b_i \geq a_\mi \\
                    0 & \text{otherwise}
                \end{cases}
            \end{split}
        \end{equation*}
        \item All payments made by the users are transferred to the miner.
        \[
        \mathsf{Rev}(a_\mi; b) = \sum_i P_i(a_\mi ;b)
        \]
    \end{itemize}
\end{definition}

We will analyze the simplicity notions of the strategy profile $\sigma^\onCG = (s_{\mi}^{\mathsf{comp}}(a_\mi), v_1, \dots, v_n)$ for the monopoly reserve $a_\mi$ for regular distributions $\TypeDistr$ of user values in both the cryptographic models.
\subsubsection{The Cryptographic Posted-Price Mechanism} \label{sec:CPostedPrice}

\begin{proposition} \label{thm:CPostedSummary}
    The strategy profiles $\sigma^\onCG$ is on-chain simple and off-chain influence proof for all regular distributions $\TypeDistr$.
    However, there exist regular distributions $\TypeDistr$ such that $\sigma^\onCG$ is not strongly collusion proof.
\end{proposition}

We observe that the cryptographic posted-price mechanism is identical to the cryptographic $(k+1)\textsuperscript{th}$-price auction when $k = \infty$.
Thus, arguments similar to the proofs of \autoref{thm:C2PAPosSummary} will prove on-chain simplicity and off-chain influence proofness of the cryptographic posted-price mechanism.

The argument that $\sigma^\onCG$ is not strongly collusion proof is much simpler.
For a given price $a_\mi > 0$, consider a user $i$ with value $v_i < a_\mi$.
The miner and the user can increase their joint utility by the user bidding $a_\mi$ on-chain.
By not colluding, the joint utility of the cartel from $i$'s transaction equals zero, while the joint utility increases to $v_i$ if the user bids $a_\mi$.
The revenue from other users remains constant irrespective user $i$'s behaviour.
Thus, $\sigma^\onCG$ is not strongly collusion proof.

\subsubsection{The plaintext posted-price mechanism}

\begin{proposition} \label{thm:PPostedSummary}
    The strategy profiles $\sigma^\onCG$ is on-chain user simple and off-chain influence proof for all regular distributions $\TypeDistr$.
    However, there exist regular distributions $\TypeDistr$ such that $\sigma^\onCG$ is not on-chain miner simple and not collusion proof.
\end{proposition}

Similar to the connection between the cryptographic posted-price mechanism and C$(k+1)$PA, the plaintext posted-price mechanism is the plaintext $(k+1)\textsuperscript{th}$-price auction (\autoref{sec:P2PA}) for $k = \infty$.
Thus, $\sigma^\onCG$ is on-chain user simple and off-chain influence proof but not on-chain miner simple.
The proof that $\sigma^\onCG$ is not collusion proof is identical to the analogous claim for the cryptographic posted-price mechanism.

\subsection{Bonus-on-Matching-Bids (BoMB) Auction} \label{sec:BOMB}

The \emph{Bonus-on-Matching-Bids} (BoMB) auction (a variant of an auction originally proposed in \citet{chawlaH13}) allows two different looking equilibria--- one similar to the winner-pays-bid mechanism and the other similar to the posted-price mechanism.

At a high level, when the largest bid is unique, the BoMB auction is identical to the winner-pays-bid mechanism for $k = 1$ (a.k.a the first-price auction).
However, if the largest bid is not unique, all users that placed the largest bid get included.
Even though the mechanism is similar to the winner-pays-bid mechanism and all included users pay their bid, we argue that every user with a value greater than the reserve price would bid the reserve price in equilibrium, inducing an equilibrium similar to the posted-price mechanism with an unlimited supply.

\begin{definition}[Bonus-on-Matching-Bids Auction] \label{def:BOMB}
        The block building process $\bBuild$ of the bonus-on-matching-bids auction is as follows:
    \begin{itemize}
        \item $\bBuild$ takes no advice from the miner and thus, $\mathsf{Adv}_{\mi} = \{\emptyset\}$.
        We will assume that $\bBuild$ knows the monopoly reserve $\mathsf{r}$ of the regular distribution $\TypeDistr$ of user values.
        \item $\mathsf{Bid}_{\usr} = \R_{\ge 0}$. Bids are submitted in the miner-gatekeeper model.
        \item Allocate all users with the largest bid $b_{\max}$, if $b_{\max} \ge \mathsf{r}$ .
        Formally, for a set of bids $b$ such that $b_{\max} = \max_i b$ is the highest bid,
        \begin{equation}
            \notag
            \begin{split}
                X_i(\emptyset; b_i, b_{-i}) =
                \begin{cases}
                    1 & \text{if } b_i \geq \mathsf{r} \text{ and } b_i = b_{\max} \\
                    0 & \text{otherwise}
                \end{cases}
            \end{split}
        \end{equation}
        \item All included users pay their bid. For a set of bids $b$,
        \begin{equation}
            \notag
            \begin{split}
                P_i(\emptyset; b_i, b_{-i}) =
                \begin{cases}
                    b_i & \text{if } b_i \geq \mathsf{r} \text{ and } b_i = b_{\max} \\
                    0 & \text{otherwise}
                \end{cases}                
            \end{split}
        \end{equation}
        \item All collected payments are transferred to the miner.
        \[
        \mathsf{Rev}(\emptyset; b_i, b_{-i}) = \sum_i P_i(\emptyset; b_i, b_{-i})
        \]
    \end{itemize}
\end{definition}

\subsubsection{The Posted-Price Equilibrium} \label{sec:BOMBPP}
First, we consider user strategies that will induce an equilibrium similar to the posted-price mechanism with an unlimited supply.
Consider user $i$ playing the strategy $s_{\usr, \mathsf{pp}}^{\onCG, \mathsf{r}}$ given by
\[
s_{\usr, \mathsf{pp}}^{\onCG, \mathsf{r}}(v_i) = \begin{cases}
    \mathsf{r} & \text{if } v_i \geq \mathsf{r} \\
    0 & \text{otherwise}
\end{cases}
\]
We analyze the strategy profile $\sigma^{\onCG}_{\mathsf{pp}} = (s_{\mi}^{\mathsf{comp}}(\emptyset), s_{\usr, \mathsf{pp}}^{\onCG, \mathsf{r}}, \dots, s_{\usr, \mathsf{pp}}^{\onCG, \mathsf{r}})$.
Getting included would yield a negative utility to users with a value smaller than the reserve $\mathsf{r}$.
All users with value at least $\mathsf{r}$ get included in $\sigma^{\onCG}_{\mathsf{pp}}$.
Bidding higher would only result in a larger payment while bidding any lower than $\mathsf{r}$ will get them excluded.
Thus, $\sigma^{\onCG}_{\mathsf{pp}}$ is a user BNE.

Since users are not bidding truthfully, $\sigma^{\onCG}_{\mathsf{pp}}$ is not on-chain user simple. From the miner's viewpoint, fabricating a bid larger than $\mathsf{r}$ drive the miner revenue to zero, while censoring bids or fabricating bids below the reserve does not help the miner's revenue either.
Thus, $\sigma^{\onCG}_{\mathsf{pp}}$ is on-chain miner simple.

The proof of off-chain influence proofness is identical to the cryptographic posted-price mechanism, since the miner already receives the optimal revenue over all possible BNEs and mechanisms for an unlimited supply. A manipulation by a colluding cartel similar to the one described in the cryptographic posted-price mechanism (\autoref{sec:CPostedPrice}) shows $\sigma^{\onCG}_{\mathsf{pp}}$ is not strongly collusion proof.
The colluding user with value less than $\mathsf{r}$ can bid $\mathsf{r}$ instead of zero to increase the joint utility of the cartel.

\begin{proposition} \label{thm:BombPPSummary}
    The on-chain strategy profile $\sigma^{\onCG}_{\mathsf{pp}} = (s_{\mi}^{\mathsf{comp}}(\emptyset), s_{\usr, \mathsf{pp}}^{\onCG, \mathsf{r}}, \dots, s_{\usr, \mathsf{pp}}^{\onCG, \mathsf{r}})$ in $\onCG$ induced by the BoMB auction satisfies on-chain miner simplicity, off-chain influence proofness but not on-chain user simplicity or strong collusion proofness.
\end{proposition}

\subsubsection{The Winner-Pays-Bid Equilibrium} \label{sec:BOMBWPB}
In the section on first-price auctions (winner-pays-bid mechanisms for $k = 1$; \autoref{sec:FPA}), we discussed the existence of an equilibrium strategy for users $s_{\usr, \mathsf{val}}^{\onCG, \TypeDistr, \mathsf{r}}$, for a given number of users $n$ and the distribution $\TypeDistr$ with monopoly reserve $\mathsf{r}$.
We discuss properties satisfied by the strategy profile $\sigma^\onCG_{\mathsf{wpb}} = (s_{\mi}^{\mathsf{comp}}(\emptyset), s_\usr^{\onCG, \TypeDistr, \mathsf{r}}, \allowbreak \dots, \allowbreak s_\usr^{\onCG, \TypeDistr, \mathsf{r}})$ in the on-chain game of the BoMB auction.

We will argue that  $\sigma^\onCG_{\mathsf{wpb}}$ is a user BNE for most common regular distributions like $U[0, 1]$ which do not have point masses.
The strategy $s_{\usr, \mathsf{val}}^{\onCG, \TypeDistr, \mathsf{r}}$ is injective when the distribution $\TypeDistr$ has no point mass \citep{chawlaH13}. Thus, the probability of any arbitrary bid $b_i$ of user $i$ matching the bid $b_j$ equals zero (particularly in the miner-gatekeeper model, when users cannot see each other's bids).
Any deviation by users will fetch the user same utility as the deviation in the first-price auction.
Since $\sigma^\onCG_{\mathsf{wpb}}$ is a user BNE in the first-price auction with reserve $\mathsf{r}$, it is also a user BNE in the BoMB auction.

Arguing other properties are fairly straightforward.
$\sigma^\onCG_{\mathsf{wpb}}$ is not on-chain user simple, but satisfies on-chain user simplicity.
Censoring bids will only decrease the miner's revenue, while fabricating bids will not extract a larger payment from included users either.
$\sigma^\onCG_{\mathsf{wpb}}$ is not strongly collusion proof for the same reason as the first-price auction (\autoref{thm:FPAValueEquilibriumSummary}).

Finally, $\sigma^\onCG_{\mathsf{wpb}}$ is not off-chain influence proof.
The miner only needs to communicate to the users to steer them towards $\sigma^\onCG_{\mathsf{pp}}$, which yields a larger expected revenue. We summarize the properties in the proposition below.

\begin{proposition} \label{thm:BOMBFPASummary}
    The strategy profile $\sigma^\onCG_{\mathsf{wpb}} = (s_{\mi}^{\mathsf{comp}}(\emptyset), s_{\usr, \mathsf{val}}^{\onCG, \TypeDistr, \mathsf{r}}, \allowbreak \dots, \allowbreak s_{\usr, \mathsf{val}}^{\onCG, \TypeDistr, \mathsf{r}})$ induced by the BoMB auction is on-chain miner simple but does not satisfy on-chain user simplicity, strong collusion proofness and off-chain influence proofness.
\end{proposition}

\begin{remark}
    Even though $\sigma^\onCG_{\mathsf{wpb}}$ is not off-chain influence proof, the miner's optimal strategy in the off-chain game is to steer users away from $\sigma^\onCG_{\mathsf{wpb}}$ towards $\sigma^{\onCG}_{\mathsf{pp}}$, which is also an entirely on-chain equilibrium.
    $\sigma^\onCG_{\mathsf{wpb}}$ is therefore, not on-chain-miner-optimal.
    The miner's optimal off-chain strategy does not involve payments and only assists in selecting the optimal amongst all equilibria in the on-chain game.
    Thus, the ``consequence'' of having a off-chain deviation from the miner is not as bad as the ``consequence'' in EIP-1559, where the miner runs a separate auction altogether.
    On-chain optimality helps tease apart this distinction of coordinating equilibrium through off-chain messages and running an entire auction off-chain.
    
\end{remark}

\subsection{The Squared Revenue Second-Price Auction, Revisited} \label{sec:SR2PA}

We revisit the squared revenue second-price auction (SR2PA) described in \autoref{sec:ModelExample}.
As a quick recap, (assuming the distribution $\TypeDistr$ is supported on $[0, 1]$), the SR2PA implements the plaintext second-price auction, allowing the miner to set the reserve.
However, instead of transferring the entire payment made by the included user to the miner, the SR2PA only transfers the square of the payment (which is smaller than the payment since users have values and therefore bids supported in $[0, 1]$) and burns the rest.

We consider two strategy profiles. The first where the miner sets the monopoly reserve and users bid truthfully and the second where the miner fabricates a bid just below the highest bid, reducing the on-chain allocation and payments rule into a first-price auction (winner-pays-bid mechanism with $k = 1$).

\subsubsection{The Second-Price equilibrium} \label{sec:SR2PA2PA}

We analyze the properties satisfied by the strategy profile $\sigma^\onCG_{\mathsf{2PA}} = (s_{\mi}^{\mathsf{comp}}(a_\mi), v_1, \dots, v_n)$, where $a_\mi$ is the monopoly reserve of the distribution $\TypeDistr$.

Conditioned on the miner playing ${\mathsf{comp}}(a_\mi)$, users respond to a second-price auction with a reserve on-chain and thus, truthful bidding is a user BNE and the on-chain mechanism induced is also DSIC, and thus, on-chain user simple.
Indeed, it is not on-chain miner simple since the miner can increase her revenue by injecting a fake bid just below the largest bid.
It is also not collusion proof due to the simple deviation described next.
For the monopoly reserve $a_\mi$ and a user $i$ with value $v_i \in [a_\mi - a_\mi^2, a_\mi]$, the miner advices the block-building process to set a reserve zero. and $i$ bids $a_\mi$.
When a different user has a value larger than $a_\mi$, the allocation and payments are not different between the miner setting a reserve $a_\mi$ or faking a reserve $a_\mi$ by having user $i$ bid $a_\mi$.
However when all users have a value less than $a_\mi$, instead of not including any user and obtaining a zero joint utility, the user $i$ gets included to bag a utility $v_i - a_\mi \geq a_\mi - a_\mi^2 - a_\mi = -a_\mi^2$.
The miner gets a revenue $a_\mi^2$, and thus, the joint utility of the cartel is at least zero.

We argued that $\sigma^\onCG_{\mathsf{2PA}}$ is not off-chain influence proof in \autoref{sec:ModelExample}.
At a high level, the miner runs a second-price auction off-chain and includes the winner in the block (censoring all other bids).
The winner only has to bid and pay zero on-chain.
Thus, even though the total expected payments made by the users remains the same, the miner increases her revenue by saving on the money that gets burnt on-chain.

\begin{proposition} \label{thm:SR2PA2PASummary}
    Given a distribution supported over $[0, 1]$ with monopoly reserve $a_\mi$, the user BNE $\sigma^\onCG_{\mathsf{2PA}} = (s_{\mi}^{\mathsf{comp}}(a_\mi), v_1, \dots, v_n)$ in the on-chain game $\onCG$ is on-chain user simple, but not on-chain miner simple or collusion proof. The corresponding user BNE $\sigma^\offCG_{\mathsf{2PA}} = (\Moff^{\mathsf{triv}, a_\mi}, \bot_\mathsf{off}(v_1), \allowbreak \dots, \allowbreak \bot_\mathsf{off}(v_n))$ in the off-chain game is not off-chain influence proof.
\end{proposition}

\subsubsection{The First-Price Equilibrium} \label{sec:SR2PA1PA}

Consider the miner setting the reserve to be just below the maximum bid $b_{\max}$ (for convenience, we pretend the miner sets the reserve to be equal to the maximum bid).
Thus, the largest bidders always gets included and has to pay his bid.
Thus, the on-chain game is reduced to the winner-pays-bid mechanism for $k = 1$ with reserve zero.
For the optimal strategy $s_\usr^{\onCG, \TypeDistr, 0}$ for users in the winner-pays-bid mechanism (\autoref{sec:FPA}), we consider the user BNE $\sigma^\onCG_{\mathsf{wpb}} = ((b_{\max}, H, \emptyset), s_\usr^{\onCG, \TypeDistr, 0}, \allowbreak \dots, \allowbreak s_\usr^{\onCG, \TypeDistr, 0})$.

$\sigma^\onCG_{\mathsf{wpb}}$ is neither on-chain user simple, nor is it on-chain miner simple (since the miner is not playing a compliant strategy). It is not collusion proof and off-chain influence proof for reasons identical to the second-price equilibrium $\sigma^\onCG_{\mathsf{2PA}}$.

\begin{proposition} \label{thm:SR2PAFPASummary}
    Given a distribution supported over $[0, 1]$, the user BNE $\sigma^\onCG_{\mathsf{wpb}} = ((b_{\max}, H, \emptyset), s_\usr^{\onCG, \TypeDistr, 0}, \allowbreak \dots, \allowbreak s_\usr^{\onCG, \TypeDistr, 0})$ in the on-chain game $\onCG$ and the corresponding strategy profile $\sigma^\offCG_{\mathsf{wpb}} = (\Moff^{\mathsf{triv}, (b_{\max}, H, \emptyset)}, \bot_\mathsf{off}(s_\usr^{\onCG, \TypeDistr, 0}), \allowbreak \dots, \allowbreak \bot_\mathsf{off}(s_\usr^{\onCG, \TypeDistr, 0}))$ satisfies none of the four properties: on-chain user simplicity, on-chain miner simplicity, collusion proofness and off-chain influence proofness.
\end{proposition}

\section{The Deferred-Revelation Auction (DRA)}
\label{sec:DRA}

Thus far, we have focused on just one idealized type of cryptographic auction, namely, the miner-gatekeeper model.
In the miner-gatekeeper model, we assume that the commitments that are forwarded to the miner are always revealed at a later stage.
In other words, users (or for that matter, the miner for fabricated bids) cannot be strategic by concealing their bids after submitting a commitment to the block-building process.
Such as assumption is compatible with a few cryptographic primitives such as secure multi-party computations \citep{ShiCW23} or verifiable delay functions.
However, such an assumption might be unreasonable when using more lightweight cryptographic machinery such as commitment schemes via digital signatures.
If cryptographic commitment schemes are used, a concealed commitment cannot be recovered by the miner.
In such a setting, \citet{FerreiraW20} introduce the \emph{deferred-revelation auction} (henceforth, DRA), which was further refined by \citet{ChitraFK23} for blockchain environments.

The DRA happens over two stages. First, users submit encrypted bids in the commitment phase. Second, users can choose to either decrypt (``reveal'') their bids, or leave their bids concealed in the decryption phase.
The auction is executed on all bids that were revealed in the decryption phase.
To disincentivize bids from being withdrawn midway, all concealed bids are penalized with a non-trivial fine (which in practice would be implemented via a deposit during the commitment phase).

Note that the above cryptographic model is neither the plaintext nor the miner-gatekeeper model,
e.g., all bids are assumed to be revealed in the miner-gatekeeper model, while users (and the miner for fabricated bids) can choose to strategically forgo revealing bids.
We formally describe the \emph{deferred-revelation model} and capture the differences from the miner-gatekeeper model below.

\subsection{Deferred-Revelation Model}

We now describe the deferred-revelation model, which captures the framework proposed by \cite{ChitraFK23} (also closely related to the model of \cite{FerreiraW20}).
The primary difference between the deferred-revelation model and the miner-gatekeeper model is the bidding process, which happens across two stages---the commitment phase and the decryption phase.
In the commitment phase, encrypted bids are submitted from a bid space $\mathsf{Bid}_{\usr}$.
During the decryption phase, bid are appended with a signal from the revelation space $\{\mathsf{reveal}, \mathsf{conceal}\}$ which corresponds to revelation and concealment of their bids respectively.

The action space and information sets available to the miner is tailored to semantically match the dynamic nature of the deferred-revelation model.
The miner plays an advice $a_\mi$ from an advice space $\mathsf{Adv}_\mi$ during the commitment phase.
Along with deciding the advice $a_\mi$, the miner also decides the set of bids $J = (\hat{\beta}_j)_{1 \leq j \leq |J|}$ she wants to fabricate, and the set of bids $I$ that she wants to submit to the block-building process $\bBuild$ during the commitment phase.\footnote{
    Note that the miner chooses the set of encrypted bids to forward in the commitment phrase, but does not choose a set of `reveals' to forward in the revelation phase.
    This is because
    we want to avoid the miner coercing users by threatening to not forward their `reveals' in the decryption phase. If their bids are not revealed, users would have to pay a penalty $P_{\mathsf{conceal}}$. Thus, users are better off paying any ransom less than $P_{\mathsf{conceal}}$ demanded by the miner if at all they decide to bid during the commitment phase. This might nudge users to either not submit bids or to be strategic and bid something other than their value keeping the ransom in mind. We avoid this issue and stick to the framework proposed by \citet{ChitraFK23} by allowing the miner to censor bids only in the commitment phase. \label{fn:Ransom}
}
While the miner is oblivious to the contents of the bids during the commitment phase (the bids are encrypted), she can see the contents of all revealed bids during the decryption phase.
Thus, like in the miner-gatekeeper model, the set $J$ can depend only on the identifiers of the submitted bids and not the set $H$ of bids themselves.
However, deciding which subset of fabricated bids $J_{\mathsf{reveal}}\subseteq J$ to reveal can depend on the contents of all revealed bids $H_{\mathsf{reveal}}$.

Formally, the miner's strategy $s^\onCG_{\mi, \mathsf{commit}} \in S^\onCG_{\mi, \mathsf{commit}}$ in the commitment phase maps the set of all users who have submitted a bid to three pieces of information: an advice $a_\mi$, a subset $I \subseteq H$ of user-submitted bids to pass on to $\bBuild$, and a set of fabricated bids $J$.
The miner's strategy $s^\onCG_{\mi, \mathsf{decrypt}} \in S^\onCG_{\mi, \mathsf{decrypt}}$ in the decryption phase takes as input the (contents of) set of revealed bids $H_{\mathsf{reveal}}$ and a set of fabricated bids $J$ and outputs a subset $J_{\mathsf{reveal}} \subseteq J$ to reveal (i.e, play $\mathsf{reveal}$ in the decryption phase for all bids in $J_{\mathsf{reveal}}$ and $\mathsf{conceal}$ for all bids in $J_{\mathsf{conceal}} = J \setminus J_{\mathsf{reveal}}$).

Finally, we have to constrain the block-building process to be invariant to the set of bids concealed during the decryption phase.
For a set of revealed bids $\beta_{\mathsf{reveal}}$ and a set of concealed bids $\beta_{\mathsf{conceal}}$ passed on to $\bBuild$, the allocation rule is identical to passing on just the set of revealed bids $\beta_{\mathsf{reveal}}$, i.e,
\[X(a_\mi; \beta_{\mathsf{reveal}}; \beta_{\mathsf{conceal}}) = X(a_\mi; \beta_{\mathsf{reveal}}; \emptyset).\]
(And in particular, in both auction outcomes, the users who conceal their bids are not allocated.)
Similarly, the payments made by all users that revealed their bids should not depend on the set of concealed bids. Formally, for a user $i$ who revealed his bid,
\[P(a_\mi; \beta_{\mathsf{reveal}}; \beta_{\mathsf{conceal}}) = P(a_\mi; \beta_{\mathsf{reveal}}; \emptyset).\]
A user $i$ who chooses to conceal their bid pay a penalty $P_i(a_\mi; \beta_{\mathsf{reveal}}; \beta_{\mathsf{conceal}}) = P_{\mathsf{conceal}} $ for some constant $ P_{\mathsf{conceal}} \geq 0$.

Simplicity definitions carry-over naturally to the deferred-revelation model.
For instance, on-chain miner simplicity corresponds to the miner playing a compliant strategy, i.e, playing an advice, not censoring any bids and not fabricating bids in the commitment phase (and thus, the miner has no non-trivial action in the decryption phase), while on-chain user simplicity corresponds to users best responding to a compliant miner strategy by bidding truthfully in the commitment phase and revealing their bids in the decryption phase.
Off-chain influence proofness and strong collusion proofness also have similar extensions.

\subsection{Second-Price Auction in the Deferred-Revelation Model}

\citet{FerreiraW20} and \citet{ChitraFK23} consider the second-price auction with reserve in the deferred-revelation model.
They assume that the block-building process knows the monopoly reserve $\mathsf{r}$ of the distribution $\TypeDistr$ of values and thus, do not receive any advice from the miner.
In particular, unlike other auctions throughout our paper, the block-building process of this auction is prior-dependent.

\begin{definition}[\citealp{FerreiraW20, ChitraFK23}, Deferred-Revelation (Second-Price) Auction] \label{def:DRA}
    The block building process $\bBuild$ of the deferred revelation auction (DRA) is as follows:
    \begin{itemize}
        \item $\mathsf{Adv}_\mi = \{\emptyset\}$. $\bBuild$ takes no inputs from the miner. $\bBuild$ is exogenously informed of the monopoly reserve $\mathsf{r}$ of the distribution $\TypeDistr$.
        \item $\mathsf{Bid}_{\usr} = \R_{\ge 0}$. The commitment phase is a standard sealed-bid auction.
        \item Let $\beta_{\mathsf{reveal}}$ be the set of all revealed bids in the decryption phase.
        The largest bid in $\beta_{\mathsf{reveal}}$ conditioned on being larger than the monopoly $\mathsf{r}$ is included.
        Formally, for a set of revealed bids $\beta_{\mathsf{reveal}}$ and a set of concealed bids $\beta_{\mathsf{conceal}}$,
        \begin{equation}
            \notag
            \begin{split}
                X_i(\emptyset; \beta_{\mathsf{reveal}}, \beta_{\mathsf{conceal}}) =
                \begin{cases}
                    1 & \text{if } \beta_i \in \beta_{\mathsf{reveal}}, \beta_i \geq \mathsf{r} \text{ and } \beta_i = \arg \max_{j} \{\beta_j | \beta_j \in \beta_{\mathsf{reveal}}\} \\
                    0 & \text{otherwise}
                \end{cases}
            \end{split}
        \end{equation}
        \item If any user is included, the users pays the second highest bid $\beta^{(2)}_\mathsf{reveal}$ amongst all revealed bids, or the reserve $\mathsf{r}$, whichever is higher. All concealed bids on the other hand, pay a penalty $P_\mathsf{conceal}$. In other words, for a set of bids $b$ of which $\beta_{\mathsf{reveal}}$ are revealed and $\beta_{\mathsf{conceal}}$ are concealed,
        \begin{equation}
            \notag
            \begin{split}
                P_i(\emptyset; \beta_{\mathsf{reveal}}; \beta_{\mathsf{conceal}}) =
                \begin{cases}
                    \max \{\beta^{(2)}_{\mathsf{reveal}}, \mathsf{r}\} & \text{if } \beta_i \in \beta_{\mathsf{reveal}} \text{ and } \beta_i = \arg \max_{j} \{\beta_j | b_j \in b_{\mathsf{reveal}}\} \cup \{\mathsf{r}\} \\
                    P_\mathsf{conceal} & \text{if } b_i \in \beta_\mathsf{conceal} \\
                    0 & \text{otherwise}
                \end{cases}                
            \end{split}
        \end{equation}
        \item All collected payments that are not penalties are transferred to the miner.
        \[
        \mathsf{Rev}(\emptyset; \beta_{\mathsf{reveal}}; \beta_\mathsf{conceal}) = \sum_i P_i(\emptyset; \beta_{\mathsf{reveal}}; \beta_\mathsf{conceal}) \cdot \1{\beta_i \in \beta_\mathsf{reveal}}
        \]
    \end{itemize}
\end{definition}

We analyze the simplicity properties satisfied by the DRA.
Conditioned on the miner playing a compliant strategy, the deferred-revelation auction is only a different implementation of the second-price auction for all users who reveal their bids.
Thus, it is DSIC for users to bid truthfully and reveal their bids, and is on-chain user simple.

On-chain miner simplicity on the other hand is not as straightforward.
The most natural argument would claim that the miner best responds by revealing all fabricated bids in the decryption phase irrespective of the bids submitted by the users, followed by claiming that the miner optimizes her revenue by injecting no fake bids in the commitment phase conditioned on always revealing all her bids (the latter claim is similar to on-chain miner simplicity of C$(k+1)$PA in the miner-gatekeeper model).
However, the former of the two claims is not true, and proving on-chain miner simplicity requires a bit more subtlety.\footnote{
    In fact, it is not DSIC for users to bid truthfully and reveal irrespective of the miner's strategy (similar to \autoref{fn:OneBid}).
    For example, the miner can inject a fake bid each for all non-negative reals and reveal only the fabricated bids below the largest revealed user bid.
    Thus, the user with the highest bid pays (just below) his bid and the auction reduces to a first-price auction.
    However, the above strategy is not profitable for the miner since, if $P_{\mathsf{conceal}} > 0$ and many fake bids are placed, she pays more in penalties for concealing bids than she gains in revenue.
    \label{fn:DRANotDSIC}}
\citet{ChitraFK23} show that for any $\alpha$-regular distribution $\TypeDistr$ for any $\alpha > 0$ (formally,the virtual value function $\phi$ satisfies $\phi(v_i) - \phi(\Tilde{v}_i) \geq \alpha \, (v_i - \Tilde{v}_i)$ for all $v_i \geq \Tilde{v}_i$) there exists a large enough fine $P_{\mathsf{conceal}}$ such that the strategy profile $\sigma^\onCG = (s_{\mi}^{\mathsf{comp}}(\emptyset), (v_1, \mathsf{reveal}), \dots, (v_n, \mathsf{reveal}))$ in the DRA is on-chain miner simple.\footnote{
    The penalty $P_{\mathsf{conceal}}$ prescribed by \citet{ChitraFK23} is a function of the monopoly reserve $\mathsf{r}$. Thus, it is important that the block-building process knows $\mathsf{r}$ and is not advised by the miner. If instead the miner could set the penalty, she would conveniently set it to zero, thus incurring no penalty for concealing bids, and deviate according to the strategy sketched in \autoref{fn:DRANotDSIC}.
    This is the reason we assume that the deferred-revelation auction is not prior-independent. } 

\begin{theorem}[\citealp{ChitraFK23}] \label{thm:OriginalDRA}
    For all $\alpha$-regular distributions $\TypeDistr$ ($\alpha > 0$), there exists a penalty $P_{\mathsf{conceal}}$ such that the miner's optimal strategy in the DRA is playing the compliant strategy $(\emptyset, H, \emptyset)$ and conditioned on the miner playing $(\emptyset, H, \emptyset)$, the DRA is DSIC for users, i.e, user $i$'s optimizes his utility by bidding $(v_i, \mathsf{reveal})$ regardless of the actions played by the other users.
\end{theorem}

Based on the above discussions, we conclude \autoref{thm:DRASummary}.
The proofs for on-chain user simplicity, off-chain influence proofness, and strong collusion non-proofness hold for exactly the same reason as in the C$(k+1)$PA (\autoref{thm:C2PAPosSummary} and \autoref{thm:C2PACol}). \autoref{thm:OriginalDRA} shows the result for on-chain miner simplicity.

\begin{proposition} \label{thm:DRASummary} 
    For any $\alpha$-regular distributions $\TypeDistr$ with $\alpha > 0$, consider the DRA with a high enough fine $P_{\mathsf{conceal}}$ as given by \autoref{thm:OriginalDRA}.
    Then, 
    the strategy profile $\sigma^\onCG = (s_{\mi}^{\mathsf{comp}}(\emptyset), (v_1, \mathsf{reveal}),\allowbreak \dots,\allowbreak (v_n, \mathsf{reveal}))$ is on-chain user and miner simple and off-chain influence proof, but not strong collusion proof.
\end{proposition}

\section{Omitted Proofs} \label{sec:ExcProofs}

\subsection{On-Chain Simplicity and Collusion Proofness in EIP-1559} \label{sec:ExcEIP1559}

We now prove \autoref{thm:EIP1559Rou}, directly adapting the arguments from \cite{Roughgarden20} to our setting.

As mentioned in \autoref{sec:EIP1559}, the first step is to verify that $\sigma^\onCG$ is indeed a user BNE for all distributions $\TypeDistr$ in both the on-chain game $\onCG$ and the off-chain game $\offCG$. 
\begin{lemma} \label{thm:EIPUBNE}
    $\sigma^\onCG = (s_{\mi}^{\mathsf{comp}}(\emptyset), v_1, \dots, v_n)$ is a user BNE in both, the on-chain game $\onCG$ and off-chain game $\offCG$ induced by EIP-1559 for any distribution $\TypeDistr$.
\end{lemma}
\begin{proof}
    We verify that $\sigma^\onCG$ is a user Bayes-Nash equilibrium in $\onCG$. The proof for $\offCG$ is analogous.

    Observe that the allocation and payments made by user $i$ is independent of the other bids submitted to the mechanism.
    Consider a value $v_i$ larger than the price $p$.
    Placing any bid larger than the price $p$ will lead to getting included and making a payment equal to $p$.
    The user bags a utility equal to $v_i - p \geq 0$, which is larger than the zero utility from bidding below the price and getting excluded.
    Similarly, for a value $v_i$ smaller than the price $p$, bidding below the price leads to a zero utility, while bidding above the price leads to a utility $v_i - p < 0$.
    Bidding $v_i$ optimizes user $i$'s utility in both cases, as required.
\end{proof}

\noindent Next, we show that $\sigma^\onCG$ is on-chain simple.

\begin{lemma} \label{thm:EIPOnChain}
    Consider the on-chain game $\onCG$ induced by EIP-1559 for some distribution $\TypeDistr$. The strategy profile $\sigma^\onCG$ is on-chain simple.
\end{lemma}
\begin{proof}
    We have verified that $\sigma^\onCG$ is a user Bayes-Nash equilibrium in \autoref{thm:EIPUBNE}. On-chain user simplicity follows, since all users bid truthfully in $\sigma^\onCG$. The miner will not be able to receive any revenue on-chain irrespective of the set of fake bids or censored bids.
    The miner could potentially end up burning her own money if she fabricates bids larger than the price posted by the mechanism.
    Thus, the miner optimizes her revenue by playing $s_{\mi}^{\mathsf{comp}}(\emptyset)$.
\end{proof}

\noindent Finally, we verify that $\sigma^\onCG$ is collusion proof.

\begin{lemma} \label{thm:EIPCollusion}
    The strategy profile $\sigma^\onCG  = (s_{\mi}^{\mathsf{comp}}(\emptyset), v_1, \dots, v_n)$ is collusion proof for the on-chain game $\onCG$ induced by EIP-1559 for all distributions $\TypeDistr$.
\end{lemma}
\begin{proof}
    Suppose that the miner colludes with a user $i$ with value $v_i$.
    Observe that user $i$'s inclusion in the block and utility is invariant to any fake bids injected by the miner or any bids censored by the miner (as long as the miner does not censor $i$'s bid).
    As argued in the proof of on-chain miner simplicity (\autoref{thm:EIPOnChain}), the miner will not be able to increase her revenue by injecting fake bids, nor would it increase user $i$'s utility.
    Thus, without loss of generality, assume that the miner does not fabricate bids or censor bids of all users apart from $i$.

    Irrespective of whether the miner includes or censors $i$'s bid, the miner always gets a zero revenue. Thus, the joint utility of the coalition equals user $i$'s utility.
    If the user has a value $v_i$ larger than the price $p$, $i$'s utility is optimized by bidding truthfully and getting included.
    Similarly, if the user has a value smaller than $p$, $i$'s utility is optimized by bidding truthfully and getting excluded.

    Playing $s_{\mi}^{\mathsf{comp}}(\emptyset)$ and $v_i$ optimizes the joint utility of the coalition, implying $\sigma^\onCG$ is collusion proof for $\onCG$ for all distributions $\TypeDistr$.
\end{proof}

\autoref{thm:EIPUBNE}, \autoref{thm:EIPOnChain} and \autoref{thm:EIPCollusion} conclude the proof of \autoref{thm:EIP1559Rou}.

\subsection{The C$(k+1)$PA is DSIC} \label{sec:C2PADSIC}

\begin{lemma}
    For a given reserve $a_\mi$, bidding truthfully is utility optimizing for users irrespective of the strategies employed by the other users.
\end{lemma}
\begin{proof}
    Fix the bids $b_{-i}$ of users other than $i$.
    For a value $v_i$ for user $i$ larger than the reserve $a_\mi$, and amongst the $k$ largest bids, we will prove that $i$ optimizes his utility by bidding $v_i$.
    The proof of the other cases are analogous.

    By being amongst the $k$ largest bids and the reserve, user $i$ will get included by bidding $v_i$.
    He would need to make payment $P_i(v_i, b_{-i}) = \max \{a_\mi, b^{(k+1)}\}$ and therefore, bag a utility $v_i - \max \{a_\mi, b^{(k+1)}\} \ge 0$.

    By placing a bid $b_i$ larger than $v_i$, user $i$ continues to get allocated and still have to pay $P_i(b_i, b_{-i}) = \max \{a_\mi, b^{(k+1)}\}$, resulting in the same utility as bidding $v_i$.
    Under-bidding by placing a bid $b_i$ less than $v_i$ can result in two outcomes --- where $i$ continues winning the good and where $i$ does not win the good.
    If $i$ wins the good, he still has to make a payment equal to $P_i(b_i, b_{-i}) = \max \{a_\mi, b^{(k+1)}\}$.
    If $i$ does not win the good anymore, he gets a utility equal to zero, which is worse than the utility from truthful bidding.

    Thus, each user optimizes his utility by bidding truthfully.
\end{proof}

\end{document}